\RequirePackage{luatex85}

\documentclass[11pt,a4paper]{article}
\usepackage{iftex}

\ifPDFTeX
\usepackage[utf8]{inputenc}
\usepackage[T1]{fontenc}
\fi

\usepackage{amsmath}
\usepackage{amssymb}
\usepackage{amsthm}
\usepackage{mathtools}
\usepackage{thmtools}

\ifPDFTeX
\usepackage{libertine}
\usepackage[libertine]{newtxmath} 
\usepackage[scaled=0.96]{zi4} 
\else
\usepackage{fontspec}
\usepackage{unicode-math}
\fi

\usepackage[DIV=11]{typearea} 

\usepackage{microtype}

\usepackage[procnumbered,ruled,vlined,linesnumbered,algo2e]{algorithm2e}

\usepackage{xcolor}

\usepackage[
	style=alphabetic,
	backref=true,
	doi=true,
	url=true,
	maxcitenames=3,
	mincitenames=3,
	maxbibnames=10,
	minbibnames=10,
	backend=biber
]{biblatex}

\usepackage[ocgcolorlinks]{hyperref} 
\usepackage{cleveref}

\allowdisplaybreaks[1]

\makeatletter
\g@addto@macro\bfseries{\boldmath}
\makeatother

\makeatletter
\g@addto@macro\mdseries{\unboldmath}
\g@addto@macro\normalfont{\unboldmath}
\g@addto@macro\rmfamily{\unboldmath}
\g@addto@macro\upshape{\unboldmath}
\makeatother

\DeclareCiteCommand{\citem}
    {}
    {\mkbibbrackets{\bibhyperref{\usebibmacro{postnote}}}}
    {\multicitedelim}
    {}

\renewcommand*{\labelalphaothers}{\textsuperscript{+}}

\renewcommand*{\multicitedelim}{\addcomma\space}

\AtEveryCitekey{\clearfield{note}}

\makeatletter
\AtBeginDocument{%
    \newlength{\temp@x}%
    \newlength{\temp@y}%
    \newlength{\temp@w}%
    \newlength{\temp@h}%
    \def\my@coords#1#2#3#4{%
      \setlength{\temp@x}{#1}%
      \setlength{\temp@y}{#2}%
      \setlength{\temp@w}{#3}%
      \setlength{\temp@h}{#4}%
      \adjustlengths{}%
      \my@pdfliteral{\strip@pt\temp@x\space\strip@pt\temp@y\space\strip@pt\temp@w\space\strip@pt\temp@h\space re}}%
    \ifpdf
      \typeout{In PDF mode}%
      \def\my@pdfliteral#1{\pdfliteral page{#1}}
      \def\adjustlengths{}%
    \fi
    \ifxetex
      \def\my@pdfliteral #1{}
      \def\adjustlengths{\setlength{\temp@h}{-\temp@h}\addtolength{\temp@y}{1in}\addtolength{\temp@x}{-1in}}%
    \fi%
    \def\Hy@colorlink#1{%
      \begingroup
        \ifHy@ocgcolorlinks
          \def\Hy@ocgcolor{#1}%
          \my@pdfliteral{q}%
          \my@pdfliteral{7 Tr}
        \else
          \HyColor@UseColor#1%
        \fi
    }%
    \def\Hy@endcolorlink{%
      \ifHy@ocgcolorlinks%
        \my@pdfliteral{/OC/OCPrint BDC}%
        \my@coords{0pt}{0pt}{\pdfpagewidth}{\pdfpageheight}%
        \my@pdfliteral{F}
        \my@pdfliteral{EMC/OC/OCView BDC}%
        \begingroup%
          \expandafter\HyColor@UseColor\Hy@ocgcolor%
          \my@coords{0pt}{0pt}{\pdfpagewidth}{\pdfpageheight}%
          \my@pdfliteral{F}
        \endgroup%
        \my@pdfliteral{EMC}%
        \my@pdfliteral{0 Tr}
        \my@pdfliteral{Q}%
      \fi
      \endgroup
    }%
}
\makeatother

\colorlet{DarkRed}{red!50!black}
\colorlet{DarkGreen}{green!50!black}
\colorlet{DarkBlue}{blue!50!black}

\hypersetup{
	linkcolor = DarkRed,
	citecolor = DarkGreen,
	urlcolor = DarkBlue,
	bookmarksnumbered = true,
	linktocpage = true
}

\declaretheorem[numberwithin=section]{theorem}
\declaretheorem[numberlike=theorem]{lemma}

\declaretheorem[numberlike=theorem]{corollary}
\declaretheorem[numberlike=theorem]{definition}
\declaretheorem[numberlike=theorem]{claim}
\declaretheorem[numberlike=theorem, style=remark]{remark}
\declaretheorem[numberlike=theorem]{observation}
\crefname{algorithm}{Procedure}{Procedures}
\Crefname{algorithm}{Procedure}{Procedures}

\DontPrintSemicolon
\SetKwProg{Procedure}{Procedure}{}{}
\SetProcNameSty{textsc}
\SetFuncSty{textsc}
\SetKw{KwAnd}{and}
\SetKw{KwDo}{do}

\SetCommentSty{mycommfont}

\SetKwFunction{ProcessQueue}{ProcessQueue}
\SetKwFunction{Delete}{Delete}
\SetKwFunction{Insert}{Insert}
\SetKwFunction{FixFreeVertex}{FixFreeVertex}
\SetKwFunction{ResetMatching}{ResetMatching}
\SetKwFunction{IncrementPhi}{CheckForRise}
\SetKwFunction{DecrementPhi}{Decrement-$\phi$}
\SetKwFunction{GenericRandomSettle}{RandomSettle}
\SetKwFunction{Rise}{Rise}
\SetKwFunction{Fall}{Fall}
\SetKwFunction{Flush}{Flush}

\newcommand{\ignore}[1]{}
\newcommand{\etal}{et al.\xspace}

\newcommand{\del}{\delta}
\newcommand{\eps}{\epsilon}

\newcommand{\floor}[1]{\left\lfloor#1\right\rfloor}
\newcommand{\mmax}[2]{\max \left\{ {#1} , {#2} \right\} }

\newcommand{\PP}{\mathbf{Pr}}
\newcommand{\EX}{{\mathbf{E}}}
\newcommand{\cond}{ \ | \ }

\newcommand{\bad}{B}

\newcommand{\degree}{\operatorname{\normalfont\textsc{degree}}}
\newcommand{\mate}{\operatorname{\normalfont\textsc{mate}}}
\newcommand{\oldmate}{\operatorname{\normalfont\textsc{old-mate}}}
\newcommand{\LL}{{\normalfont\textsc{level}}}
\newcommand{\OO}{\mathcal{O}}
\newcommand{\EE}{\mathcal{E}}
\newcommand{\below}[1]{N_{< #1}}
\newcommand{\equal}[1]{N_{= #1}}

\newcommand{\prise}{p^{\mathrm{rise}}}
\newcommand{\preset}{p^{\mathrm{reset}}}
\newcommand{\resetprob}[1]{1/4^{#1+3}}

\newcommand{\xreset}{X^{\text{reset}}}
\newcommand{\xrise}{X^{\text{rise}}}

\newcommand{\efall}{\EX^{\mathrm{fall}}}
\newcommand{\egenericsettle}{\EX^{\mathrm{settle}}}
\newcommand{\eprocessfree}{\EX^{\mathrm{free}}}
\newcommand{\emax}{\EX^{\mathrm{max}}}

\newcommand{\tstar}{t^{*}}

\newcommand{\xsettle}{X^{\text{settle}}}
\newcommand{\xcgrs}{X^{\text{crs}}}

\newcommand{\spast}{S_{\text{past}}}

\newcommand{\hook}{\operatorname{hook}}

\title{A Deamortization Approach for Dynamic Spanner and Dynamic Maximal Matching\thanks{A preliminary version of this article was presented at the 30th Annual ACM-SIAM Symposium on Discrete Algorithms (SODA 2019).}}

\author{
Aaron Bernstein\thanks{Rutgers University, Department of Computer Science, USA. Work done in part while at TU Berlin and while visiting University of Vienna.} \and
Sebastian Forster\thanks{University of Salzburg, Department of Computer Sciences, Austria. Work done in part while at University of Vienna. This author previously published under the name Sebastian Krinninger.} \and
Monika Henzinger\thanks{University of Vienna, Faculty of Computer Science, Austria.}
}

\date{}

\hypersetup{
	pdftitle = {A Deamortization Approach for Dynamic Spanner and Dynamic Maximal Matching},
	pdfauthor = {Aaron Bernstein, Sebastian Forster, Monika Henzinger}
}

\begin{document}
\maketitle
\begin{abstract}
Many dynamic graph algorithms have an \emph{amortized} update time,
rather than a stronger \emph{worst-case} guarantee.
But amortized data structures are not suitable for real-time systems, where each individual operation has to be executed quickly. 
For this reason, there exist many recent randomized results that aim to provide a guarantee stronger than amortized
expected. 
The strongest possible guarantee for a randomized algorithm is that it is always correct (Las Vegas), and has \emph{high-probability worst-case} update time, which gives a bound on the time for each individual operation that holds with high probability.  

In this paper we present the first polylogarithmic high-probability worst-case time bounds for the dynamic spanner and the dynamic maximal matching problem. 
\begin{enumerate}
\item For dynamic spanner, the only known $o(n)$ worst-case bounds were
$O(n^{3/4})$ high-probability worst-case update time for maintaining a 3-spanner and $O(n^{5/9})$ for maintaining a 5-spanner. We give a $O(1)^k \log^3 (n)$  high-probability worst-case time bound for maintaining a $(2k-1)$-spanner, 
which yields the first worst-case polylog update time for all constant~$k$. (All the results above maintain the optimal tradeoff of stretch $2k-1$ and $\tilde{O}(n^{1+1/k})$ edges.) 
\item For dynamic \emph{maximal} matching, or dynamic $2$-approximate maximum matching, 
no algorithm with $o(n)$ worst-case time bound was known and we present an algorithm with $O(\log^5 (n))$ high-probability worst-case time; similar worst-case bounds existed only for maintaining a matching that was $(2+\eps)$-approximate, and hence not maximal.
\end{enumerate}

Our results are achieved using
a new approach for converting amortized guarantees to worst-case ones for randomized data structures
by going through a third type of guarantee, which is a middle ground between the two above: an algorithm is 
said to have \emph{worst-case expected} update time $\alpha$ if for \emph{every} update $\sigma$, the expected time to process $\sigma$ is at most 
$\alpha$.
Although stronger than amortized expected, the worst-case expected guarantee does
not resolve the fundamental problem of amortization: a worst-case expected update time of $O(1)$ still allows for the possibility that every 
$1/f(n)$ updates requires $\Theta(f(n))$ time to process, for arbitrarily high $ f(n) $.
In this paper we present a \emph{black-box} reduction that converts any data structure with worst-case expected 
update time into one with a high-probability worst-case update time: the query time remains the same, while 
the update time increases by a factor of $O(\log^2(n))$.

Thus we achieve our results in two steps:
(1) First we show how to convert existing dynamic graph algorithms with {\em amortized} expected polylogarithmic running times into algorithms with {\em worst-case expected} polylogarithmic running times.
(2) Then we use our black-box reduction to achieve the polylogarithmic high-probability worst-case time bound.
All our algorithms are Las-Vegas-type algorithms.

\end{abstract}

\section{Introduction}
 
A \emph{dynamic graph algorithm} is a data structure that maintains information in a graph that is being modified
by a sequence of edge insertion and deletion operations. 
For a variety of graph properties there exist dynamic graph algorithms for which amortized expected time bounds are known and the main challenge is to de-amortize and de-randomize these results. Our paper addresses the first challenge: de-amortizing dynamic data structures. 

An amortized algorithm guarantees a small average update time for a ``large enough'' sequence of operations: dividing the total time for $T$ operations by $T$ leads to the \emph{amortized time} per operation.
If the dynamic graph algorithm is randomized, then 
the \emph{expected total time} for a sequence of operations is analyzed,
giving a bound on the \emph{amortized expected time} per operation. 
But in real-time systems, where each individual operation has to be executed quickly, we need a stronger guarantee than amortized expected time for randomized algorithms. The strongest possible guarantee for a randomized algorithm is that it
is always correct (Las Vegas), and has \emph{high-probability worst-case} update time, 
which gives an upper bound on the time for \emph{every} individual
operation that holds with high probability.
(The probability that the time bound is not achieved should be polynomially small in the problem size.)
There are many recent results which provide randomized data structures
with worst-case guarantees 
(see e.g.~\cite{Sankowski04,KapronKM13,GibbKKT15,AbrahamDKKP16,BodwinK16,AbrahamCK17,NanongkaiSW17,CharikarS18,ArarCCSW18}), 
often via a complex ``deamortization'' of previous results.

In this paper we present the first algorithms with worst-case polylog update time for two classical problems in the dynamic setting: dynamic spanner, and dynamic maximal matching. In both cases, polylog \emph{amortized} results were already known, but the best worst-case results required polynomial update time. 

Both results are based on a new de-amortization approach for randomized dynamic graph algorithms. We bring attention to a third possible type of guarantee:  an algorithm is said to have \emph{worst-case expected} update time $\alpha$ if for \emph{every} update $\sigma$, the expected time to process $\sigma$ is at most~$\alpha$. On its own this guarantee does not resolve the fundamental problem of amortization, since a worst-case expected update time of $O(1)$ still allows for the possibility that every 
$1/f(n)$ updates requires $\Theta(f(n))$ time to process, for arbitrarily high $ f(n) $.
But by using some relatively simple probabilistic bootstrapping techniques, we show a \emph{black-box} reduction that converts any algorithm with a worst-case expected update time into one with a high-probability worst-case update time. 

This leads to the following deamortization approach: rather than directly aiming for high-probability worst-case, first aim for the weaker worst-case expected guarantee, and then apply the black-box reduction. Achieving such a worst-case expected bound can involve serious technical challenges, in part because one cannot rely on the standard charging arguments used in amortized analysis. We nonetheless show how to achieve such a guarantee for both dynamic spanner and dynamic maximal matching, which leads to our improved results for both problems.

\paragraph{Details of the New Reduction.}
We show a black-box conversion of an algorithm with worst-case expected update time
into one with worst-case high-probability update time: the worst-case query time remains the same,
while the update time increases by a $\log^2(n)$ factor.
Our reduction is quite general, but with our applications to dynamic graph algorithms in mind, we restrict ourselves to dynamic data structures that support only two types of operations: (1) \emph{update} operations, which manipulate the internal state of the data structure, but do not return any information, and (2) \emph{query} operations, which return information about the internal state of the data structure, but do not manipulate it. We say the data structure has \emph{update time $\alpha$} if the maximum update time of any type of update (e.g.~insertion or deletion) is $\alpha$.

\begin{theorem}
\label{thm:main}
Let $A$ be an algorithm that maintains a dynamic data structure $D$
with worst-case expected update time $\alpha$ for each update operation and
let $n$ be a parameter such that the maximum number of items stored in the data structure at any point in time is polynomial in $n$.
We assume that for any set of elements $S$ such that $|S|$ is polynomial in $n$, a new version of the data
structure $D$ containing exactly 
the elements of $S$ can be constructed in polynomial time. 
If this assumption holds, then there exists an algorithm $A'$ with the following properties:
\begin{enumerate}
\item \label{item:main-worst}
For any sequence of updates $ \sigma_1, \sigma_2, \ldots, $ $A'$ processes each update $\sigma_i$ in $O(\alpha \log^2(n))$ time
with high probability.
The amortized expected update time of $A'$ is $O(\alpha\log(n))$.
\item \label{item:main-correctness}
$A'$ maintains $\Theta(\log(n))$ data structures $D_1, D_2, ..., D_{\Theta(\log(n))}$,
as well as a pointer to some $D_i$ that is guaranteed to be correct at the current time. Query operations are answered with $D_i$.
\end{enumerate}
\end{theorem}

The theorem applies to any dynamic data structure, but we will apply it to dynamic graph algorithms. 
Due to its generality, however, we expect that the theorem will prove useful in other settings as well.
When applied to a dynamic graph algorithm, 
$n$ denotes the number of vertices, and
at most $n^2$ elements (the edges) are stored at any point in time.
Note that our assumption about polynomial preprocessing time for any polynomial-size set of elements $S$
is satisfied by the vast majority of data structures, and is in particular satisfied by all 
dynamic graph algorithms that we know of.

Observe that a high-probability worst-case update time bound of $O(\alpha \log^2(n))$ allows us to stop the algorithm whenever its update time exceeds the $O(\alpha \log^2(n))$ bound and in this way obtain an algorithm that is correct with high probability.

\begin{remark}
\label{remark:main}
By Item~\ref{item:main-correctness}, the converted algorithm $A'$ stores a slightly different
data structure than the original algorithm $A$, because it maintains $O(\log(n))$ copies $D_i$ of the
data structure in $A$. The data structure in $A'$ is equally powerful to that in $A$
because it can answer all the same queries in the same asymptotic time: $A'$ always has a pointer to some $D_i$ that is guaranteed to be fixed, so it can use $D_i$ to answer queries. The main difference is that the answers produced by $A'$ may have less ``continuity'' than those produced by $A$:
for example, in a dynamic maximal matching algorithm, if each query outputs the entire maximal matching, then a single update may change the pointer in $A'$
from some $D_i$ to some $D_j$, and $A'$ will then output a completely different maximal matching before and after the update. 
However, combining $A'$ with the very recent black-box reduction in \cite{SolomonS18} we can turn $A'$ into a ``continuous'' one at the cost of an extra $(1+\eps)$ factor in the approximation.
By applying this reduction with $ \epsilon' = \epsilon/2 $ we obtain a fully dynamic algorithm for maintaining a matching with an approximation factor of $ 2 (1+\eps/2) = (2+\eps)$ and a high-probability worst-case update time of $O(\log^6(n) + 1/\epsilon)$.
(See the end of Section~\ref{sec:matching} for more explanations.)

Note that this issue does not arise in our dynamic spanner algorithm, as the spanner is formed by the union of the spanners of all copies.
\end{remark}

\paragraph{First Result: Dynamic Spanner Maintenance.} 
Given a graph $G$, a spanner $H$ with stretch $\alpha$ is a subgraph of $G$ such that for any pair of vertices $(u,v)$,
the distance between $u$ and $v$ in $H$ is at most an $\alpha$ factor larger than their distance in $G$.
In the dynamic spanner problem the main goal is to maintain, for any given integer~$ k \geq 2 $, a spanner of stretch $ 2k - 1 $ with $ \tilde O (n^{1 + 1/k}) $ edges; we focus on these particular bounds because
spanners of stretch $ 2k - 1 $ and $ O (n^{1 + 1/k}) $ edges exist for every graph~\cite{Awerbuch85}, and this trade-off is presumed tight under Erd\H{o}s's girth conjecture.
The dynamic spanner problem was introduced by Ausiello, Franciosa, and Italiano~\cite{AusielloFI06}
and has been improved upon by~\cite{Elkin11,BaswanaKS12,BodwinK16}.
There currently exist near-optimal amortized expected bounds: a $ (2k - 1) $-spanners can be maintained with expected amortized update time $ O(1)^k $~\cite{BaswanaKS12} or  
time $ O (k^2 \log^2 (n))$~\cite{GoranciK18}.
The state-of-the-art for high-probability worst-case lags far behind: $O(n^{3/4})$ update time for maintaining a 3-spanner,
and $O(n^{5/9})$ for a 5-spanner \cite{BodwinK16}; no $o(n)$ worst-case update time was known for larger~$k$. All of these algorithms exhibit the stretch/space trade-off mentioned above, up to polylogarithmic factors in the size of the 
spanner\footnote{A standard assumption for the analysis of randomized dynamic graph algorithms is that the ``adversary'' who supplies the sequence of updates
is assumed to have fixed this sequence $\sigma_1, \sigma_2, ...$ \emph{before} the dynamic algorithm starts to operate,
and the random choices of the algorithm then define a distribution on the time to process each $\sigma_i$. 
This is called an \emph{oblivious} adversary.
Our dynamic algorithms for spanners and matching share this assumption, as does all work prior to the conference version of this paper. A randomized dynamic fractional matching algorithm that does not need this assumption was presented in~\cite{wajc2020rounding}}. 

We give the first dynamic spanner algorithm with polylog worst-case update time with high probability for any constant~$k$, which significantly improves upon the
result of \cite{BodwinK16} both in update time and in range of $k$. Our starting point is the earlier result
of Baswana, Khurana, and Sarkar~\cite{BaswanaKS12} that maintains a $(2k-1)$ spanner with $O(n^{1+1/k}\log^2(n))$ edges with update time $O(1)^k$. We show that while their algorithm is amortized expected, it can be modified to yield worst-case expected bounds: this requires a few changes to the algorithm, as well as significant changes to the analysis. We then apply the reduction in Theorem~\ref{thm:main}.

\begin{theorem}
\label{thm:spanner-expected}
There exists a fully dynamic (Las Vegas) algorithm for maintaining a $(2k-1)$ spanner with $O(n^{1+1/k} \log^6(n) \log{\log{(n)}})$ edges that 
has worst-case expected update time $O(1)^k \log(n)$. 
\end{theorem}  

\begin{theorem}
\label{thm:spanner-converted}
There exists a fully dynamic (Las Vegas) algorithm for maintaining a $(2k-1)$ spanner with $O(n^{1+1/k} \log^7(n) \log{\log{(n)}})$ edges that 
has high-probability worst-case update time $O(1)^k \log^3(n)$. 
\end{theorem}

The proof follows directly from Theorem ~\ref{thm:spanner-expected} and Theorem~\ref{thm:main}. In the case of maintaining a spanner, the potential lack of continuity discussed in Remark~\ref{remark:main} does not exist,
as instead of switching between the $O(\log(n))$ spanners maintained by the conversion in Theorem~\ref{thm:main}, we can just let the final spanner be the union of all of them. This incurs an extra $\log(n)$ factor in the size of the spanner.

\paragraph{Second Result: Dynamic Maximal Matching.}

A maximum cardinality matching can be maintained dynamically in $O(n^{1.495})$ amortized expected time per operation~\cite{Sankowski07}.
Due to conditional lower bounds of $\Omega(\sqrt m)$ on the time per operation for this 
problem~\cite{AbboudW14,HenzingerKNS15}, there is a large body of work on the dynamic \emph{approximate} matching problem 
~\cite{OnakR10,BaswanaGS18,NeimanS16,GuptaP13,BhattacharyaHI18,BhattacharyaHN16,Solomon16,BhattacharyaHN17,BhattacharyaCH17,GuptaK0P17,CharikarS18,ArarCCSW18}.
Still the only algorithms with polylogarithmic (amortized or worst-case) time per operation require a 2 or larger approximation ratio. 

A matching is said to be \emph{maximal} if the graph contains no edges between unmatched vertices. 
A maximal matching is guaranteed to be a 2-approximation of the maximum matching, and is also a well-studied object
in its own right (see e.g.\ \cite{HanckowiakKP01,GrandoniKP08,LattanziMSV11,BaswanaGS18,NeimanS16,Solomon16,Fischer17}).
The groundbreaking result of Baswana, Gupta, and Sen~\cite{BaswanaGS18} showed how to maintain a maximal matching (and so 2-approximation) with $O(\log(n))$ expected amortized update time. Solomon \cite{Solomon16} improved the update time to $O(1)$ expected amortized. There has been recent work on either deamortizing or derandomizing this result 
\cite{BhattacharyaHN16,BhattacharyaCH17,BhattacharyaHN17,CharikarS18,ArarCCSW18}. 
Most notably, the two independent results in
~\cite{CharikarS18} and ~\cite{ArarCCSW18} both present algorithms with polylog high-probability worst-case update time
that maintain a $(2+\eps)$-approximate matching. 
Unfortunately, all these results comes at the price of increasing the approximation factor from $2$ to $(2+\eps)$, and in particular no longer ensure that the matching is maximal. 
One of the central questions in this line of work is thus whether it is possible
to maintain a maximal matching without having to use both randomization \emph{and} amortization. 

We present the first affirmative answer to this question by removing the amortization requirement, thus resolving an open question of \cite{ArarCCSW18}. Much like for dynamic spanner, we use an existing amortized algorithm as our starting point: namely, the $O(\log(n))$ amortized algorithm of \cite{BaswanaGS18}. We then show how the algorithm and analysis can be modified to achieve a worst-case expected guarantee, 
and then we apply our reduction. 
\begin{theorem}
\label{thm:matching-expected}
There exists a fully dynamic (Las Vegas) algorithm for maintaining a maximal matching with worst-case expected update time $O(\log^3(n))$.
\end{theorem}  

\begin{theorem}
\label{thm:matching-converted}
There exists a fully dynamic (Las Vegas) algorithm that maintains a maximal matching with high-probability worst-case update time $O(\log^5(n))$.
\end{theorem}
The proof follows directly from Theorem~\ref{thm:matching-expected} and Theorem~\ref{thm:main}.
As in Remark~\ref{remark:main} above, we note that our worst-case algorithm in
Theorem~\ref{thm:matching-converted} stores the matching in a different data structure
than the original amortized algorithm of Baswana \etal \cite{BaswanaGS18}: while the latter stores the edges of the maximal matching in a single list $D$, our algorithm stores $O(\log(n))$ lists $D_i$, along with a pointer to some specific list $D_j$ that is guaranteed to contain the edges of 
a maximal matching. In particular, the algorithm always knows which $D_j$ is correct. The pointer to $D_j$ allows our algorithm to answer queries about the maximal matching in optimal time.

\paragraph{Discussion of Our Contribution.}

We present the first dynamic algorithms with worst-case polylog update times for two classical graph problems: dynamic spanner, and dynamic maximal matching. Both results are achieved with a new de-amortization approach, which shows that the concept of worst-case expected time can be a very fruitful way of thinking about dynamic graph algorithms. From a technical perspective, the conversion from worst-case expected to high-probability worst-case (Theorem~\ref{thm:main}) is relatively simple. The main technical challenge lies in showing how the existing amortized algorithms for dynamic spanner and maximal matching can be modified to be worst-case expected. The changes to the algorithms themselves are not too major, but a very different analysis is required, because we can no longer rely on charging arguments and potential functions. Our tools for proving  worst-case expected guarantees can be used to de-amortize other existing dynamic algorithms has already been used for two novel fully dynamic maximal independent set algorithms~\cite{ChechikZ19,BehnezhadDHSS19} and we expect it to find further use.
 For example, the dynamic coloring algorithm of~\cite{BhattacharyaCHN18}, the dynamic spectral sparsifier algorithm of \cite{AbrahamDKKP16}, the dynamic distributed maximal independent set algorithm of~\cite{Censor-HillelHK16}, and the dynamic distributed spanner algorithm of \cite{BaswanaKS12} (all amortized) seem like natural candidates for our approach.

Section~\ref{sec:reduction} provides a proof of the black-box reduction in Theorem~\ref{thm:main}. 
Section~\ref{sec:matching} presents our dynamic matching algorithm, and Section ~\ref{sec:spanner} presents our dynamic spanner algorithm. 

\subsubsection*{Acknowledgments}

	The conference version of this paper \cite{BernsteinFH19} had an error in the analysis of the dynamic matching algorithm. In particular, Lemma 4.5 assumed an independence between adversarial updates to the hierarchy which is in fact true, but which requires a sophisticated proof. We are very grateful to anonymous reviewers for pointing out this mistake in our analysis. The mistake is fixed in Section~\ref{subsec:matching-probability}. Almost the entire fix is a matter of analysis: the only change to the algorithm itself is the introduction of responsible bits in Algorithm~\ref{alg:matching 2}.
	
	The first author would like to thank Mikkel Thorup and Alan Roytman for a very helpful discussion of the proof of
	Theorem \ref{thm:main}.
	
	The research leading to these results has received funding from the European Research Council under the European Union's Seventh Framework Programme (FP/2007-2013) / ERC Grant Agreement no.~340506.
	
	The first author conducted this research while funded by NSF grant 1942010.

\section{Converting Worst-Case Expected to High-Probability Worst-Case}\label{sec:reduction}

In this section we give the proof of Theorem~\ref{thm:main}. To do so, we first prove the following theorem
that restricts the length of the update sequence
and then show how to extend it.

\begin{theorem}
\label{thm:main-extended}
Let $A$ be an algorithm that maintains a dynamic data structure $D$
with worst-case expected update time $\alpha$, let $n$ be a parameter such that the maximum number of items stored in the data structure at any point in time is polynomial in~$n$,
and let $ \ell $ be a parameter for the length of the update sequence to be considered that is polynomial in $n$. 
Then there exists an algorithm $A'$ with the following properties:
\begin{enumerate}
\item \label{item:main-worst-extended}
For any sequence of updates $ \sigma_1, \sigma_2, \ldots, \sigma_{\ell} $ with $\ell$ polynomial in $n$, $A'$ processes each update $\sigma_i$ in $O(\alpha \log^2(n))$ time
with high probability.
The amortized expected update time of $A'$ is $O(\alpha\log(n))$.
\item \label{item:main-correctness-extended}
$A'$ maintains $\Theta(\log(n))$ data structures $D_1, D_2, \ldots, D_{\Theta(\log(n))}$,
as well as a pointer to some $D_i$ that is guaranteed to be correct at the current time. Query operations are answered with $D_i$.
\end{enumerate}
\end{theorem}

\begin{proof}
Let $q = c\log(n)$ for a sufficiently large constant $c$. 
The algorithm runs $q$ versions of the algorithm $A$, denoted $A_1, \ldots, A_q$,
each with their own independently chosen random bits.
This results in $q$ data structures $D_i$.
Each $D_i$ maintains a possibly empty buffer $L_i$ of uncompleted updates. 
If $L_i$ is empty, $D_i$ is marked as \emph{fixed}, otherwise it is marked as 
\emph{broken}. The algorithm maintains a list of all the fixed data structures,
and a pointer to the $D_i$ of smallest index that is fixed.

Let $r = 4\alpha\log(\ell) = O(\alpha\log(\ell))$.
Given an update $\sigma_j$ the algorithm adds 
$\sigma_j$ to the end of each~$L_i$ and then allows each $A_i$ to run for $r$ steps. Each $A_i$ will work on the uncompleted updates
in $L_i$, continuing where it left off after the last update, and completing the first uncompleted update before starting the
next one in the order in which they appear in $L_i$. 
If within these $r$ steps all uncompleted updates in $L_i$ have been completed, $A_i$ marks itself as fixed;
otherwise it marks itself as broken.
If at the end of update $\sigma_j$ \emph{all} of the $q$ data structures $D_i$ are broken
then the algorithm performs a \Flush, which simply processes all the updates in all the versions $A_i$: 
this could take much more than $r$ work, but our analysis will show that this event happens with extremely small probability. 
The \Flush ensures Property~\ref{item:main-correctness-extended} of Theorem~\ref{thm:main-extended}: at the end of every update, some $D_i$ is fixed. 

By linearity of expectation, the expected amortized update time is  $O(\alpha q) = O(\alpha \log(n))$, 
and the worst-case update time is $rq= O(\alpha \log^2(n))$
unless a \Flush occurs. All we have left to show is that 
after every update the probability of a \Flush is at most $(1/2)^q = 1/n^c$.
We use the following counter analysis:

\begin{definition}
\label{dfn:counter}
We define the \emph{dynamic counter problem with positive integer parameters 
$\alpha$ (for average), $r$ (for reduction), and $\ell$ (for length)} as follows.
Given a finite sequence of possibly dependent random variables $X_1, X_2, \ldots, X_\ell $ such that for each $t$, $E[X_t] \leq \alpha$,
we define a sequence of counters $C_t$ which changes over a finite sequence of time steps. 
Let $C_0 = 0$ and let
$C_t = \max({X_t + C_{t-1} - r},{0})$.
\end{definition}
As we show in Lemma~\ref{lem:counter}  with constant probability $C_t$ is 0. We use this fact as follows:
Let $A_i$ be any version.
Each $D_i$ exactly mimics the dynamic counter of Definition~\ref{dfn:counter}:
$X_j$ corresponds to the time it takes for $A_i$ to process update $\sigma_j$; 
by the assumed properties of $A$, we have $E[X_j] = \alpha$. 
The counter $C_j$ then corresponds to the amount of work that $A_i$ has left to do
after the $j$-th update phase; in particular, $C_j = 0$ corresponds to $D_i$ being fixed after time $j$,
which by Lemma~\ref{lem:counter} occurs with probability at least $1/2$. Since all the $q$ 
versions $A_i$ have independent randomness, 
the probability that \emph{all} the $D_i$ are broken and a \Flush 
occurs is at most 
$(1/2)^{q} = 1/n^c$.
\end{proof}

\begin{lemma}
\label{lem:counter}
Given a dynamic counter problem with parameters $\alpha$, $r$, and $\ell$, if $r \geq 4 \alpha \log (\ell) $ and $ \alpha \geq 1 $
then for every $t$ we have $\PP[C_t = 0] \geq 1/2$.
\end{lemma}

\begin{proof}[Proof of Lemma~\ref{lem:counter}]
Let us focus on some $C_t$, and say that $k$ is the \emph{critical} moment
if it is the smallest index such that $C_j > 0$ for all $k \leq j \leq t$.
Note that there is exactly one critical moment if $C_t > 0$ (possibly $k=t$)
and none otherwise. Define $\bad_i$ ($\bad$ for bad) for $0 \leq i \leq \log(t)$ to be the event that the critical moment occurs in interval $(t + 1 - 2^{i+1}, t+1-2^i]$. Thus, 
\begin{equation}
\label{eq:probability-total}
\PP[C_t > 0] = \PP[\bad_0 \lor \bad_1 \lor \bad_2 \ldots \lor \bad_{\log(t)}]
= \sum_{0 \leq i \leq \log(t)} \PP[\bad_i] \, .
\end{equation}
We now need to bound $\PP[\bad_i]$. Note that if $\bad_i$ occurs, then $C_j > 0$
for $t+1-2^i \leq j \leq t$. Thus the counter reduces by $r$ at least $2^i$ times
between the critical moment and time $t$ ($2^i$ and not $2^i - 1$ because the counter reduces
at time $t$ as well).
Furthermore, the counter is always non-negative.
Thus,
\begin{equation*}
\bad_i \rightarrow \sum_{t+1-2^{i+1} \leq j \leq t} X_j \geq r 2^i \, ,
\end{equation*}
meaning that the event $ B_i $ implies the event $ \sum_{t+1-2^{i+1} \leq j \leq t} X_j \geq r 2^i $.
Plugging in for $r = 4 \alpha \log(\ell)$ and recalling that if event $E_1$ implies $E_2$ then
$\PP[E_1] \leq \PP[E_2]$ we have that
\begin{equation}
\label{eq:bad-event}
\PP[\bad_i] \leq \PP \left[ \sum_{t+1-2^{i+1} \leq j \leq t} X_j \geq 2 \cdot \log(\ell) \cdot \alpha \cdot 2^{i+1} \right] \, .
\end{equation}
Now observe that, by linearity of expectation,
\begin{equation}
\label{eq:expected-r}
E \left[ \sum_{t+1-2^{i+1} \leq j \leq t} X_j \right] = \sum_{t+1-2^{i+1} \leq j \leq t} E [ X_j ] \le \alpha \cdot 2^{i+1} \, .
\end{equation}
Combining the Markov inequality with Equations~\ref{eq:bad-event} and~\ref{eq:expected-r} yields 
$\PP[B_i] \leq 1/(2\log(\ell))$ for any $i$. Plugging that into Equation~\ref{eq:probability-total}, and recalling
that $t \leq \ell$, we get $\PP[C_t > 0] \leq \sum_{0 \leq i \leq \log(t)} 1/(2\log(\ell)) \leq 1/2.$
\end{proof}

Note that the $\log(\ell)$ factor is necessary, even though intuitively
$r = O(\alpha)$ should be enough, since at each step the counter only goes up 
by $\alpha$ (in expectation) and goes down by $r > \alpha$, so we would expect it
to be zero most of the time. And that is in fact true: with $r = 4\alpha$ 
one could show that for any~$\ell$, the probability that $C_t = 0$
for at least half the values of $t \in [0,\ell]$ is at least $1/2$.
But this claim is not strong enough because it still
leaves open the possibility that even if the counter is usually zero,
there is some particular time $t$ at which $\PP[C_t = 0]$ is very small.

To exhibit this bad case, consider the following sequence $X_1, X_2, \ldots X_\ell$,
where each $X_t$ is chosen independently and is set to $2r(\ell+1-t) $ with probability 
$\frac{\alpha}{2r(\ell+1-t)}$ and to $0$ otherwise. 
It is easy to see that for each $t \leq \ell$ we have $E[X_t] = \alpha$.
Now, what is $\PP[C_\ell = 0]$? 
For each $t \leq \ell$ if $X_t \neq 0$, the counter will reduce by $r(\ell+1-t)$ from time
$t$ to time $\ell$, which still leaves us with $C_\ell \geq 2r(\ell+1-t) - r(\ell+1-t) > 0$. 
Let $Y_t$ be the indicator variable for the event that $X_t \neq 0$.
Then, $\PP[C_\ell > 0] = \PP[Y_1 \lor Y_2 \ldots \lor Y_\ell]$.
This probability is hard to bound exactly, but note that
since the~$Y_t$'s are independent random variables between $0$ and $1$ and we can apply the following Chernoff bound.
\begin{lemma} [Chernoff Bound]
\label{lem:chernoff}
Let $Y_1, Y_2, \ldots, Y_k$ be a sequence of independent random variables
such that $ 0 \leq Y_t \leq U $ for all $ t $.
Let $ Y = \sum_{1 \leq t \leq k} Y_t $ and $\mu = E[Y]$.
Then the following two properties hold for all $\del > 0$:
\begin{align}
\PP [ Y \leq (1-\del) \mu ] &\leq e^{- \frac{\del^2 \mu}{2 U}} \\
\PP [ Y \geq (1+\del)\mu ] &\leq e^{- \frac{\del \mu}{3 U}} \, .
\end{align}
\end{lemma}
Formulation~1 with $\del = .74$ yields that if $ \sum_{1 \leq t \leq \ell} E[Y_t] \geq 4 $, then
\begin{equation*}
\PP[C_\ell = 0] = \PP \left[ \sum_{1 \leq t \leq \ell} Y_t < 1 \right] < .34 < 1/2 \, .
\end{equation*}
Thus, to have $\PP[C_\ell = 0] \geq 1/2$ we certainly need $\sum_{1 \leq t \leq \ell} E[Y_t] < 4.$
Now observe that
\begin{equation*}
\label{eq:lower-bound}
\sum_{1 \leq t \leq \ell} E[Y_t] 
= \frac{\alpha}{2r} \sum_{1 \leq t \leq \ell} \frac{1}{\ell+1-t} = \frac{\alpha \cdot \Omega (\log{\ell})}{r}
\end{equation*}
Thus, to have $\sum_{1 \leq t \leq \ell} E[Y_t] \leq 4$, we  indeed need $r = \Omega(\alpha\log(\ell))$.

This now completes the proof of all parts of Theorem~\ref{thm:main-extended}.

Finally, we observe that the restriction to an update sequence of finite length is mainly a technical constraint and we show next how to remove it. The basic idea is to periodically rebuild a new copy of the data structure ``in the background'' by spreading this computation over the time period.

\begin{proof}[Proof of Theorem~\ref{thm:main}]
Note that if the data structure does not allow any updates then Theorem~\ref{thm:main-extended} gives the desired bound.
Otherwise the data structure allows either insertions or deletions or both.
In this case we use a standard technique to enhance the algorithm $ A' $ from Theorem~\ref{thm:main-extended} providing worst-case high probability update time for a \emph{finite} number of updates to an algorithm $ A''$ providing worst-case high probability update time for an \emph{infinite} number of updates.
Recall that we assume that the maximum number of items that are stored in the data structure at any point in
time as well as the preprocessing time  to build the data structure for any set $S$ of size polynomial in~$n$ is polynomial in~$n$. Let this polynomial be upper bounded by $n^c$ for some constant c.
We break the infinite sequence of updates into non-overlapping \emph{phases}, such that phase $i$ consists of all updates
between update $i \times n^c$ to update $(i+1) \times n^c -1$. 

During each phase the algorithm uses two instances of algorithm $ A' $, one of them being called \emph{active} and
one being called \emph{inactive.} For each instance the algorithm has a pointer that points to the corresponding data structure.
Our new algorithm $ A'' $ always points to the data structures $D_1, D_2, \ldots, D_{\log(1/p)}$ of the active instance, where $p$ is a suitably chosen parameter.
In particular it also points to the $ D_i $ for which the active instance ensures correctness.
At the end of a phase the inactive data structure of the current phase becomes the active data structure for the next phase and the active one becomes the inactive one.

Additionally, $A'$ keeps a list $L$ of all items (e.g.~edges in the graph) that are currently stored in the data structure, stored in a balanced binary search tree, such
that adding and removing an item takes time $O(\log n)$ and the set of items that are currently in the data structure can be listed in time linear in their number.

We now describe how each of the two instances is modified during a phase. In the following when we use the term \emph{update} we mean an update in the (main) data 
structure.

(1) \emph{Active instance.} All updates are executed in the active instance and these are the only modifications performed on the active data structure.

(2) \emph{Inactive instance.} 
During the first $ n^c /2 $ updates in a phase, we do not allow any changes to $L$, but record all these updates.
Additionally during the first $n^c/4$ updates in the phase, we enumerate all items in $L$ and store them in an array by performing a constant amount of work of the enumeration and copy algorithm for each update. Let $S$ denote this set of items.
During the next $n^c/4$ updates we run the preprocessing algorithm for $S$ to build the corresponding data structure, 
again by performing a constant amount of work per update. This data structure becomes our current version of the inactive instance.

We also record all updates of the second half of the phase.
During the third $ n^c / 4 $ updates in the phase, we forward to the inactive instance and to~$L$ all $ n^c / 2 $ updates of the \emph{first} half of the current phase, by performing two recorded updates to the 
inactive instance and to $L$ per update in the second half of the phase.
Finally, during the final $ n^c / 4 $ updates, we forward to the inactive 
instance and to~$L$ all $ n^c / 2 $ updates of the \emph{second} half of the current phase, again performing two recorded update per update. 
This process guarantees that at the end of a phase the items stored in the active and the inactive instance are
identical. 

The correctness of this approach is straightforward.
To analyze the running time, observe that each update to the data structure will result in one update being processed by the active instance and at most two updates being processed in the inactive instance. Additionally maintaining $L$ 
increases the time per update by an additive amount of $O(\log n)$.
By the union bound, our new algorithm $ A'' $ spends worst-case time $2 \cdot O(\alpha\log(n)\log(1/p))$ with probability $ 1 - 2/p $.
By linearity of expectation, $ A'' $ has amortized expected update time  $2 \cdot O(\alpha\log(1/p))$.
By initializing the instance in preparation with the modified probability parameter $ p' = p/2 $ we obtain the desired formal guarantees.
\end{proof}

\section{Dynamic Spanner with Worst-Case Expected Update Time}
\label{sec:spanner}

In this section, we give a dynamic spanner algorithm with worst-case expected update time that, by our main reduction, can be converted to a dynamic spanner algorithm with high-probability worst-case update time with polylogarithmic overheads.
We heavily build upon prior work of Baswana et al.~\cite{BaswanaKS12} and replace a crucial subroutine requiring deterministic amortization by a randomized counterpart with worst-case expected update time guarantee.
In Subsection~\ref{sec:spanner high level}, we first give a high-level overview explaining where the approach of Baswana et al.~\cite{BaswanaKS12} requires (deterministic) amortization and how we circumvent it.
We then, in Subsection~\ref{sec:spanner algorithm Baswana}, give a more formal review of the algorithm of Baswana et al.\ together with its guarantees and isolate the dynamic subproblem we improve upon.
Finally, in Subsection~\ref{sec:worst-case filtering}, we give our new algorithm for this subproblem and work out its guarantees.

\subsection{High-Level Overview}\label{sec:spanner high level}

Recall that in the dynamic spanner problem, the goal is to maintain, for a graph $ G = (V, E) $ with $ n = |V| $ vertices that undergoes edge insertions and deletions, and a given integer $ k \geq 2 $, a subgraph $ H = (V, F) $ of size $ |F| = \tilde O (n^{1 + 1/k}) $ such that for every edge $ (u, v) \in E $ there is a path from $ u $ to $ v $ in~$ H $ of length at most $ 2 k - 1 $.
If the latter condition holds, we also say that the spanner has stretch $ 2k - 1 $.

The algorithm of Baswana et al.\ emulates a ``ball-growing'' approach for maintaining hierarchical clusterings.
In each ``level'' of the construction, we are given some clustering of the vertices and each cluster is sampled with probability $ p = \tfrac{1}{n^{1/k}} $.
The sampled clusters are grown as follows:
Each vertex in a non-sampled cluster that is incident on at least one sampled cluster, joins one of these neighboring sampled clusters.
Thus, for each unclustered vertex, there might be a choice as to which of its neighboring sampled clusters to join.
Furthermore, the algorithm distinguishes the edge that a non-sampled vertex uses to ``hook'' onto the sampled cluster it joins.
All sampled clusters (after possibly being extended by the hooks) together with the edges between them move to the next level of the hierarchy and in this way the growing of clusters is repeated $ k - 1 $ times.
The main idea why this hierarchy gives a good spanner is the following:
If a vertex belonging to an unsampled cluster has many neighboring clusters, then one of them is likely to be a sampled one and so the vertex joins a sampled cluster and is passed on to the next level of the hierarchy.
Conversely, if it stays at the current level of the hierarchy, then it only has few neighboring clusters, namely $ O(\tfrac{1}{p}) = O(n^{1/k}) $ many in expectation.
For such vertices, one can therefore afford to add one edge per neighboring cluster to the spanner.
By doing so, it is ensured that there is a path of length $ 2 k - 1 $ for each incident edge as every cluster has radius at most $ k - 1 $.

This hierarchy is maintained with the help of sophisticated data structures and some crucial applications of randomization to keep the expected update time low.
One important aspect for bounding the update time in such a hierarchical approach is the following:
It is not sufficient to analyze the update time at each level of the hierarchy in isolation as updates performed to one level might lead to changes in the clustering that lead to \emph{induced updates} to the next level.
In principle, by such a propagation of updates, a single update to the input graph might lead to an exponential number of induced updates to be processed by the last level.
Baswana et al.\ show that the amortized expected number of induced updates at level~$ i $ per update to the input graph is at most $ O (1)^i $.
Our contribution in this section is to remove the amortization argument, i.e., to give a bound of $ O (1)^i $ with worst-case expected guarantee

In the first level of the hierarchy, each vertex is a singleton cluster and each non-sampled vertex picks, among all edges going to neighboring sampled vertices, one edge uniformly at random as its hook.
Now consider the deletion of some edge $ e = (u, v) $.
If $ e $ was not the hook of $ u $, then the clustering does not need to be fixed.
However, if $ e $ was the hook, then the algorithm spends time up to $ O (\deg (u)) $ for picking a new hook, possibly joining a different cluster, and if so informing all neighbors about the cluster change.
If the adversary deleting $ e $ is oblivious to the random choices of the algorithm (both the choice of the sampled singleton clusters and the choice of the hooks), then every edge incident on $ u $ has the same probability of being the hook of $ u $, i.e., the probability of $ e $ being the hook of $ u $ is $ \tfrac{1}{\deg (u)} $.
Thus, the expected update time is $ \tfrac{1}{\deg (u)} \cdot O (\deg (u)) = O (1) $.

The situation is more complex at higher levels, when the clusters are not singleton anymore.
While the time spent upon deleting the hook is still $ O (\deg_i (u)) $, where $ \deg_i (u) $ is the degree of $ u $ at level~$ i $, one cannot argue that the probability of the deleted edge being the hook is $ O (\tfrac{1}{\deg_i (u)}) $.
To see why this could be the case, Baswana et al.\ provide the following example of a ``skewed'' distribution of edges to neighboring clusters:
Suppose $ u $ has $ \ell = \Theta (\frac{1}{p} \log(n)) $ neighboring clusters such that there are $ \Theta (n) $ edges from $ u $ into the first neighboring cluster and each remaining neighboring cluster has only one edge incident on~$ u $.
Now there is a quite high probability (namely $ 1 - p \approx 1 $) that the first cluster is not sampled and with high probability $ O (\log(n)) $ of the remaining clusters will be sampled, as follows from the Chernoff bound.
Thus, if $ u $ picked the hook uniformly at random from all edges into neighboring sampled clusters, it would join one of the single-edge clusters with high-probability.
As there are $ \ell $ edges incident on $ u $ from these single-edge clusters, this gives a probability of approximately $ \tfrac{1}{\ell} $ for some deleted edge $ (u, v) $ being the hook, which can be much larger than $ \tfrac{1}{\deg_i (u)} $.
This problem would not appear if among all edges going to neighboring clusters a $p$-th fraction would be incident on sampled clusters.
Then, intuitively speaking, one could argue that the probability of some edge $ e = (u, v) $ being the hook of $ u $ is at most $ p \cdot \tfrac{1}{\Omega(p \deg_i (u))} $, the probability that the cluster containing $ v $ is a sampled one times the probability that a particular edge among all edges to sampled clusters was selected.

This is why Baswana et al.\ introduce an edge \emph{filtering} step to their algorithm.
By making a sophisticated selection of edges going to the next level of the hierarchy, they can ensure that (a) among all such selected edges going to neighboring clusters a $p$-th fraction goes to sampled clusters and (b) to compensate for edges not being selected for going to the next level, each vertex only needs to add $ O (\frac{1}{p} \log^2 (n)) = O (n^{1/k} \log^2 (n)) $ edges to neighboring clusters to the spanner.
The filtering boils down to the following idea:
For each vertex~$ u $, group the neighboring non-sampled clusters into $ O (\log (n)) $ buckets such that clusters in the same bucket have approximately the same number of edges incident on $ u $.
For buckets that are large enough (containing $ \Theta (\tfrac{1}{p} \log (n)) $ clusters), a standard Chernoff bound for binary random variables guarantees that a $ p $-th fraction of \emph{all} clusters in the respective range for the number of edges incident on $ u $ go to sampled clusters.
As all these clusters have roughly the same number of edges incident on $ u $, a Chernoff bound for positive random variables with bounded aspect ratio also guarantees that a $ p $-th fraction of the edges of these clusters will go to sampled clusters.
Therefore, one gets the desired guarantee if all edges incident on clusters of small buckets are prevented from going to the next level in the hierarchy.
To compensate for this filtering, it is sufficient to add one edge -- picked arbitrarily -- from $ u $ to each cluster in a small bucket to the spanner.
As there are at most $ O (\log (n)) $ small buckets containing $ O (\tfrac{1}{p} \log (n)) $ clusters each, this step is affordable without blowing up the asymptotic size of the spanner too much.

Maintaining the bucketing is not trivial because whenever a cluster moves from one bucket to the other it might find itself in a small bucket coming from a large bucket, or vice versa.
In order to enforce the filtering constraint, this might cause updates to the next level of the hierarchy.
One way of controlling the number of induced updates is amortization:
Baswana et al.\ use soft thresholds for the upper and lower bounds on the number of edges incident on $ u $ for each bucket.
This ensures that updates introduced to the next level can be charged to updates in the current level, and leads to an amortized bound of $ O (1) $ on the number of induced updates.
Note that the filtering step is the only part in the spanner algorithm of Baswana et al.\ where this deterministic amortization technique is used.
If it were not for this specific sub-problem, the dynamic spanner algorithm would have worst-case expected update time.

Our contribution is a new dynamic filtering algorithm with worst-case expected update time, which then gives a dynamic spanner algorithm with worst-case expected update time.
Roughly speaking, we achieve this as follows: whenever the number of edges incident on $ u $ for a cluster~$ c $ in some bucket~$ j $ (with $ 0 \leq j \leq O (\log (n)) $) exceeds a bucket-specific threshold of~$ \alpha_j $, we move $ c $ up to the appropriate bucket with probability $ \Theta(\tfrac{1}{\alpha_j}) $ after each insertion of an edge between $ u $ and~$ c $.
This ensures that, with high probability, the number of edges to $ u $ for clusters in bucket~$ j $ is at most $ O (\alpha_j \log (n)) $.
Such a bound immediately implies that the expected number of induced updates to the next level per update to the current level is $ O (\tfrac{1}{\alpha_j} \cdot \alpha_j \log (n)) = O (\log (n)) $, which is already non-trivial but also unsatisfactory because it would lead to an overall update time of $ O(\log (n))^{k/2} $ for a $ (2 k - 1) $-spanner, instead of $ O (1)^{k/2} $ as in the case of Baswana et al.
By a more careful analysis we do actually obtain the $ O (1)^{k/2} $-bound.
By taking into account the diminishing probability of not having moved up previously, we argue that the probability to exceed the threshold by a factor of $ 2^t $ is proportional to $ 1/e^{(2^t)} $.
This bounds the expected number of induced updates by $ \sum_{t \geq 1} 2^t / e^{(2^t)} $, which converges to a constant.
A similar, but slightly more sophisticated approach, is applied for clusters moving down to a lower-order bucket.
Here we essentially need to adapt the sampling probability to the amount of deviation from the threshold because in the analysis we have fewer previous updates available for which the cluster has not moved, compared to the case of moving up.

\subsection{The Algorithm of Baswana et al.}\label{sec:spanner algorithm Baswana}

In the following, we review the algorithm of Baswana et al.~\cite{BaswanaKS12} for completeness and isolate the filtering procedure we want to modify.
We deviate from the original notation only when it is helpful for our purposes.

\subsubsection{Static Spanner Construction}

Let us first explain the principle behind the algorithm of Baswana et al.\ by reviewing a purely static version of the construction.

Given an integer parameter $ k \geq 2 $, the construction uses clusterings $ C_0, C_1, \ldots, C_{k-1} $ of subgraphs $ G_0 = (V_0, E_0), G_1 = (V_1, E_1), \ldots, G_{k-1} = (V_{k-1}, E_{k-1}) $, both to be specified in the following, where $ G_0 = G $ and, for each $ 0 \leq i \leq k - 2 $, $ G_{i+1} $ is a subgraph of $ G_i $ (i.e., $ V_{i+1} \subseteq V_i $ and $ E_{i+1} \subseteq E_i $).
For each $ 0 \leq i \leq k - 1 $, a \emph{cluster} of $ G_i $ is a connected subset of vertices of $ G_i $ and the \emph{clustering} $ C_i $ is a partition of $ G_i $ into disjoint clusters.
To control the size of the resulting spanner, the clusterings are partially determined by a hierarchy of randomly sampled subsets of vertices $ S_0 \supseteq S_1 \supseteq \dots \supseteq S_k $ in the sense that each cluster $ c $ in $ C_i $ contains a designated vertex of $ S_i $ called the \emph{center} of $ c $.
This sampling is performed by setting $ S_0 = V $, $ S_k = \emptyset $, and by forming~$ S_i $, for each $ 1 \leq i \leq k-1 $, by selecting each vertex from $ S_{i-1} $ independently with probability $ p = \tfrac{1}{n^{1/k}} $.
In addition to the clusterings, the construction uses a forest $ (V_i, F_i) $ consisting of a spanning tree for each cluster of~$ C_i $ rooted at its center such that each vertex in the cluster has a path to the root of length at most~$ i $.
Informally, level~$ i $ of this hierarchy denotes all the sets of the construction indexed with~$ i $.
Initially, $ G_0 = G $, $ F_0 = \emptyset $ and the clustering $ C_0 $ consists of singleton clusters $ \{ v \} $ for all vertices $ v \in S_0 = V $.

We now review how to obtain, for every $ 0 \leq i \leq k-1 $, the graph $ G_{i + 1} = (V_{i + 1}, E_{i + 1}) $, the clustering $ C_{i + 1} $ of $ G_{i + 1} $, and the set of edges $ F_{i + 1} $, based on the graph $ G_i = (V_i, E_i) $, the clustering $ C_i $, the edge set~$ F_i $, and the set of vertices $ S_{i + 1} $.
Let $ R_i $ be the set of all ``sampled'' clusters in the clustering~$ C_i $, i.e., all clusters in $ C_i $ whose cluster center is contained in $ S_{i + 1} $.
Furthermore, let $ V_{i + 1} $ be the set consisting of all vertices of $ V_i $ that belong to or are adjacent to clusters in $ R_i $ and let $ \mathcal{N}_i $ be the set consisting of all vertices of $ V_i $ that are adjacent to, but do not belong to, clusters in~$ R_i $.
Finally, for every $ u \in V_i $, let $ E_i (u) $ denote the set of edges of~$ E_i $ incident on $ u $ and any other vertex of $ V_i $, and, for every $ u \in V_i $ and every $ c \in C_i $, let $ E_i (u, c) $ denote the set of edges of~$ E_i $ incident on $ u $ and any vertex of~$ c $.
For each vertex $ u \in \mathcal{N}_i $, the construction takes an arbitrary edge $ (u, v) \in \bigcup_{c \in R_i} E_i (u, c) $ as the \emph{hook} of $ u $ at level~$ i $, called $ \hook (u, i) $.
Now the clustering $ C_{i + 1} $ is obtained by adding each vertex $ u \in \mathcal{N}_i $ to the cluster of the other endpoint of its hook and the forest $ F_{i + 1} $ is obtained from~$ F_i $ by extending the spanning trees of the clusters by the respective hooks.
To compensate for vertices that cannot hook onto any cluster in~$ R_i $, let $ X_i $ be a set of edges containing for each vertex $ v \in V_i \setminus V_{i+1} $ exactly \emph{one} edge of $ E_i (u, c) $ -- picked arbitrarily -- for each non-sampled neighboring cluster~$ c \in C_i \setminus R_i $.
Finally, the edge set~$ E_{i + 1} $ is defined as follows.
Every edge $ (u, v) \in E_i $ with $ u, v \in V_{i + 1} $ belongs to~$ E_{i + 1} $ if and only if $ u $ and $ v $ belong to different clusters in $ C_{i + 1} $ and at least one of $ u $ and $ v $ belongs to a sampled cluster (in $ R_i $) at level~$ i $.

The static spanner $ H $ now consists of the set of edges $ \bigcup_{0 \leq i \leq k-1} (F_i \cup X_i) $.
To analyze the stretch of $ H $, consider some edge $ e = (u, v) $ and let $ i $ be the largest index such that $ e \in E_i $.
If $ u $ and $ v $ are contained in the same cluster in the clustering $ C_i $, then the path from $ u $ to $ v $ in $ F_i $ via the common cluster center has length at most $ 2 i \leq 2 k - 2 $, as each cluster has radius at most $ i \leq k - 1 $.
If $ u $ and $ v $ are contained in different clusters in the clustering $ C_i $, then $ X_i $ contains an edge $ e' = (u, v') $ from $ u $ to the cluster of $ v $.
Now there is a path in $ H $ of length at most $ 2 i + 1 \leq 2 k - 1 $ from $ u $ to $ v' $ by first taking the edge $ e' $ to $ v' $ and then taking the path from $ v' $ to $ v $ in $ F_i $ via the common cluster center.

To analyze the size of the spanner, observe first that each forest $ F_i $ consist of at most $ n - 1 $ edges.
Furthermore, for each $ 0 \leq i \leq k - 2 $, each vertex in $ V_i \setminus V_{i+1} $, which is the set of vertices not being adjacent to a sampled cluster, is adjacent to at most $ \tfrac{1}{p} = n^{1/k} $ clusters in expectation (all of which are non-sampled clusters)
Thus, the number of edges contributed to $ X_i $ by each vertex in $ V_i \setminus V_{i+1} $ is at most $ n^{1/k} $ in expectation.
At level $ k - 1 $, no clusters are sampled ones anymore and $ X_{k-1} $ contains for each vertex in $ V_{k-1} $ one edge to each neighboring cluster.
As the number of clusters has been reduced to $ n^{1/k} $ in expectation at level $ k - 1 $, each vertex in $ V_{k-1} $ again contributes $ n^{1/k} $ edges to~$ X_{k-1} $ in expectation.
This results in an overall spanner size of $ O (k n^{1 + 1/k}) $ in expectation.

\subsubsection{Dynamic Spanner Maintenance}

The dynamic spanner algorithm uses the same definitions as above, with some minor modifications regarding how the hooks and the sets $ E_i $ are determined, and an additional edge set $ Y_i $ being included in $ H $ for each level $ i $.
Note that the sampling of $ S_0 \supseteq S_1 \supseteq \dots \supseteq S_k $ is performed a priori at the initialization and does not change over the course of the algorithm.
At each level $ i $ (for $ 0 \leq i \leq k-1 $), instead of selecting an arbitrary edge from $ (u, v) \in \bigcup_{c \in R_i} E_i (u, c) $ as the hook of $ u $ for each vertex $ u \in \mathcal{N}_i $, the hook is picked uniformly at random guaranteeing the following ``hook invariant'':
\begin{enumerate}
\item[(HI)] For every edge $ (u, v) \in \bigcup_{c \in R_i} E_i (u, c) $, where $ \bigcup_{c \in R_i} E_i (u, c) $ is the set of edges of $ E_i $ incident on $ u $ and any vertex contained in a cluster of $ R_i $, $ \PP [(u, v) = \hook (u, i)] = \tfrac{1}{| \bigcup_{c \in R_i} E_i (u, c) |} $.
\end{enumerate}
The main idea of Baswana et al.\ is that this simple method of choosing the hook leads to a fast update time if an additional filtering step is performed for selecting the edges that go to the next level.

For this purpose, the algorithm maintains, for each $ u \in \mathcal{N}_i $, and for certain parameters $ \lambda \geq g > 1 $, $ 0 < \epsilon < 1 $ and $ a > 1 $, a partition of the non-sampled neighboring clusters of $ u $ into $ \lceil \log_g (n) \rceil $ subsets called ``buckets'', a set of edges $ \mathcal{F}_i (u) \subseteq \bigcup_{c \in C_i \setminus R_i} E_i (u, c) $ and a set of clusters $ \mathcal{I}_i (u) \subseteq C_i \setminus R_i $ such that:\footnote{Here we slightly deviate from the original presentation of Baswana et al.\ by making the filtering process more explicit and also by giving the set $ \mathcal{I}_i (u) $ a name. We further deviate by suggesting to maintain this partitioning into buckets (which we call dynamic filtering) for each node in $ V_i $ (a superset of $ \mathcal{N}_i $). This does not increase the asymptotic running time of the overall algorithm and avoids special treatment when vertices join or leave $  \mathcal{N}_i $. Baswana et al. explicitly provide an argument for charging the initialization for of a vertex joining $ \mathcal{N}_i $ to a sequence of induced updates. We believe that our variant that avoids initialization slightly simplifies the formulation of Theorem~\ref{thm:main theorem Baswana}.}
\begin{enumerate}
\item[(F1)] For every $ 0 \leq j \leq \lfloor \log_g (n) \rfloor $ and every cluster $ c $ in bucket~$ j $, $ \tfrac{g^j}{\lambda} \leq | E_i (u, c) | \leq \lambda g^j $.
\item[(F2)] For every edge $ (u, v) \in \mathcal{F}_i (u) $, the bucket containing the cluster of $ v $ contains at least $ \ell := 4 \gamma a \lambda^2 \tfrac{1}{\epsilon^3} n^{1/k} \ln (n) \ln (\lambda) $ clusters (where $ \gamma \leq 80 $ is a given constant).
\item[(F3)] For every edge $ (u, v) \in \bigcup_{c \in C_i \setminus R_i} E_i (u, c) \setminus \mathcal{F}_i (u) $, the (unique) cluster of $ v $ in $ C_i $ is contained in $ \mathcal{I}_i (u) $.\footnote{The filtering algorithm of Baswana et al. guarantees the following stronger version of (F3): For every cluster $ c \in C_i \setminus R_i $ either $ E_i (u, c) \subseteq \mathcal{F}_i (u) $ or $ c \in \mathcal{I}_i (u) $. However, for the spanner algorithm to be correct, only the weaker guarantee of (F3) stated above is necessary. We will use this degree of freedom in our new filtering algorithm to avoid unnecessary ``bookkeeping'' work.}
\end{enumerate}
Intuitively, the set $ \mathcal{F}_i (u) $ is a \emph{filter} on the edges from $ u $ to non-sampled neighboring clusters and only edges to non-sampled clusters in $ \mathcal{F}_i (u) $ may be passed on to the next level in the hierarchy.
The clusters in $ \mathcal{I}_i (u) $ are those for which not all edges incident on $ u $ are contained in $ \mathcal{F}_i (u) $ and thus the algorithm has to compensate for these missing edges to keep the spanner intact.
For this purpose, the algorithm maintains a set of edges $ Y_i $ containing, for each vertex $ u \in V_{i+1} $ and each cluster $ c \in \mathcal{I}_i (u) $, exactly \emph{one} edge from $ E_i (u, c) $ -- picked arbitrarily.\footnote{Note that the lack of ``disjointness'' between $ \mathcal{F}_i (u) $  and $ \mathcal{I}_i (u) $ might lead to the situation that some edge is contained in both $ Y_i $ and $ \mathcal{F}_i (u) $. This was not the case in the original algorithm of Baswana et al., but it is correct to allow this behavior and allows us to avoid unnecessary ``bookkeeping'' work in our new filtering algorithm.}
In the following, we call an algorithm maintaining $ \mathcal{F}_i (u) $ and $ \mathcal{I}_i (u) $ satisfying (F1), (F2), and (F3) for a given vertex~$ u $ a \emph{dynamic filtering algorithm} with parameters $ \epsilon $ and $ a $.

For every vertex $ u $, let $ \mathcal{E}_i (u) = \mathcal{F}_i (u) \cup \bigcup_{c \in R_i} E_i (u, c) $ (where the latter is the set of edges incident on $ u $ from sampled clusters).
Now, the edge set $ E_{i + 1} $ is defined as follows.
Every edge $ (u, v) \in E_i $ with $ u, v \in V_{i + 1} $ belongs to $ E_{i + 1} $ if and only if $ u $ and $ v $ belong to different clusters in $ C_{i + 1} $ and one of the following conditions holds:
\begin{itemize}
\item At least one of $ u $ and $ v $ belongs to a sampled cluster (in $ R_i $) at level~$ i $, or
\item $ (u, v) $ belongs to $ \mathcal{E}_i (u) $ as well as $ \mathcal{E}_i (v) $.
\end{itemize}
Having defined this hierarchy, the dynamic spanner $ H $ consists of the set of edges $ \bigcup_{0 \leq i \leq k-1} (F_i \cup X_i \cup Y_i) $.

\subsubsection{Sketch of Analysis}

As explained above, it follows from standard arguments that $ | F_i \cup X_i | \leq O (n^{1 + 1/k}) $ for each $ 0 \leq i \leq k - 1 $.
Furthermore, the size of $ Y_i $ is bounded by $ n \cdot \max_{u} | \mathcal{I}_i (u) | $ for each $ 0 \leq i \leq k - 1 $.
The stretch bound of $ 2k - 1 $ follows from the clusters having radius at most $ k - 1 $ together with an argument that for each edge $ e = (u, v) $ not moving to the next level $ u $ has an edge to the cluster of~$ v $ (or vice versa) in one of the $ X_i $'s or one of the $ Y_i $'s.
Finally, the fast amortized update time of the algorithm is obtained by the random choice of the hooks.
Roughly speaking, the algorithm only has to perform significant work when the oblivious adversary hits a hook upon deleting some edge $ (u, v) $ from $ E_i $; this happens with probability $ \Omega (\tfrac{1}{| \mathcal{E}_i (u) |}) $ and -- by using appropriate data structures -- incurs a cost of $ O (| \mathcal{E}_i (u) |) $, yielding constant expected cost per update to~$ E_i $.
More formally, the filtering performed by the algorithm together with invariant~(HI) guarantees the following property.
\begin{lemma}[\cite{BaswanaKS12}]\label{lem:filtering property}
For every $ 0 \leq i \leq k - 1 $ and every edge $ (u, v) \in E_i $, $ \PP [ (u, v) = \hook (u, i)] \leq \tfrac{1 + 2 \epsilon}{| \mathcal{E}_i (u) |} $ for any constant $ 0 < \epsilon \leq \tfrac{1}{4} $.
\end{lemma}

The main probabilistic tool for obtaining this guarantee is a Chernoff bound for positive random variables.
Compared to the well-known Chernoff bound for binary random variables, the more general tail bound needs a longer sequence of random variables to guarantee a small deviation from the expectation with high probability: the overhead is a factor of $ b \log (b) $, where $ b $ is the ratio between the largest and the smallest value of the random variables.

\begin{theorem}[\cite{BaswanaKS12}]\label{thm:special Chernoff bound}
Let $ o_1, \ldots, o_\ell $ be $ \ell $ positive numbers such that the ratio of the largest to the smallest number is at most $ b $, and let $ Z_1, \ldots, Z_\ell $ be $ \ell $ independent random variables such that $ Z_i $ takes value $ o_i $ with probability $ p $ and $ 0 $ otherwise.
Let $ \mathcal{Z} = \sum_{1 \leq i \leq \ell} Z_i $ and $ \mu = \EX [\mathcal{Z}] = \sum_{1 \leq i \leq \ell} o_i p $.
There exists a constant $ \gamma \leq 80 $ such that if $ \ell \geq \gamma a b \tfrac{1}{\epsilon^3 p} \ln (n) \log (b) $ for any $ 0 < \epsilon \leq \tfrac{1}{4} $, $ a > 1 $, and a positive integer~$ n $, then the following inequality holds:
\begin{equation*}
\PP [\mathcal{Z} < (1 - \epsilon) \mu] < \frac{1}{n^a}
\end{equation*}
\end{theorem}

The running-time argument sketched above only bounds the running time of each level ``in isolation''.
For every $ 0 \leq i \leq k - 1 $, one update to $ G_i $ could lead to more than one \emph{induced} update to~$ G_{i + 1} $.
Thus, the hierarchical nature of the algorithm leads to an exponential blow-up in the number of induced updates and thus in the running time.
Baswana et al.\ further argue that the hierarchy only has to be maintained up to level $ \lfloor \tfrac{k}{2} \rfloor $ by using a slightly more sophisticated rule for edges to enter the spanner from the top level.
Together with a careful choice of data structures that allows constant expected time per atomic change, this analysis gives the following guarantee.

\begin{theorem}[Implicit in~\cite{BaswanaKS12}]\label{thm:main theorem Baswana}
Assume that for constant $ 0 < \epsilon < 1 $ and $ a > 1 $ there is a fully dynamic edge filtering algorithm~$ \mathfrak{F} $, in expectation, generates at most $ U (n) $ changes to $ \mathcal{F}_i (u) $ per update to $ E_i (u) $ and, in expectation, has an update time of $ U (n) \cdot T (n) $.
Then, for every $ k \geq 2 $, there is a fully dynamic algorithm $ \mathfrak{S} $ for maintaining a $ (2k - 1) $-spanner of expected size $ O (k n^{1 + 1/k} + k n \max_{i, u} | \mathcal{I}_i (u) |) $ with expected update time $ O ((3 + 4 \epsilon + U (n))^{k/2} \cdot T (n)) $.
If the bounds on $ \mathfrak{F} $ are amortized (worst-case), then so is the update time of $ \mathfrak{S} $.
\end{theorem}

\subsubsection{Summary of Dynamic Filtering Problem}

As we focus on the dynamic filtering in the rest of this section, we summarize the most important aspects of this problem in the following.
In a dynamic filtering algorithm we focus on a specific vertex $ u \in V_i $ at a specific level~$ i $ of the hierarchy, i.e., there will be a separate instance of the filtering algorithm for each vertex in $ V_i $.
The algorithm takes parameters $ 0 < \epsilon < 1 $ and $ a > 1 $ and fixes some choice of $ \lambda \geq g > 1 $.
It operates on the subset of edges of $ E_i $ incident on $ u $ and any vertex~$ v $ in a non-sampled cluster $ c \in C_i \setminus R_i $.
These edges are given to the filtering algorithm as a partition $ \bigcup_{c \in C_i \setminus R_i} E_i (u, c) $, where $ C_i \setminus R_i $, the set of non-sampled clusters at level~$ i $, will never change over the course of the algorithm.\footnote{Note if vertices join or leave clusters the dynamic filtering algorithm only sees updates for the corresponding edges.}
The dynamic updates to be processed by the algorithm are of two types: insertion of some edge $ (u, v) $ to some $ E_i (u, c) $, and deletion of some edge $ (u, v) $ from some $ E_i (u, c) $.
The goal of the algorithm is to maintain a partition of the \emph{clusters} into $ \lceil \log_g (n) \rceil $ buckets numbered from $ 0 $ to $ \lfloor \log_g (n) \rfloor $, a set of clusters $ \mathcal{I}_i (u) $ and a set of edges $ \mathcal{F}_i (u) $ such that conditions (F1), (F2) and (F3) are satisfied.

Condition~(F1) states that clusters in the same bucket need to have approximately the same number of edges incident on $ u $.
The ``normal'' size of $ | E_i (u, c) | $ for a cluster~$c$ in bucket~$ j $ would be~$ g^j $ and the algorithm makes sure that $ \tfrac{g^j}{\lambda} \leq | E_i (u, c) | \leq \lambda g^j $.
Thus, the ratio between the largest and the smallest value of $ | E_i (u, c) | $ among clusters~$ c $ in the same bucket is at most $ \lambda^2 $.
This value corresponds to the parameter~$ b $ in Theorem~\ref{thm:special Chernoff bound}.
The edges in $ \mathcal{F}_i (u) $ serve as a filter for the dynamic spanner algorithm in the sense that only edges in this set are passed on to level~$ i + 1 $ in the hierarchy.
Condition~(F2) states that an edge $ (u, v) $ may only be contained in $ \mathcal{F}_i (u) $ if the bucket containing the cluster of $ v $ contains at least $ \ell := 4 \gamma a \lambda^2 \tfrac{1}{\epsilon^3} n^{1/k} \ln (n) \ln(\lambda) $ clusters.
Here the choice of $ \ell $ comes from Theorem~\ref{thm:special Chernoff bound}; $ a $ is a constant that controls the error probability, $ \epsilon $ controls the amount of deviation from the mean in the Chernoff bound, and $ \gamma $ is a constant from the theorem.
Condition~(F3) states that clusters $ c $ for which some edge incident on $ u $ and $ c $ is not contained in $ \mathcal{F}_i (u) $ need to be contained in $ \mathcal{I}_i (u) $ (called \emph{inactive} clusters in~\cite{BaswanaKS12}).
Intuitively this is the case because for such clusters the spanner algorithm cannot rely on all relevant edges being present at the next level and thus has to deal with these clusters in a special way.

The goal is to design a filtering algorithm with a small value of $ \lambda $ that has small update time.
An additional goal in the design of the algorithm is to keep the number of changes performed to~$ \mathcal{F}_i (u) $ small.
A change to $ \mathcal{F}_i (u) $ after processing an update to $ E_i (u, c) $ is also called an \emph{induced update} as, in the overall dynamic spanner algorithm, such changes might appear as updates to level~$ i + 1 $ in the hierarchy, i.e., the insertion (deletion) of an edge $ (u, v) $ to (from) $ \mathcal{F}_i (u) $ might show up as an insertion (deletion) at level~$ i + 1 $.
As this update propagation takes place in all levels of the hierarchy, we would like to have a dynamic filtering algorithm that only performs $ O(1) $ changes to $ \mathcal{F}_i (u) $ per update to its input.

\subsubsection{Filtering Algorithm with Amortized Update Time}

The bound of Baswana et al.\ follows by providing a dynamic filtering algorithm with the following guarantees.

\begin{lemma}[Implicit in~\cite{BaswanaKS12}]\label{thm:filtering algorithm Baswana}
For any $ a > 1 $ and any $ 0 < \epsilon \leq \tfrac{1}{4} $, there is a dynamic filtering algorithm with amortized update time $ O (\tfrac{1}{\epsilon}) $ for which the amortized number of changes performed to~$ \mathcal{F}_i (u) $ per update to~$ E_i (u) $ is at most $ 4 + 10 \epsilon $ such that $ \mathcal{I}_i (u) \leq O (\tfrac{a}{\epsilon^7} n^{1/k} \log^2 (n)) $, i.e., $ U (n) = 4 + 10 \epsilon = O(1) $ and $ T (n) = O (\tfrac{1}{\epsilon}) $.
\end{lemma}

Note that the dynamic filtering algorithm is the only part of the algorithm by Baswana et al.\ that requires amortization.
Thus, if one could remove the amortization argument from the dynamic filtering algorithm, one would obtain a dynamic spanner algorithm with worst-case expected guarantee on the update time, which in turn could be strengthened to a worst-case high-probability guarantee.
This is exactly how we proceed in the following.

To facilitate the comparison with our new filtering algorithm, we shortly review the amortized algorithm of Baswana et al.
Their algorithm uses $ g = \lambda = \tfrac{1}{\epsilon} $ where $ \epsilon $ is a constant that is optimized to give the fastest update time for the overall spanner algorithm.
This leads to $ O (\log_g (n)) $ overlapping buckets such that all clusters in bucket~$ j $ have between $ g^{j-1} $ and $ g^j $ edges incident on~$ u $.

The algorithm does the following:
Every time the number of edges incident on $ u $ of some cluster~$ c $ in bucket~$ j $ grows to $ g^{j + 1} $, $ c $ is moved to bucket~$ j + 1 $, and every time this number falls to~$ g^{j - 1} $, $ c $ is moved to bucket~$ j - 1 $.
The algorithm further distinguishes \emph{active} and \emph{inactive} buckets such that active buckets contain at least $ \ell $ clusters and all inactive buckets contain at most $ \kappa \ell $ clusters for some constant~$ \kappa $.
An active bucket will be inactivated if its size falls to $ \ell $ and an inactive bucket will be activated if its size grows to $ \kappa \ell $.
Additionally, the algorithm makes sure that $ \mathcal{F}_i (u) $ consists of all edges incident on clusters from active buckets and that $ \mathcal{I}_i $ consists of all clusters in inactive buckets.

By employing these soft thresholds for maintaining the buckets and their activation status, Baswana et al.\ make sure that for each update to~$ E_i (u) $ the running time and the number of changes made to~$ \mathcal{F}_i (u) $ is constant.
For example, every time a cluster $ c $ is moved from bucket~$ j $ to bucket~$ j + 1 $ with a different activation status, the algorithm incurs a cost of at most $ O (g^{j + 1}) $ -- i.e., proportional to~$ | E_i (u, c) | $ -- for adding or removing the edges of $ E_i (u, c) $ to $ \mathcal{F}_i (u) $.
This cost can be amortized over at least $ g^{j + 1} - g^j = \Theta (g^{j+1}) $ insertions to $ E_i (u, c) $, which results in an amortized cost of $ O (g) = O (\tfrac{1}{\epsilon}) $, i.e., constant when $ \tfrac{1}{\epsilon} $ is constant.
Similarly, the work connected to activation and de-activation is~$ O (g) $ when amortized over $ \Theta (\ell) $ clusters joining or leaving the bucket, respectively.

\subsection{Modified Filtering Algorithm}\label{sec:worst-case filtering}

In the following, we provide our new filtering algorithm with worst-case expected update time, i.e., we prove the following theorem.
\begin{theorem}
For every $ 0 \leq i \leq k-1 $ and every $ u \in \mathcal{N}_i $, there is a filtering algorithm that has worst-case expected update time $ O (\log (n)) $ and per update performs at most $ 10.6 $ changes to $ \mathcal{F}_i (u) $ in expectation, i.e., $ U (n) = 10.6 $ and $ T (n) = O (\log (n)) $.
The maximum size of $ \mathcal{I}_i (u) $ is $ O (n^{1/k} \log^{6} (n) \log \log (n)) $.
\end{theorem}

Together with Theorem~\ref{thm:main theorem Baswana} the promised result follows.

\begin{corollary}[Restatement of Theorem~\ref{thm:spanner-expected}]
For every $ k \geq 2 $, there is a fully dynamic algorithm for maintaining a $ (2k - 1) $-spanner of expected size $ O (k n^{1 + 1/k} \log^{6} (n) \log \log (n)) $ that has worst-case expected update time $ O (14^{k/2} \log(n)) $.
\end{corollary}

We now apply the reduction of Theorem~\ref{thm:main} to maintain $ O (\log (n)) $ instances of the dynamic spanner algorithm and use the union of the maintained subgraphs as the resulting spanner.
The reduction guarantees that, at any time, one of the maintained subgraphs, and thus also their union, will indeed be a spanner and that the update-time bound holds with high probability.
\begin{corollary}[Restatement of Theorem~\ref{thm:spanner-converted}]
For every $ k \geq 2 $, there is a fully dynamic algorithm for maintaining a $ (2k - 1) $-spanner of expected size $ O (k n^{1 + 1/k} \log^{7} (n) \log \log (n)) $ that has worst-case update time $ O (14^{k/2} \log^3 (n)) $ with high probability.
\end{corollary}

\subsubsection{Design Principles}

Our new algorithm uses the following two ideas.
First, we observe that it is not necessary to keep only the edges incident from clusters of small buckets away from $ \mathcal{F}_i (u) $.
We can also, somewhat more aggressively, keep away the edges incident from the first $ \ell $ clusters of large buckets out of~$ \mathcal{F}_i (u) $.
In this way, we avoid that many updates are induced if the size of a bucket changes from small to large or vice versa.
Our modified filtering is deterministic based only on the current partitioning of the clusters into buckets and on an arbitrary, but fixed ordering of vertices, clusters, and edges.
This is a bit similar to the idea in~\cite{BodwinK16} of always keeping the ``first'' few incident edges of each vertex in the spanner.

Second, we employ a probabilistic threshold technique where, after exceeding a certain threshold on the size of the set $ E_i (u, c) $, a cluster $ c $ changes its bucket with probability roughly inverse to this size threshold.
Moving a cluster is an expensive operation that generates changes to the set of filtered edges, which the next level in the spanner hierarchy has to process as induced updates.
The idea behind the probabilistic threshold approach is that by taking a sampling probability that is roughly inverse to the number of updates induced by the move, there will only be a constant number of changes in expectation.
A straightforward analysis of this approach shows that in each bucket the size threshold will not be exceeded by a factor of more than $ O (\log (n)) $ with high probability, which immediately bounds the expected number of changes to the set of filtered edges by $ O (\log (n)) $.
By a more sophisticated analysis, taking into account the diminishing probability of not having moved up previously, we can show that exceeding the size threshold by a factor of $ 2^t $ happens with probability $ O (1/2^{e^t}) $.
Thus, the expected number of induced updates is bounded by an exponentially decreasing series converging to a constant.
A similar, but slightly more involved algorithm and analysis is employed for clusters changing buckets because of falling below a certain size threshold.

We remark that a deterministic deamortization of the filtering algorithm by Baswana et al.\ might be possible in principle without resorting to the probabilistic threshold technique, maybe using ideas similar to the deamortization in the dynamic matching algorithm of Bhattacharya et al.~\cite{BhattacharyaHN17}.
However, such a deamortization needs to solve non-trivial challenges and the other parts of the dynamic spanner algorithm would still be randomized.
Furthermore, we believe that the probabilistic threshold technique leads to a significantly simpler algorithm.
Similarly it might be possible to use the probabilistic threshold technique to emulate the less aggressive filtering of Baswana et al.\ that only filters away edges incident on large buckets.
Here, not using the probabilistic threshold technique seems the simpler choice.

\subsubsection{Setup of the Algorithm}

In our algorithm, described below for a fixed vertex $ u $, we work with an arbitrary, but fixed, order on the vertices of the graph.
The order on the vertices induces an order on the edges, by lexicographically comparing the ordered pair of incident vertices of the edges, and an order on the clusters, by comparing the respective cluster centers.
For each $ 0 \leq j \leq \lfloor \log (n) \rfloor $, we maintain a bucket by organizing the clusters in bucket~$ j $ in a binary search tree $ B_j $, employing the aforementioned order on the clusters.
Similarly, for $ 0 \leq j \leq \lfloor \log (n) \rfloor $, we organize the edges incident on $ u $ and each bucket~$ j $ in a binary search tree $ T_j $, i.e., a search tree ordering the set of edges $ \bigcup_{c \in B_j} E_i (u, c) $, where these edges are compared lexicographically as cluster-edge pairs.

We set $ \lambda = 2^{\lceil \log (4 + \ln (n)) \rceil} = O (\log (n)) $, $ \ell = 4 \gamma a \lambda^2 \tfrac{1}{\epsilon^3} n^{1/k} \ln (n) \ln (\lambda) = O (n^{1/k} \log^{3} (n) \log{\log (n)}) $ and, for every $ 0 \leq j \leq \lfloor \log (n) \rfloor $ we set $ \alpha_j = 2^j $.
Our algorithm will maintain the following invariants for every $ 0 \leq j \leq \lfloor \log (n) \rfloor $:
\begin{enumerate}
\item[(B1)] For each cluster $ c $ in bucket~$ j $, $ \tfrac{\alpha_j}{\lambda} \leq | E_i (u, c) | \leq \lambda \alpha_j $.
\item[(B2)] The edges of the first $ \ell \cdot \lambda \alpha_j $ cluster-edge pairs of $ T_j $ (or all cluster-edge pairs of $ T_j $ if there are less than $ \ell \cdot \lambda \alpha_j $ of them) are not contained in $ \mathcal{F}_i (u) $ and the remaining edges of $ T_j $ are contained in $ \mathcal{F}_i (u) $.
\item[(B3)] The first $ 1 + \lambda^2 \ell $ clusters of $ B_j $ are contained in $ \mathcal{I}_i (u) $ and the remaining clusters of $ B_j $ are not contained in $ \mathcal{I}_i (u) $.
\end{enumerate}

Observe that invariant~(B1) is equal to condition~(F1) and that invariant~(B3) immediately implies the claimed bound on $ \mathcal{I}_i (u) $ as there are $ O (\log (n)) $ buckets, each contributing $ O (\lambda^2 \ell) $ clusters.

Furthermore, the invariants also imply correctness in terms of conditions (F2) and~(F3) because of the following reasoning:
For condition~(F2), let $ (u, v) \in \mathcal{F}_i (u) $ and let $ c $ denote the cluster of~$ v $.
Then, by invariant~(B2), there are at least $ \ell \cdot \lambda \alpha_j $ cluster-edge pairs contained in $ T_j $ that are lexicographically smaller than the pair consisting of $ c $ and $ (u, v) $.
As each cluster in bucket~$ j $ has at most $ \lambda \alpha_j $ edges incident on $ u $ by invariant~(B1), it follows that there are at least $ \ell $ clusters contained in bucket~$ j $ as otherwise $ T_j $ could not contain at least $ \ell \cdot \lambda \alpha_j $ cluster-edge pairs.

For condition~(F3), let $ (u, v) \in E_i (u) \setminus \mathcal{F}_i (u) $ and let $ c $ denote the cluster of~$ v $.
Then the pair consisting of $ c $ and $ (u, v) $ must be among the first $ \ell \cdot \lambda \alpha_j $ entries of~$ T_j $ by invariant~(B2).
As each cluster in bucket~$ j $ has at least $ \tfrac{\alpha_j}{\lambda} $ edges incident on $ u $ by invariant~(B1), there are thus at most $ \tfrac{\lambda \alpha_j}{\alpha_j / \lambda} \ell = \lambda^2 \ell $ clusters in bucket~$ j $ that are smaller than $ c $ in terms of the chosen ordering on the clusters.
It follows that $ c $ must be among the first $ 1 + \lambda^2 \ell $ clusters of $ B_j $ and by invariant~(B3) is thus contained in $ \mathcal{I}_i (u) $ as required by condition~(F1).

\subsubsection{Modified Bucketing Algorithm}

The algorithm after an update to some edge $ (u, v) $ is as follows, where we denote the unique cluster of $ v $ by $ c $:
\begin{itemize}
	\item If the edge $ (u, v) $ was inserted, check if one of the following cases applies:
	\begin{itemize}
		\item If $ | E_i (u, c) | = 1 $ after the insertion (i.e., $ c $ becomes a neighbor of $ u $), then move $ c $ into bucket~$ 0 $ by performing the following steps:
		\begin{enumerate}
			\item Add $ c $ to $ B_0 $.
			\item Add $ (u, v) $ to $ T_0 $.
		\end{enumerate}
		\item If $ | E_i (u, c) | \geq 2 \alpha_j $ after the insertion, where $ j $ is the number $ c $'s current bucket, do the following: Flip a biased coin that is ``heads'' with probability $ \min (\tfrac{1}{\alpha_j}, 1) $. If the coin shows ``heads'' or if $ | E_i (u, c) | = \lambda \cdot \alpha_j $, then move cluster~$ c $ up to bucket~$ j' = \lceil \log (| E_i (u, c) |) \rceil $ by performing the following steps:
		\begin{enumerate}
			\item Remove $ c $ from $ B_j $ and add it to $ B_{j'} $.
			\item Remove all edges of $ E_i (u, c) $ from $ T_j $ and add them to $ T_{j'} $.
		\end{enumerate}
	\end{itemize}
	\item If the edge $ (u, v) $ was deleted, check if one of the following cases applies:
	\begin{itemize}
		\item If $ | E_i (u, c) | = 0 $ after the deletion (i.e., $ c $ ceases to be a neighbor of $ u $), then move $ c $ out of bucket~$ 0 $ by performing the following steps:
		\begin{enumerate}
			\item Remove $ c $ from $ B_0 $.
			\item Remove $ (u, v) $ from $ T_0 $.
		\end{enumerate}
		\item If $ | E_i (u, c) | \leq \tfrac{\alpha_j}{2} $ after the deletion, where $ j $ is the number $ c $'s current bucket, do the following: Flip a biased coin that is ``heads'' with probability $ \min (\tfrac{2^{2 t + 1}}{\alpha_j}, 1) $ for the maximum $ t \geq 1 $ such that $ | E_i (u, c) | \leq \tfrac{\alpha_j}{2^t} $ (i.e., $ t = \lfloor \log (\tfrac{\alpha_j}{| E_i (u, c) |}) \rfloor $). If the coin shows ``heads'' or if $ | E_i (u, c) | = \tfrac{\alpha_j}{\lambda} $, then move cluster~$ c $ down to bucket~$ j' = \lfloor \log (| E_i (u, c) |) \rfloor $ by performing the following steps:
		\begin{enumerate}
			\item Remove $ c $ from $ B_j $ and add it to $ B_{j'} $.
			\item Remove all edges of $ E_i (u, c) $ from $ T_j $ and add them to $ T_{j'} $.
		\end{enumerate}
	\end{itemize}
\end{itemize}
Additionally, invariants (B2) and (B3) are maintained in the trivial way by making the necessary changes to $ \mathcal{F}_i (u) $ after a change to $ T_j $ and to $ \mathcal{I}_i $ after a change to $ B_j $, respectively.
Furthermore, invariant~(B1) is satisfied because the following invariant (B1') holds as well for every $ 0 \leq j \leq \lfloor \log (n) \rfloor $ by the design of the algorithm:
\begin{enumerate}
\item[(B1')] Whenever a cluster $ c $ moves to bucket~$ j $, $ \tfrac{\alpha_j}{2} <| E_i (u, c) | < 2 \alpha_j $.
\end{enumerate}

\subsubsection{Analysis of Induced Updates and Running Time}

We now analyze the update time and the number of changes to~$ \mathcal{F}_i (u) $ per update to some~$ E_i (u, c) $ for some cluster~$ c $.
These changes are also called induced updates.

If by the update $ c $ becomes a neighbor of $ u $, then one cluster is added to $ B_0 $ and one edge is added to $ T_0 $.
Similarly, if by the update $ c $ ceases to be a neighbor $ u $, then one cluster is removed from $ B_0 $ and one edge is removed from $ T_0 $.
Clearly, both of these cases lead to at most $ 2 $ changes to~$ \mathcal{F}_i (u) $ and a running time of $ O (\log n) $ as both $ B_0 $ and $ T_0 $ are organized as binary search trees.

Now observe that each other type of update causes at most one move of $ c $ from some bucket~$ j $ to some other bucket~$ j' $.
Each such move can be processed in time $ O (| E_i (u, c) | \log{n}) $ as only the cluster~$ c $ is moved from some binary search tree~$ B_j $ to another binary search tree~$ B_{j'} $ and correspondingly only $ | E_i (u, c) | $ cluster-edges pairs are moved from the binary search tree~$ T_j $ to the binary search tree~$ T_{j'} $.
To analyze the number of changes to~$ \mathcal{F}_i (u) $ consider the following case distinction for some cluster-edge pair $ (c, (u, v)) $ being moved from~$ T_j $ to~$ T_{j'} $:
\begin{itemize}
\item If $ (c, (u, v)) $ is among the first $ \ell \cdot \lambda \alpha_j $ cluster-edge pairs of $ T_j $ before being removed from~$ T_j $ and is among the first $ \ell \cdot \lambda \alpha_{j'} $ cluster-edge pairs of $ T_{j'} $ after being added to~$ T_{j'} $, then $ (u, v) $ is neither contained in $ \mathcal{F}_i (u) $ before nor after the move.
Furthermore, at most one cluster-edge pair might start being among the first $ \ell \cdot \lambda \alpha_j $ pairs in $ T_j $ (resulting in the removal of the corresponding edge from~$ \mathcal{F}_i (u) $) and at most one cluster-edge pair might stop being among the first $ \ell \cdot \lambda \alpha_{j'} $ pairs in $ T_{j'} $ (resulting in the addition of the corresponding edge to~$ \mathcal{F}_i (u) $).
Thus, we perform at most $ 2 $ changes to~$ \mathcal{F}_i (u) $ for moving $ (c, (u, v)) $.
\item If $ (c, (u, v)) $ is among the first $ \ell \cdot \lambda \alpha_j $ cluster-edge pairs of $ T_j $ before being removed from~$ T_j $ and is not among the first $ \ell \cdot \lambda \alpha_{j'} $ cluster-edge pairs of $ T_{j'} $ after being added to~$ T_{j'} $, then $ (u, v) $ is not contained in $ \mathcal{F}_i (u) $ before the move, but it is contained in $ \mathcal{F}_i (u) $ after the move.
Furthermore, at most one cluster-edge pair might start being among the first $ \ell \cdot \lambda \alpha_j $ pairs in $ T_j $ (resulting in the removal of the corresponding edge from~$ \mathcal{F}_i (u) $) and no cluster-edge pair will stop being among the first $ \ell \cdot \lambda \alpha_{j'} $ pairs in $ T_{j'} $.
Thus, we perform at most $ 2 $ changes to~$ \mathcal{F}_i (u) $ for moving $ (c, (u, v)) $.
\item If $ (c, (u, v)) $ is not among the first $ \ell \cdot \lambda \alpha_j $ cluster-edge pairs of $ T_j $ before being removed from~$ T_j $ and is among the first $ \ell \cdot \lambda \alpha_{j'} $ cluster-edge pairs of $ T_{j'} $ after being added to~$ T_{j'} $, then $ (u, v) $ is contained in $ \mathcal{F}_i (u) $ before the move, but it is not contained in $ \mathcal{F}_i (u) $ anymore after the move.
Furthermore, no cluster-edge pair will start being among the first $ \ell \cdot \lambda \alpha_j $ pairs in $ T_j $ and at most one cluster-edge pair might stop being among the first $ \ell \cdot \lambda \alpha_{j'} $ pairs in $ T_{j'} $ (resulting in the addition of the corresponding edge to~$ \mathcal{F}_i (u) $).
Thus, we perform at most $ 2 $ changes to~$ \mathcal{F}_i (u) $ for moving $ (c, (u, v)) $.
\item If $ (c, (u, v)) $ is not among the first $ \ell \cdot \lambda \alpha_j $ cluster-edge pairs of $ T_j $ before being removed from~$ T_j $ and is not among the first $ \ell \cdot \lambda \alpha_{j'} $ cluster-edge pairs of $ T_{j'} $ after being added to~$ T_{j'} $, then $ (u, v) $ is contained in $ \mathcal{F}_i (u) $ before and after the move.
Furthermore, no cluster-edge pair will start being among the first $ \ell \cdot \lambda \alpha_j $ pairs in $ T_j $ and no cluster-edge pair will stop being among the first $ \ell \cdot \lambda \alpha_{j'} $ pairs in $ T_{j'} $.
Thus, we perform no changes to~$ \mathcal{F}_i (u) $ for moving $ (c, (u, v)) $.
\end{itemize}
Thus, for each move of a cluster $ c $, we incur at most $ 2 | E_i (u, c) | $ changes to~$ \mathcal{F}_i (u) $.

For technical reasons, we go on by giving slightly different analyses for the cases of moving up and moving down.

\paragraph{Moving Up.}

For every integer $ 1 \leq t \leq \log (\lambda) - 1 $, let $ p_t $ be the probability that $ 2^t \alpha_j \leq | E_i (u, c) | < 2^{t + 1} \alpha_j $ when $ c $~is moved up and let $ q $ be the probability that $ | E_i (u, c) | = \lambda \cdot \alpha_j $ when $ c $~is moved up.
Note that this covers all events for $ c $ being moved up.
As observed above, each move induces at most $ 2 | E_i (u, c) | $ updates, where $ | E_i (u, c) | < 2^{t + 1} \alpha_j $ with probability $ p_t $ and $ | E_i (u, c) | \leq n $ in any case.
Thus, by the law of total expectation, the expected number of induced updates per insertion to~$ E_i (u, c) $ is at most
\begin{equation*}
\sum_{1 \leq t \leq \log (\lambda) - 1} p_t \cdot 2 \cdot 2^{t + 1} \alpha_j + q \cdot 2 n \, .
\end{equation*}

We now bound $ p_t $, the probability that $ 2^t \alpha_j \leq | E_i (u, c) | < 2^{t + 1} \alpha_j $ when $ c $ is moved up.
As soon as $ | E_i (u, c) | $ exceeds the threshold $ 2 \alpha_j $, each insertion makes $ c $ move up with probability $ \tfrac{1}{\alpha_j} $ (when the biased coin shows ``heads'').
For $ t = 1 $, we clearly have $ p_t \leq \tfrac{1}{\alpha_j} $.
For $ 2 \leq t \leq \log (\lambda) - 1 $, $ p_t $ is determined by one ``heads'' preceded by at least $ 2^t \alpha_j - 2 \alpha_j $ ``tails'' in the coin flips of previous insertions to $ E_i (u, c) $, i.e., $ p_t $ is bounded by
\begin{equation*}
p_t \leq \frac{1}{\alpha_j} \cdot \left( 1 - \frac{1}{\alpha_j} \right)^{(2^t - 2) \cdot \alpha_j} \leq \frac{1}{\alpha_j} \cdot \frac{1}{e^{2^t - 2}} \, .
\end{equation*}
Here we use the inequality $ (1 - \tfrac{1}{x})^x \leq \tfrac{1}{e} $, where $ e $ is Euler's constant.
Similarly, $ q $, the probability that $ | E_i (u, c) | = \lambda \alpha_j $ with $ \lambda = 2^{\lceil \log (4 + \ln (n)) \rceil} $ when $ c $ is moved up, is determined by $ \alpha_j (\lceil a \ln (n) \rceil + 2) - 2 \alpha_j $ ``tails''.
Thus, $ q $ is bounded by
\begin{equation*}
q \leq \frac{1}{e^{2^{\lceil \log (4 + \ln (n)) \rceil} - 2}} \leq \frac{1}{e^{2 + \ln (n)}} = \frac{1}{e^2 n} \, .
\end{equation*}

We can now bound the expected number of induced updates by
\begin{align*}
\sum_{1 \leq t \leq \log (\lambda) - 1} p_t \cdot 2 \cdot 2^{t + 1} \alpha_j + q \cdot 2 n &= \frac{1}{\alpha_j} \cdot 2 \cdot 2^2 \alpha_j + \sum_{2 \leq t \leq \log (\lambda) - 1} \frac{1}{\alpha_j} \cdot \frac{1}{e^{2^t - 2}} \cdot 2 \cdot 2^{t + 1} \alpha_j + \frac{1}{e^2 n} \cdot 2 n \\
	&\leq 8 + 4 \cdot \sum_{2 \leq t < \infty} \frac{2^t}{e^{2^t - 2}} + 0.28 \\
	&\leq 8 + 4 \cdot 0.57 + 0.28 \\
	&\leq 10.6 \, .
\end{align*}

\paragraph{Moving Down.}
 
For every $ 1 \leq t \leq \log (\lambda) - 1 $, let $ p_t $ be the probability that $ \tfrac{\alpha_j}{2^{t + 1}} < | E_i (u, c) | \leq \tfrac{\alpha_j}{2^t} $ when $ c $ is moved down and let $ q $ be the probability that $ | E_i (u, c) | = \tfrac{\alpha_j}{\lambda} $ when $ c $ is moved down.
As observed above, each move induces at most $ 2 | E_i (u, c) | $ updates and thus, by the law of total expectation, the expected number of induced updates per deletion from $ E_i (u, c) $ is at most
\begin{equation*}
\sum_{1 \leq t \leq \log (\lambda) - 1} p_t \cdot 2 \frac{\alpha_j}{2^t} + q \cdot 2 n \, .
\end{equation*}


We now bound $ p_t $, the probability that $ \tfrac{\alpha_j}{2^{t + 1}} < | E_i (u, c) | \leq \tfrac{\alpha_j}{2^t} $ when $ c $ is moved down.
For $ t = 1 $, we clearly have $ p_1 \leq \tfrac{2^3}{\alpha_j} = \tfrac{8}{\alpha_j} $ as this is the probability that just a single coin flip made the cluster move down.
For $ 2 \leq t \leq \log (\lambda) - 1 $, observe that for $ | E_i (u, c) | \leq \tfrac{\alpha_j}{2^t} $ to hold, there must have been at least $ t - 1 $ subsequences of deletions such that after every deletion in subsequence $ s $ we had $ \tfrac{\alpha_j}{2^{s + 1}} < | E_i (u, c) | \leq \tfrac{\alpha_j}{2^s} $ (where $ 1 \leq s \leq t - 1 $).
Observe that the $s$-th subsequence consists of $ d_s := \tfrac{\alpha_j}{2^s} - \tfrac{\alpha_j}{2^{s + 1}} = \tfrac{\alpha_j}{2^{s + 1}} $ many deletions.
Remember that during the $ s $-th subsequence the probability of $ c $ moving down is $ r_s := \min (\tfrac{2^{2 s + 1}}{\alpha_j}, 1) $.
If $ r_s < 1 $, then by the inequality $ (1 - \tfrac{1}{x})^x \leq \tfrac{1}{e} $ we have
\begin{equation*}
(1 - r_s)^{d_s} = \left( 1 - \frac{2^{2 s + 1}}{\alpha_j} \right)^{\frac{\alpha_j}{2^{s + 1}}} = \left( 1 - \frac{2^{2 s + 1}}{\alpha_j} \right)^{\frac{\alpha_j}{2^{2s + 1}} \cdot 2^{s}} \leq \frac{1}{e^{2^s}} \, .
\end{equation*}
In the other case, $ r_s = 1 $, we clearly have $ (1 - r_s)^{d_s} = 0 \leq \frac{1}{e^{2^s}} $.
Now $ p_t $ is determined by one ``heads'' preceded by at least $ d_s $ ``tails'' for each subsequence $ s $, i.e., $ p_t $ is bounded by
\begin{align*}
p_t &\leq \frac{2^{2t + 1}}{\alpha_j} \cdot \prod_{1 \leq s \leq t - 1} \left( 1 - r_s \right)^{d_s} \\
	&\leq \frac{2^{2t + 1}}{\alpha_j} \cdot \prod_{1 \leq s \leq t - 1} \frac{1}{e^{2^s}} \\
	&= \frac{2^{2t + 1}}{\alpha_j} \cdot \frac{1}{e^{\sum_{1 \leq s \leq t - 1} 2^s}} \\
	&= \frac{2^{2t + 1}}{e^{2^t - 2} \alpha_j} \, .
\end{align*}

Similarly, $ q $, the probability that $ | E_i (u, c) | = \tfrac{\alpha_j}{\lambda} $ with $ \lambda = 2^{\lceil \log (4 + \ln (n)) \rceil} $ when $ c $ is moved down, is bounded by
\begin{equation*}
q \leq \frac{1}{e^{2^{\lceil \log (4 + \ln (n)) \rceil} - 2}} \leq \frac{1}{e^{2 + \ln (n)}} = \frac{1}{e^2 n} \, .
\end{equation*}

We can now bound the expected number of induced updates by
\begin{align*}
\sum_{1 \leq t \leq \log (\lambda) - 1} p_t \cdot 2 \cdot \frac{\alpha_j}{2^t} + q \cdot 2 n &= p_1 \cdot 2 \cdot \frac{\alpha_j}{2} + \sum_{2 \leq t \leq \log (\lambda) - 1} p_t \cdot 2 \cdot \frac{\alpha_j}{2^t} + q \cdot 2 n \\
	&\leq \frac{8}{\alpha_j} \cdot 2 \cdot \frac{\alpha_j}{2} + \sum_{2 \leq t \leq \log (\lambda) - 1} \frac{2^{2t + 1}}{e^{2^t - 2} \alpha_j} \cdot 2 \cdot \frac{\alpha_j}{2^t} + \frac{1}{e^2 n} \cdot 2 n \\
	&\leq 8 + 4 \cdot \sum_{2 \leq t < \infty} \frac{2^t}{e^{2^t - 2}} + 0.28 \\
	&\leq 8 + 4 \cdot 0.57 + 0.28 \\
	&\leq 10.6 \, .
\end{align*}
This concludes the proof that the expected number of induced updates is at most $ 10.6 $.

\section{Dynamic Maximal Matching with Worst-Case Expected Update Time}
\label{sec:matching}

\newcommand{\hh}{\mathcal{H}}
\renewcommand{\aa}{\mathcal{A}}
\newcommand{\qstar}{q*}
\newcommand{\sigmagap}{\sigma^{\textrm{gap}}}
\newcommand{\sigmalevel}{\sigma^{\textrm{level}}}
\newcommand{\sigmadummy}{\sigma^{\textrm{dummy}}} 
\newcommand{\sigmapivot}{\sigma^{\textrm{pivot}}}
\newcommand{\elemma}{\EE^{\textrm{lemma}}}
\newcommand{\epivot}{\EE^{\textrm{pivot}}}
\newcommand{\ebad}{\EE^{\textrm{no-pivot}}}
\newcommand{\esettle}{\EE^{\textrm{settle}}}
\newcommand{\responsible}{\textrm{{\sc responsible}}}

In this section we turn to proving Theorem~\ref{thm:matching-expected}\footnote{The correctness proof had to be significantly extended due to a mistake in our SODA 2019 paper.}. We achieve our result
by modifying the algorithm of Baswana \etal~\cite{BaswanaGS18}, which achieves \emph{amortized} expected time $O(\log(n))$.
We start by describing the original algorithm of Baswana \etal, and then 
discuss why their algorithm does not provide a worst-case expected guarantee, and the modifications
we make to achieve this guarantee. Throughout this section, we define a vertex to be \emph{free} if it is not matched,
and we define $\mate(v)$, for matched $v$, to be the vertex that $v$ is matched to.

\subsection{The Original Matching Algorithm of Baswana \etal}
\label{subsec:matching-baswana}

\paragraph{High-Level Overview.} Let us consider the trivial algorithm for maintaining a maximal matching. Insertion of an edge $(u,v)$ is easy to handle in $O(1)$ time: if $u$ and $v$ are both free then we add the edge to the matching; otherwise, we do nothing. Now consider deletion of an edge $(u,v)$. If
$(u,v)$ was not in the matching then the current matching remains maximal, so there is nothing to be done and the update time is only $O(1)$.
If $(u,v)$ was in the matching, then both $u$ and $v$ are now free and must scan all of their neighbors looking for a new neighbor to match to.
The update time is thus $ O(\mmax{\degree(u)}{\degree(v)})$. This is the only expensive operation.

At a very high level, the idea of the Baswana \etal\ algorithm is to create a hierarchy of the vertices (loosely) according to their degrees. High degree vertices
are more expensive to handle. To counterbalance this, the algorithm ensures that when a high degree vertex $v$ picks a new mate, it chooses
that mate \emph{at random} from a large number of neighbors of $v$. Thus, although the deletion of the matching edge $(v,\mate(v))$ will be expensive,
there is a high probability that the adversary will first have to delete many non-matching $(v,w)$ (which are easy to process) before it finds $(v,\mate(v))$. (Recall that the algorithm of Baswana \etal and our modification both assume an oblivious adversary).

\paragraph{Setup of the Algorithm.}
Let $\textsc{l}_0 = \lfloor \log_{4} (n) \rfloor$.

\begin{itemize}

\item Each edge $(u,v)$ will be \emph{owned} by exactly one of its endpoints. 
Let $\OO_v$ contain all edges owned by $v$. Loosely speaking, 
if $(u,v) \in \OO_v$ then $v$ is responsible for telling $u$ about any changes in its status (e.g.\ $v$ becomes unmatched or changes levels in the hierarchy), but not vice versa. 

\item  The algorithm maintains a partition of the vertices into $\textsc{l}_0 +2$ \emph{levels},
numbered from $-1$ to $\textsc{l}_0$. 
During the algorithm, when a vertex moves to level $i$, it owns at least $4^i$ edges. Level $-1$ then contains the vertices that own no edges.
The algorithm always maintains the invariant that if $\LL(u) < \LL(v)$ then edge $(u,v) \in \OO_v$.

\item For every vertex $u$, the algorithm stores a dynamic hash table of the edges in $\OO_u$. 
The algorithm also maintains the following list of edges for $u$: for each $i\ge\LL(u)$,
let $\mathcal{E}_u^i$ be the set of all those edges incident on
$u$ from vertices at level $i$ that are not owned by $u$. 
The set $\EE_u^i$ will be maintained in a dynamic hash table.
However, the onus of maintaining $\EE_u^i$ will not be on~$u$, because these edges are by definition not owned by $u$.
For example, if a neighbor~$v$ of~$u$ moves from level $i > \LL(u)$ to level $j > i$, then $v$
will remove $(u,v)$ from~$\EE_u^i$ and insert it to~$\EE_u^j$. 
\end{itemize}

\paragraph{Invariants and Subroutines.}
Define $\below{j}(v)$ to contain all neighbors of $v$ strictly below level~$j$ and $\equal{j}(v)$ to contain all neighbors of $v$ at level exactly $j$. The key invariant of the hierarchy is that a vertex moves up to a higher level in the hierarchy (via what we call a \Rise operation) 
it will have
sufficiently many neighbors below it. For $j > \LL(v)$, define $\phi_v(j) = |\below{j}(v)|$, and $\phi_v(j) = 0$ otherwise
\footnote{Baswana et al.~gave an equivalent definition  in terms of the $\OO_v$ and $\EE_v^i$ structures.}.
We now describe some guarantees of the Baswana \etal\ algorithm. Note that the hierarchy only maintains an upper bound on $\below{j}(v)$ (Invariant~3),
{\em not} a lower bound;
a lower bound on $\below{j}(v)$ only comes into play when $v$ picks a new matching edge (Matching Property).
More specifically, right before a mate is randomly selected for a node $v$ on level $j$ the algorithm makes sure that $|\below{j}(v)| \ge 4^j$

\begin{itemize}
\item \textbf{Invariant~1:} Each edge is owned by exactly one endpoint, and if the endpoints of the edge are at different levels, the edge
is owned by the endpoint at higher level. (If the two endpoints are at the same level, then the tie is broken appropriately by the algorithm.)

\item \textbf{Invariant~2:} Every vertex at level $\ge 0$ is matched and every vertex at level $-1$ is
free.

\item \textbf{Invariant~3:} For each vertex $v$ and for all
$j > \LL(v)$, $\phi_v(j) < 4^{j}$ holds true.

\item \textbf{Invariant~4:}  Both endpoints of a matched edge are at the same level.

\item \textbf{Matching Property:}  If a vertex $v$ at level $j > -1$ is (temporarily) unmatched, the algorithm proceeds as follows: if $|\below{j}(v)| \geq 4^{j}$, $v$ picks a new mate \emph{uniformly at random} from $\below{j}(v)$; If $|\below{j}(v)| < 4^{j}$, then $v$ falls to level $j-1$ and is recursively processed there (i.e. depending on the size of 
$\below{j-1}(v)$, $v$ either picks a random mate from $\below{j-1}(v)$ or continues to fall.)
\end{itemize}
Invariant~1 and~3 combined imply that $|\OO_v| \leq \phi_v(j+1)\leq 4^{\LL(v) + 1} = O(4^{\LL(v)})$ and $\equal{\LL(v)}(v) \leq 4^{\LL(v) + 1} = O(4^{\LL(v)})$ if $ \LL (v) < \textsc{l}_0 $.
For $ \LL (v) = \textsc{l}_0 $, $ 4^{\textsc{l}_0 + 1} \geq n $, so trivially $ |\OO_v| = O(4^{\LL(v)})$ and $\equal{\LL(v)}(v) = O(4^{\LL(v)})$.
\begin{remark}
Observe that if we maintain these invariants then we always have a maximal matching:
By Invariant~1, each edge $ e $ is owned by exactly one endpoint $ v $.
As by Invariant~3 a vertex at level~$-1$ owns no edges, the endpoint $ v $ is at level $ \geq 0 $, and by Invariant~2, $ v $ must be matched.
Thus, every edge has an endpoint that is matched.
\end{remark}

We now consider the procedures used by the algorithm of Baswana \etal\ to maintain the hierarchy and the maximal matching. The bulk of the work is in maintaining $\OO_v$, $\EE_v^j$, and $\phi_v(j)$, which change due to external additions and deletions of edges, and also due to the algorithm internally moving vertices in the hierarchy to satisfy the invariants above. We largely stick to the notation of the original paper, but we omit details that remain entirely unchanged in our approach. See Section~4 in~\cite{BaswanaGS18} for the original algorithm description (and its analysis).

\begin{itemize}
\item \IncrementPhi{$v,i$} increases $\phi_v(i)$ by one, whereas \DecrementPhi{$v,i$} decreases it. (The paper of Baswana et al. instead called the increment function Increment-$\phi$, but we choose $\IncrementPhi$ because it better fits the details of our algorithm.) Note that \IncrementPhi{$v,i$} might trigger a call to \Rise{$v,i,j$} and the logic that we use for triggering this call differs from that in~\cite{BaswanaGS18} as we also have probabilistic rises, see below.

\item \Rise{$v,i,j$} (new notation) moves a vertex $ v $ from level $i$ to level $j$. This results in changes to many of the $\OO$ and~$\EE$ lists.
In particular, $v$ takes ownership of all edges $(v,w)$ with $w \in \below{j}(v)$. Moreover, for any vertex $w \in \below{i}(v)$, edge $(v,w)$
is removed from  $\EE_w^i$, and for every $w \in \below{j}(v)$, edge $(v,w)$ is added to $\EE_w^j$. As a result, the algorithm runs
\DecrementPhi{$w,k$} for every $w \in \below{j}(v)$, and every $i < k \leq j$. A careful analysis bounds the total amount of bookkeeping work
at $O(4^j)$ (see Lemma~\ref{lem:rise}).

\item \Fall{$v,i$} (new notation) moves $v$ from level $i$ to level $i-1$. 
As above this leads to bookkeeping work: $\OO_w$, 
$\EE_w^i$, and $\EE_w^{i-1}$  change for many neighbors of $w$ of $v$. Note that only edges $(v,w)$ 
previously owned by $v$ are affected, so by Invariant~3, the total
amount of bookkeeping work is at most $|\OO_v| = O(4^i)$. 

The algorithm must also do \IncrementPhi{$w,i$} for every $w$ that was previously in $\below{i}(v)$, incrementing
$\phi_w(i)$. Such an increment might result in $w$
violating Invariant~3 (if $\phi_w(i)$ goes from $4^i -1$ to $4^i$), in which case the algorithm executes \Rise{$w, \LL(w), i$)}. Moreover, if $w'$ was the previous mate of $w$, then edge $(w,w')$ is removed from the matching to preserve invariant~4, so the algorithm must also execute
\FixFreeVertex{$w$} and \FixFreeVertex{$w'$} (see below), which can in turn lead to more calls to \Fall and \Rise. One of the main tasks of the analysis is to bound this cascade.

\item \FixFreeVertex{$v$} handles the case when a vertex $v$ is unmatched; this can happen because the matching edge  incident to $v$ was deleted,
or because $v$ newly rose/fell to level $i$, where $i = \LL(v)$. Following the Matching Property, if $|\below{i}(v)| < 4^{i}$, then the algorithm executes \Fall{$v,i$}, followed by \FixFreeVertex{$v$}. On the other hand, if $|\below{i}(v)| \geq 4^{i}$, then $v$ remains at level $i$ and picks a new mate by executing \GenericRandomSettle{$v,i$}.

\item \GenericRandomSettle{$v,i$} finds a new mate $w$ for a vertex $v$ at level $i$ assuming that $|\below{i}(v)| \geq 4^{i}$. The algorithm
picks $w$ uniformly at random from $\below{i}(v)$. Let $\ell = \LL(w) < i$. The algorithm first does \Rise{$w,\ell,i$} (to satisfy Invariant~4),
and then matches $v$ to $w$. Note that if $\ell \neq -1$, then $w$ had a previous mate $w'$ which is now unmatched, so the algorithm now does
\FixFreeVertex{$w'$}.
\end{itemize}

\paragraph{Handling Edge Updates.}

We now show how the algorithm maintains the invariants under edge updates.
First consider the insertion of edge $(u,v)$. Say w.l.o.g.\ that $\LL(v) \geq \LL(u)$. Then $(u,v)$ is added to $\OO_v$ and to $\EE_u^{\LL(v)}$. 
The algorithm must then execute \IncrementPhi{$u,j$} and \IncrementPhi{$v,j$} for every $j > \LL(v)$. This takes time $O(\log(n))$ and might additionally result in some level $\ell$ for
which $\phi_v(\ell) \geq 4^\ell$ (or $\phi_u(\ell) \geq 4^\ell$), in which case Invariant~3 is violated so the algorithm performs \Rise{$v,\LL(v),\ell$} (or \Rise{$u,\LL(u),\ell$}). If $\phi_v(\ell) \geq 4^\ell$ for multiple levels $\ell$, then $v$ rises to the highest such $\ell$.

Now consider the deletion of an edge $(u,v)$ with $\LL(v) \geq \LL(u)$. The algorithm first does $O(\log(n))$ work of simple bookkeeping: 
it removes $(u,v)$ from $\OO_v$ and $\EE_u^{\LL(v)}$, and executes the corresponding calls to \DecrementPhi.
If $(u,v)$ was \emph{not} a matching edge, the work ends there: unlike with \IncrementPhi, the procedure \DecrementPhi cannot lead to the violation of any invariants. By contrast, the most expensive operation is the deletion of a \emph{matched} edge $(u,v)$, because the algorithm must execute \FixFreeVertex{$u$}, and \FixFreeVertex{$v$}.

\paragraph{Analysis Sketch.}

Whereas our final algorithm is very similar to the original algorithm of Baswana \etal, our analysis is mostly different, so we only provide a brief
sketch of their original analysis. The basic idea is that because a vertex $v$ is only responsible for edges in $\OO_v$, processing a vertex at level $i$ takes time $O(4^{i+1})$ (Invariant~3). The crux of the analysis is in arguing that vertices at high level are processed less often. 
There are two primary ways a vertex $v$ can be processed at level $i$. 
\textbf{1)} $v$ rises to level $i$ because $\phi_v(i)$ goes from  $4^i - 1$ to $4^i$. This does not happen often because many \IncrementPhi{$v,i$} are required to reach such a high $\phi_v(i)$. 
\textbf{2)} the matching edge $(v,\mate(v))$ is deleted from the graph. This does not happen often because by Matching Property, $v$ originally picks its mate at random from at least $4^i$ options, so since the adversary is oblivious, it will in expectation delete many non-matching edges $(v,w)$ (which are easy to process) before it hits upon $(v, \mate(v))$.

\subsection{Our Modified Algorithm}
\label{subsec:modification}
Recall the definition of $\phi_v(j)$ for any vertex $v$ with $i = \LL(v)$ and level $j$:
\[ 
\phi_v(j) = \left\{ 
\begin{array}{ll}
|\below{j}(v)| & \mbox{~if~} j>i \\
0 & \mbox{otherwise.}
\end{array}
\right.
\]

There are two reasons why the original algorithm of Baswana \etal\ does not guarantee a worst-case expected update time. 

\textbf{1:} The algorithm uses a hard threshold for $\phi_v(i)$: the update which increases $\phi_v(i)$ 
from $4^{i} - 1$ to~$4^i$ is guaranteed to lead to the expensive execution of \Rise{$v, \LL(v), i$}. Thus, while their algorithm guaranteed that 
overall few updates lead to this expensive event, it is not hard to construct an 
update sequence which forces one particular update to be an expensive one. To overcome this, we use a randomized
threshold, where every time $\phi_v(i)$ increases, $v$ rises to level $i$ with probability $\Theta(\log(n)/4^i)$. 

\textbf{2:} Consider the deletion of an edge $(u,v)$ where $i = \LL(v) \geq \LL(u)$. 
Baswana \etal\ showed that this deletion takes time $O(\log(n))$ if $u \neq \mate(v)$, and time $O(4^{i})$ if $u = \mate(v)$.
At first glance this seems to lead to an expected-worst-case guarantee: we know by the Matching Property that $v$ picked its mate at random
from a set of at least $4^i$ vertices, so if we could argue that for any edge $(u,v)$ we always have $\PP[\mate(v) = u] \leq 1/4^{i}$,
then the expected time to process \emph{any} deletion would be just $O(\log(n))$. 

Unfortunately, in the original algorithm it is \emph{not} the case that $\PP[\mate(v) = u] \leq 1/4^i$.
To see this, consider the following star graph with center $v$. In the sequence of updates to the edges incident to a vertex $v$, in which $v$ will be always at level~$i$, 
every updated edge $(u,v)$ will have $\LL(u) < \LL(v)$, and $|\below{i}(v)|$ will always be between $4^i$ and $2 \cdot 4^i$.
The other vertices in the sequence are $v'$, $x_1, x_2, \ldots, x_{4^{i}-1}$ and $y_1, y_2, \ldots, y_{4^{i}-1}$.
At the beginning, $v$ has an edge to $v'$ and to all the $x_i$. The update sequence repeats the following cyclical process for very many rounds: insert an edge to every $y_i$, delete the edge to every $x_i$, insert an edge to every $x_i$, delete the edge to every $y_i$, insert the edge to every $y_i$, and so on.
Note that the edge from $v$ to $v'$ is never deleted. We claim that as we continue this process for a long time, $\PP[\mate(v) = v'] \rightarrow 1$.
The reason is that the algorithm of Baswana \etal\ only picks a new mate for~$v$ when the previous matching edge was deleted. But the process
repeatedly deletes all edges except $(v,v')$, so it will continually pick a new matching edge at random until it eventually picks $(v,v')$,
at which point $v'$ will remain the mate of~$v$ throughout the process. The original algorithm of Baswana \etal\ is thus \emph{not} worst-case expected:
if the adversary starts with the above (long) sequence and then deletes $(v,v')$, 
this deletion is near-guaranteed to be expensive because $\PP[\mate(v) = v'] \sim 1$.

One way to overcome this issue is to give $v$ a small probability of resetting its matching edge every time a neighbor of $v$ undergoes certain kinds of changes in the hierarchy;
this would ensure that even if $(v,v')$ becomes the matching edge at some point during the process, it will not stick forever. This is the approach we will take.

\subsubsection{List of Changes to the Baswana \etal\ Algorithm}
\label{subsec:matching-modified}

We now describe the changes that we make the original algorithm of Baswana et al. \cite{BaswanaGS18}. Full pseudocode for our algorithm is given in Algorithms \ref{alg:matching 1} and \ref{alg:matching 2}.

\begin{itemize}
\item To simplify the algorithm we remove the lists $\OO_v$, $\EE_v^i$, and the notation of \emph{ownership}. (Note that the original algorithm also could be changed in this way.)
Instead we keep for each vertex $v$ the following sets in a dynamic hash table and also maintain their respective sizes: (a) For each level $j > \LL(v)$ the set $\equal{j}(v)$, i.e., all edges incident to neighbors on level $j$, and (b) one set containing the set
$\below{\LL(v)+1}(v)$, i.e., all edges incident to neighbors on level $\LL(v)$ and below.

\item Invariants 2-4 are exactly the same as above, Invariant 1 is no longer needed as we no longer use the concept of ownership.

\item Define $C$ to be a sufficiently large constant used by the algorithm.

\item Whenever the algorithm executes \IncrementPhi{$v,i$} for a vertex $v$ with $\LL(v) < i$, the algorithm: \textbf{1)} performs
\Rise{$v,\LL(v),i$} with probability $\prise = C\log(n)/4^i$. We call this a \emph{probabilistic rise}. 
\textbf{2)} always performs \Rise{$v, \LL(v), i$}
if $\phi_v(i)$ increases from $4^i - 1$ to~$4^i$; we call this a \emph{threshold rise}. 
(The original algorithm of Baswana \etal\ only performed threshold rises. Our new version modifies line 13 in the pseudocode of Procedure \textsc{process-free-vertices} of~\cite{BaswanaGS18}, as well as the paragraph ``Handling insertion of an edge'' in Section~4.2 in~\cite{BaswanaGS18}.)

\item \textbf{Matching Property*} If a vertex $v$ at level $i > -1$ is (temporarily) unmatched and $|\below{i}(v)| \geq 4^{i}/(32C\log(n))$, then $v$ will pick a new mate \emph{uniformly at random} from $\below{i}(v)$. If $|\below{i}(v)| < 4^{i}/(32C\log(n))$, then $v$ falls to level $i-1$ and is recursively processed from there. Note that Matching Property* is identical to Matching Property above, but with $4^{i}/(32C\log(n))$ instead of $4^{i}$.
This leads to the following change in procedure \FixFreeVertex{$v$}. Let $i = \LL(v)$: 
if $|\below{i}(v)| \geq 4^{i}/(32C\log(n))$, then the algorithm executes \GenericRandomSettle{$v,i$}, and if 
$|\below{i}(v)| < 4^{i}/(32C\log(n))$, then it executes \Fall{$v,i$}.
(Our version modifies line~$5$ of Procedure \textsc{falling} of~\cite{BaswanaGS18}).

\item We will keep a boolean $\responsible(v)$ for each vertex $v$, which is set to True if the matching edge $(v,w)$ was chosen during \GenericRandomSettle{$v,\ell$}. 
In this case we say that $v$ is \emph{responsible} for the matched edge.
If $v$ is free, $\responsible(v)$ is False.
Each matching edge will have exactly one endpoint with \responsible($v$) set to True, and free vertices will always be set to False.

\item We make the following change for processing an adversarial insertion of edge $(u,v)$. The algorithm executes \ResetMatching{$u$} with probability $\preset_{\LL(u)}$  and it executes \ResetMatching$(v)$ with probability $\preset_{\LL(v)}$, where $\preset_i = \resetprob{i}$. If $\responsible(u) =$ True, \ResetMatching{$u$} simply picks a new matching edge for~$u$ by removing edge $(u, \mate(u))$ from the matching and then calling \FixFreeVertex{$u$} and \FixFreeVertex{$\mate(u)$}; if $\responsible(u)$ is False, \ResetMatching{$u$} does nothing. 
(Our version modifies the paragraph ``Handling insertion of an edge" in Section~5.2 of~\cite{BaswanaGS18}.)

\item We also make the following change to procedure \Fall{$v,i$}. 
All edges of the set $\below{i}(v)$ are traversed, not just the ones owned by $v$.
Since $|\below{i+1}(v)| = O(4^{i+1})$ (by Invariant 3),
the running time analysis of ~\cite{BaswanaGS18} remains valid.
Furthermore, recall that as a result of $v$ falling to level $i-1$, $v$ now belongs to $\equal{i-1}(u)$ for every neighbor $u$ of $v$ at level~$ i-1 $. Each such neighbor~$u$ then executes \ResetMatching{$u$} with probability $\preset_{\LL(u)}$.
(Our version modifies lines 3 and~4 in Procedure \textsc{falling} of~\cite{BaswanaGS18}.)

\end{itemize}

\paragraph{Pseudocode.}

We give the pseudocode for the whole modified algorithm in Algorithms~\ref{alg:matching 1} and~\ref{alg:matching 2}. The pseudocode shows
how the basic procedures of the algorithm (e.g. \Rise, \Fall, \FixFreeVertex, \ResetMatching) call each other. 
We note that in the pseudocode, whenever the algorithm changes the level of a vertex in the hierarchy it also performs straightforward
\emph{bookkeeping work} that adjusts all sets
$\below{i}(v)$ and $\equal{i}(v)$ to match the new hierarchy.
For example, if a vertex falls from level $i$ to 
level $i-1$, then for every edge $(v,w)$ with $w \in \below{i+1}(v)$ we do the following: 
if $\LL(w) = i$ then we transfer $w$ from $\below{i+1}(v)$ to $\equal{i}(v)$;
if $\LL(w) = i-1$ then we transfer $w$ from $\equal{i}(w)$ to $\below{i}(w)$
and from $\below{i+1}(v)$ to $\below{i}(v)$; and
if $\LL(w) < i-1$ then we transfer $w$ from $\equal{i}(w)$ to $\equal{i-1}(w)$
and from $\below{i+1}(v)$ to $\below{i}(v)$.
Note that by Invariant 3 $\phi_{i+1}(v) < 4^{i+1} = O(4^{\LL(v)})$ and, thus, the bookkeeping can be done in time $ O (4^{\LL (v)})$.

The matching maintained by the algorithm in the pseudocode is denoted by $ \mathcal{M} $.
Note that for technical reasons calls of \FixFreeVertex{$v$} for vertices $ v $ in our algorithm are not executed immediately.
Instead, we maintain a global FIFO queue $ Q $ of vertices $ v $ for which we still need to perform \FixFreeVertex{$v$}, implemented as a doubly-linked list.
To avoid adding the same vertex to the queue twice and enable us to a vertex in the queue at the end of the queue we store at every vertex a pointer to its position in the queue. If a vertex is not in the queue, this pointer is set to NIL.

\begin{algorithm2e}
\caption{Fully Dynamic Maximal Matching Algorithm}\label{alg:matching 1}

$\prise_i \gets C\log(n)/4^i$ for some large constant $C$ \;
$\preset_i \gets 1/4^{i+3}$
Initialize empty queue $ Q $ 

\Procedure(\tcp*[f]{Process deletion of edge $ (u, v) $}){\Delete{$ u $, $ v $}}{
Perform bookkeeping work for deletion of $ (u, v) $\;
	\If{$ (u, v) \in \mathcal{M} $}{
		Perform bookkeeping work for deletion of $ (u, v) $ \label{step:delete-change} \;
		Set $\responsible(u)$ and $\responsible(v)$ to False \;
		Add $ u $ to end $ Q $ (or move $u$ to end if it is already in $ Q $)\; 
		Add $ v $ to end of $ Q $ (or move $v$ to end if it is already in $ Q $)\; 
		\ProcessQueue{} \label{step:delete-process-queue}\;
	}
}

\Procedure(\tcp*[f]{Process insertion of edge $ (u, v) $}){\Insert{$ u $, $ v $}}{
	Perform bookkeeping work for insertion of $ (u, v) $ \label{step:insert-change} \;
	\ForEach(\label{step:insert-begin-rise}){$ j > \max{ \{\LL(u),\LL(v)\} } $ in increasing order}{
		\IncrementPhi{$ v $, $ j $} \label{step:insert-rise}\; 
		\IncrementPhi{$ u $, $ j $} \label{step:insert-end-rise}\; 
	}
	With probability $ \preset_{\LL(u)} $ \KwDo \ResetMatching{$ u $} \label{step:addition/removal/move-to-end}\label{step:insert-reset} \;
	With probability $ \preset_{\LL(v)} $ \KwDo \ResetMatching{$ v $} \label{step:addition/removal/move-to-end2} \;
	\ProcessQueue{} \label{step:insert-process-queue} \;
}

\Procedure{\ProcessQueue{}}{
	\While{$ Q $ is not empty}{
		Pop the first vertex $ v $ in $ Q $ \label{step:pop} \;
		\FixFreeVertex{$v$} \label{step:fix} \;
	}
}
\end{algorithm2e}

\begin{algorithm2e}
\caption{Fully Dynamic Maximal Matching Algorithm}\label{alg:matching 2}

\setcounter{AlgoLine}{22}

\Procedure(\tcp*[f]{Called only if $ \LL (v) > - 1$}){\ResetMatching{$ v $}}{
	\If{$\responsible(v)$ = False} {Exit Procedure \ResetMatching}
	$ w \gets \mate (v) $\;
	$ \mathcal{M} \gets \mathcal{M} \setminus \{ (v, w) \} $ \label{step:reset-delete} \tcp*[r]{Unmatch $ v $ and $ w $}
	Set $\responsible(v)$ to False \tcp*[f]{Note: $\responsible(w)$ was already False, since $v$ was responsible for $(v,w)$}\;
	Add $ v $ to end of $ Q $ (or move $v$ to end if it is already in $ Q $)\; 
	Add $ w $ to end of $ Q $ (or move $w$ to end if it is already in $ Q $) 
}

\Procedure{\FixFreeVertex{$ v $}}{
	$ i \gets \LL (v) $\;
	Compute $ \below{i} (v) $ from $ \below{i+1} (v) $ \label{step:fix-compute-sets}\;
	\If{$ i > -1 $ \KwAnd $ v $ is unmatched}{
		\eIf(\label{step:fix-if}){$ |\below{i}(v)| \geq 4^i / (32 C \log{(n)}) $}{
			Compute $ \below{i} (v) $\;
			\GenericRandomSettle{$ v $, $ i $}\;
		}{
			\Fall{$ v $, $ i $} \label{step:fix-else}\;
		}
	}
}

\Procedure(\tcp*[f]{Called only if $ |\below{i}(v)| \geq 4^i / (32 C \log{(n)}$}){\GenericRandomSettle{$ v $, $ i $}}{
	 Pick $ w \in \below{i}(v) $ uniformly at random 	\label{step:settle-random-mate} \;
	 \Rise($w, \LL(w),i$) \label{step:settle-rise} \;
	$ \mathcal{M} \gets \mathcal{M} \cup \{(v, w)\} $ \label{step:settle-insert} \tcp*[r]{Match $ v $ and $ w $} 
	$\responsible(v) \leftarrow$ True
}

\Procedure{\Fall{$ v $, $ i $}}{
	Compute $ \below{i} (v) $ and $ \equal{i} (v) $ from $ \below{i+1} (v) $\;
	Perform bookkeeping work to move $ v $ from level~$ i $ to level~$ i-1 $ \label{step:fall-change} \;
	\ForEach(\label{step:fall-begin-rise}){$ w \in \below{i}(v) $}{
		\IncrementPhi{$ w $, $ i $} \label{step:fall-increment} \tcp*[r]{$v$ joins $N_{<i}(w)$} 
	}
	\ForEach{$ w \in \equal{i-1}(v) $}{
			With probability $ \preset_{i-1} $ \KwDo \ResetMatching{$ w $} \label{step:fall-reset}
	}
	Add $ v $ to end of $ Q $ (or move $v$ to end if $v$ is already in $Q$) \label{step:fall-end} 
}

\Procedure(\tcp*[f]{Called when $|N_{<i}(v)|$ increases for $ i > \LL (v) $}){\IncrementPhi{$ v $, $ i $}}{
	\lIf(\tcp*[f]{Threshold rise}){$ |\below{i}(v)| \geq 4^i $}{
		\Rise{$ v $, $ \LL (v) $, $ i $}
	}

	\lElse(\label{step:increment-prise} \tcp*[f]{Probabilistic rise}){With probability $ \prise_i $ \KwDo \Rise{$ v $, $ \LL (v) $, $ i $}}  
}

\Procedure{\Rise{$ v $, $ i $, $ j $}}{
	\If(\tcp*[f]{Check if $ v $ is matched with some neighbor $ w $}){$ \exists (v, w) \in \mathcal{M} $}{
		$ \mathcal{M} \gets \mathcal{M} \setminus \{ (v, w) \} $ \label{step:rise-mate-delete} \tcp*[r]{Unmatch $ v $ and $ w $}
		Set $\responsible(v)$ and $\responsible(w)$ to False \;
		Add $ w $ to end of $ Q $ (or move $w$ to end if it is already in $ Q $) 
	}
	Perform bookkeeping work to move $ v $ from level $ i $ to level $ j $ \label{step:rise-change} \;
	Add $ v $ to end of $ Q $ (or move $v$ to end if it is already in $Q$) 

}

\end{algorithm2e}

\subsection{Correctness of the Modified Algorithm}
To show the correctness of the modified algorithm we need to show that it fulfills Invariants 2--4 and that \textbf{Matching Property*} holds.
We will do so in this subsection.
Termination is guaranteed in the next section, which shows that the expected time to process an adversarial
edge insertion/deletion is finite. 

\begin{lemma}
	\label{lem:invariants}
Invariants 2--4, and Matching Property* hold before and after the processing of each edge update. Also, for every free vertex $v$ we have $\responsible(v)$ is False, and for any matching edge $(v,w)$, $\responsible(v)$ is True if and only if $(v,w)$ was last chosen to enter the matching during a call to \GenericRandomSettle{$v,\cdot$}; as a consequence, exactly one of $\responsible(v)$ and $\responsible(w)$ is True.
\end{lemma}

\begin{proof}
	
\emph{Responsibilities:} the claims about $\responsible(v)$ follow trivially from the pseudocode, as we always explicitly maintain these properties.

\emph{Invariant~2:}
To show invariant~2 we need to show that (a) every vertex on a level larger than $-1$ is matched and (b) every vertex on level $-1$ is free.
We show the claim by induction on the number of updates. Initially the graph is empty and every vertex is unmatched and on level $-1$. Thus, the claim holds.
Assume now that the claim holds before an edge insertion or deletion. We will show that it holds also after the edge insertion or deletion was processed.

We first show (a).
A vertex $v$ on level larger than $-1$ can violate Invariant~2 if (1) its matched edge was deleted or (2) it became unmatched in procedure \GenericRandomSettle or \Rise.
In both cases $ v $ is placed on the queue (if it is not already there).
Then the current procedure completes and then other calls to \FixFreeVertex might be executed before the call \FixFreeVertex{$v$} is started. 
Thus, it is possible that the hierarchy has changed between the time when $ v $ was placed on the queue and the time when its execution starts. 
This is the reason why \FixFreeVertex{$v$} first checks whether $v$ is still on a level larger than $-1$ and whether it is still unmatched. If this is not the case, $v$ fulfills Invariant~2.
If this is still the case, then the main body of \FixFreeVertex{$v$} is executed, which either matches $v$ with \GenericRandomSettle{$v$,$\LL(v)$}
\emph{or} it decreases the level of $v$ (if $v$ does not have ``enough'' neighbors on levels below $\LL(v)$)
and then places $ v $ on the queue.
As the update algorithm does not terminate until the queue is empty, it is guaranteed that all vertices fulfill  (a) at termination of the update.

To show (b) note that a vertex $u$ is only matched in procedure \GenericRandomSettle and in this case it needs to be on a level $i$ such that either the vertex $u$ itself or its newly matched partner~$v$ fulfill the property that there is at least one neighbor in a level \emph{below} level $i$. As $-1$ is the lowest level, it follows that $i > -1$, which shows (b).

\emph{Invariant~3:}
For Invariant~3 we need to show for all $j > \LL(v)$ that $|\below{j}(v)| < 4^j$.
We show the claim by induction on the number of updates.  
The property certainly holds at the beginning of the algorithm when there are no edges. 
Assume it was true before the current edge update. We will show that it also holds after the current edge update. 
Let $v$ be a vertex. 
The set $\below{j}(v)$ increases  only if (i) a neighbor $w$ drops from a level at or above $j$ to a level below $j$ or (ii) an edge incident to $v$ is inserted. 
We show next that Invariant~3 holds in either case.
In case (i) since each execution of \Fall decreases the level of a vertex only by one, the set $\below{j}(v)$ can only increase if a neighbor $w$ drops from $j$ to $j-1$. 
As a consequence it follows that
the sets  $\below{k}(v)$ for \emph{all} $k \ne j$ are unchanged, and, thus, $| \below{k}(v)| < 4^k$ for all $k \ne j$, i.e., there is only one set $\below{\cdot}(v)$ that might violate the invariant, namely $\below{j}(v)$ and in this case $| \below{j}(v)| = 4^j$. 
The fall of $w$ from level $j$ calls \IncrementPhi{$v,j$}, which in turn immediately calls \Rise{$v$,$\LL(v)$,$j$} if $|\below{j}(v)| =  4^j$. After $v$ has moved up to level $j$, it holds that for all
$k > j$ that $|\below{k}(v)| < 4^k$ as this was also true before the rise. Thus, Invariant~3 holds again for $v$.
In case (ii) in the insertion operation the function \IncrementPhi{$v,j$}
is called for every level $j$ that is larger than the level of $v$.
If for one of these levels, let's call it $i$, $|\below{i}(v)| \ge  4^i$,
then $v$ is moved up
to the highest level $j$ for which $|\below{j}(v)| \ge  4^j$.
Thus Invariant~3 is guaranteed.

\emph{Invariant~4:}
For Invariant~4 we have to show that the endpoints of every matched edge are at the same level. Note that two vertices $v$ and $w$  only become matched in procedure
\GenericRandomSettle and right before that the vertex (out of the two) on the lower level is ``pulled up'' to the level of the higher vertex. Thus, both are at the same level when they are matched.

\emph{Matching Property*:}
Finally Matching Property* holds for every vertex $v$ for the following reason: As soon as a vertex becomes unmatched, $ v $ is placed on the
queue. Whenever this call is executed, it checks whether   $|\below{i}(v)| \geq 4^{i}/(32C\log(n))$, where $i = \LL(v)$, is fulfilled and if so, it calls 
\GenericRandomSettle{$v$,$i$}, which in turn picks a random neighbor of $\below{i}(v)$ and matches $v$ with it. 
If, however, $|\below{i}(v)| < 4^{i}/(32C\log(n))$, then \FixFreeVertex{$v$} calls the procedure \Fall{$v$,$i$}. The procedure \Fall checks again whether it still holds that $|\below{i}(v)| < 4^{i}/(32C\log(n))$, and if so $v$ is moved one level down.  Since in this case the vertex is still unmatched, \Fall{$v$,$i$} also inserts $ v $ into the queue, which later on results in a call to \FixFreeVertex{$v$} executed on $v$'s new level $i-1$. Thus, $v$ continues to fall until it either reaches a level~$i$ where $|\below{i}(v)| \geq 4^{i}/(32C\log(n))$ (in which case it is matched there) or until it reaches level $ -1 $, in which case $|\below{i}(v)| = 0 <1$. Hence, in either case Matching Property* holds.
\end{proof}

\begin{remark}
We later prove a stronger version of Lemma \ref{lem:invariants}, which shows that Invariant 4 and relaxed version of Invariants 3 hold not just at the end of processing an adversarial update, but also at all points in the middle of processing an update. See Lemma \ref{lem:invariants-stronger} for more details.
\end{remark}

\subsection{Analysis of the Modified Algorithm}

Note that each procedure used by the algorithm (e.g. \Fall or \FixFreeVertex) incurs two kinds of costs: 
\begin{itemize}
\item \textbf{Bookkeeping work:} As discussed above, if the procedure changes the level of a vertex, the algorithm must do bookkeeping work to maintain the various sets $\below{\LL(v)+1}(v)$ and $\equal{j}(v)$ data structures.
\item \textbf{Recursive work:} a change in the hierarchy could lead other vertices to violate one of the invariants, and so lead to the execution of further procedures. 
\end{itemize}

We start with the easier task of analyzing the bookkeeping work.
The \Fall, \GenericRandomSettle, and \FixFreeVertex procedures all require $O(4^i)$ time
to process a vertex $v$ at level $i$: this is because the bookkeeping work only requires us to look at $\below{i+1}(v)$, which
by Invariant~3 contains at most $O(4^{i+1})$ edges. We now analyze the bookkeeping required for procedure \Rise:

\begin{lemma} 
\label{lem:rise}
\Rise{$v,i,j$} requires $O(4^{j})$ bookkeeping work.
\end{lemma}

\begin{proof}
Although they do not state it as such, this lemma holds for the original Baswana \etal\ algorithm as well. When $v$ rises from level $i$ to level $j$, the algorithm performs bookkeeping of two sorts. Firstly, every neighbor $u$ of $v$ whose level is
less than $j$ must update $\equal{j}(v)$ and either $\equal{i}(v)$, if $\LL(u) <i$, or
$\below{\LL(u)+1}(u)$ if $i \le \LL(u) < j$.
By Invariant~3 there are at most  
$O(4^{j})$ neighbors to update. Secondly, every neighbor $u$ of $v$ in $\below{j}(v)$ must execute \DecrementPhi{$u,k$} for every 
$\mmax{i}{\LL(u)} < k \leq j$. The total cost is upper bounded by 
\begin{equation}
\label{eq:rise}
O(\below{j}(v)\cdot(j-i) + \sum_{k=i+1}^{j-1} |\equal{k}(v)|),
\end{equation}
where $\equal{k}(v)$ is the number of neighbors of~$v$ at level~$k$. But note that for $k > i$, $|\equal{k}(v)| < 4^k + 1$,
since otherwise $v$ would have violated Invariant~3 for level $k$ even before the procedure call that led to \Rise{$v,i,j$}. 
Similarly, $|\below{i}(v)| < 4^{i+1}+1$. Plugging these bounds into Equation~\ref{eq:rise} yields 
$$\sum_{k=i}^{j-1} O(4^{k}) = O(4^j).$$
\end{proof}

Before analyzing the recursive work, we bound the probability that \IncrementPhi{$v,i$} calls \Rise. The chance of a probabilistic rise is always the same
$\prise = \Theta(\log(n)/4^i)$. We now bound threshold-rises. For this we need to introduce the notion of a hierarchy.
When we refer to the \emph{(graph) hierarchy~$\hh(t)$} at time $t$, we mean the current graph $G$, the level assigned to each vertex at time $t$,
as well as the set of edges in the matching at time $t$. 

The sequence of oblivious updates
predefined by the adversary gives some probability distribution on the point in time (in the sequence of updates) to process the first call to \IncrementPhi, the second call to \IncrementPhi, the third one, and so on.

\begin{lemma}
\label{lem:threshold-rise}
For any $k$, it holds with high probability that the $k$th call to \IncrementPhi does \emph{not} lead to a threshold rise.
\end{lemma}

\begin{proof}
Let $\bad_{v,i}$ be the bad event that the $k$th call to \IncrementPhi increments $\phi_v(i)$ from $4^{i}-1$ to~$4^i$.
It is enough to show that $\neg \bad_{v,i}$ occurs with high probability; we can then union bound over all pairs $(v,i)$.
Let $t_k$ be the time at which this $k$th call to \IncrementPhi occurs, and note that at the beginning of time~$t_k$ we have $\LL(v) < i$, 
since otherwise we would have $\phi_v(i)=0$ and no threshold rise would occur. Now, let $t$ be the earliest point in 
time such that $\LL(v) < i$ in the entire time interval from $t$ to $t_k$, i.e.,~$t$ is either the start of the algorithm or a point in time when $v$ falls below level $i$. It is not hard to see 
that because Matching Property* only allows a vertex to fall to below level $i$ when $|\below{i}(v)| < 4^{i}/(32C\log(n))$,
it must be the case that at time~$t$ we have $\phi_v(i) < 4^{i}/(32C\log(n)) < 4^i/2$. (There is also the fringe case $t=0$; in this case $\phi_v(i) = 0$ at time $t$ because in the fully dynamic setting one can assume w.l.o.g.\ that the graph starts empty.)
Thus, there must have been at least $4^i/2$ calls to \IncrementPhi{$v,i$} in time interval $(t, t_k)$, and
by the assumption that $ \LL(v) < i$ in this entire time interval, none of these $4^i/2$ calls to \IncrementPhi{$v,i$} 
led to a probabilistic rise. But each probabilistic rise occurs independently with probability $C\log(n)/4^i$ (for a sufficiently large~$C$), 
so a simple Chernoff bound shows us that with high probability this event does not occur.
\end{proof}

We now turn to bounding the recursive work incurred by a procedure. Let us first define this more formally. 
Bringing attention to the pseudocode, we note that each procedure is either directly called by some previous procedure,
or, in the case of \FixFreeVertex{$v$}, it is \emph{indirectly} called by the procedure that added $v$ to the queue; 
we say that the called procedure is
\emph{caused} by the calling procedure and, in case of \FixFreeVertex{$v$}, we say that it is \emph{caused} by the last procedure that
affected $v$'s position on the queue, either by placing it on the queue or by moving it to the end. For example, a procedure \Fall leads to many calls to \IncrementPhi, and 
so can potentially lead to many calls to procedure \Rise. We can thus construct a \emph{causation tree} for each adversarial edge update, whose root is the procedure handling the adversarial edge update and where the parent procedure
causes all the children procedure calls. We then say that the \emph{total work} of a procedure is the bookkeeping work
of that procedure, plus  (recursively) 
the total work of all of its children in the causation tree; equivalently, the total work
of a procedure is the total bookkeeping work required to process all of its descendants in the causation tree. 

We now show that for any vertex $v$ at level $i = \LL(v)$, the \emph{expected total work} of any call to \Fall{$v,i$}, \GenericRandomSettle{$v,i$}, or \FixFreeVertex{$v$} is at most $O(4^{i})$. The running time of the remaining procedures \ResetMatching, \IncrementPhi, and \Rise can then be
easily analyzed in terms of the analysis of the earlier three procedures.
Note that the time to process any procedure at time $t$ depends on two things: the hierarchy at time $t$,
and the random coin flips made after time $t$. Thus, we can define $\efall_i(v,\hh(t))$ to be the expected total work to process \Fall{$v,i$} given that the state of the current hierarchy is $\hh(t)$, 
where the expectation is taken over all coin flips made after time $t$. 
(We assume that in the hierarchy $\hh(t)$ vertex $v$ has level $i$, since otherwise \Fall{$v,i$} is not a valid procedure call.)
We define $\egenericsettle_i(v,\hh(t))$ and $\eprocessfree_i(v,\hh(t))$ analogously.
We say that some hierarchy $\hh(t)$ is valid if it satisfies all of the hierarchy invariants above; note that our
dynamic algorithm always maintains a valid hierarchy. 

We are now ready to introduce our key notation. We let $\efall_i$ be the maximum of all $\efall_j(v,\hh)$, where the maximum
is taken over all levels $j \le i$, all vertices $v$, and all valid hierarchies $\hh$ in which $v$ has level $j$.
Define $\egenericsettle_i$ and 
$\eprocessfree_i$ accordingly.
Define $\emax_i = \max \{ \efall_i, \egenericsettle_i, \eprocessfree_i \}$. 
Note that because $\emax_i$ takes the maximum over all valid hierarchies, it is an upper bound on the expected
time to process \emph{any} update at level $i$. We now prove a recursive formula for bounding $\emax_i$.

\begin{lemma}
\label{lem:emax}
$\emax_i \leq O(4^{i}) + 3 \emax_{i-1}$ 
\end{lemma}

\begin{proof}
We first show that $\egenericsettle_i \leq O(4^{i}) + \emax_{i-1}$. 
\GenericRandomSettle{$v,i$} picks some random mate $v'$ for $v$
with $\LL(v') < i$, performs $O(4^{i})$ bookkeeping work to move $v'$ to level $i$ 
(Lemma~\ref{lem:rise}), 
and then causes a single other procedure call, namely, \FixFreeVertex{$\oldmate(v')$}; 
this caused procedure call occurs at some level less than $i$,
so the expected total work can be upper bounded by~$\emax_{i-1}$. 

Now consider \FixFreeVertex{$v$}, where  $i = \LL(v)$. 
The work of this procedure to construct $\below{i+1}(v)$ from $\below{i+1}(v)$ is $O(4^i)$. 
The algorithm then causes one other procedure call: either \GenericRandomSettle{$v,i$} or \Fall{$v,i$}, depending on the size of $\below{i}(v)$. We have already bounded
$\egenericsettle_i$, so all that remains is to bound $\efall_i$.

Recall that the algorithm only executes \Fall{$v,i$} when $\below{i}(v) < 4^{i}/(32C\log(n))$. The procedure \Fall requires the standard $O(\below{i+1}(v)) = O(4^{i})$ bookkeeping work, and it also causes a call to \FixFreeVertex{$v$} at level $i-1$,
which has $\emax_{i-1}$ expected total work. \Fall{$v,i$} can also lead to additional updates at level $i-1$ due to \ResetMatching: see line \ref{step:fall-reset} of Algorithm \ref{alg:matching 2}.  Finally, unlike the other procedures,
\Fall{$v,i$} can also cause additional procedure calls at level $i$. This can happen because each neighbor $u$ of $v$ at level $\leq i-1$ executes 
\IncrementPhi{$u,i$} (line \ref{step:fall-increment}) (Note that it is not possible that a neighbor calls a
\IncrementPhi{$u,j$} for $j > i$)
. This has a small chance of resulting in \Rise{$u, \LL(u), i$} (either through a probabilistic rise or through a
threshold rise - see Lemma~\ref{lem:threshold-rise}), followed by \FixFreeVertex{$u$}, where $\LL(u)=i$, and \FixFreeVertex{$\oldmate(u)$}, where $\LL(\oldmate(u)) \le i-1$.

Let $\xreset$ be the random variable that stands for the number of \ResetMatching{$u$} triggered by the fall of $v$, and note that every such vertex is at level $i-1$. Let $\xrise$ be the number of \Rise{$u, \LL(u), i$} triggered by the fall. 
Note that $\EX[\xreset] \leq 1/4$
because $v$ is at level $i$ before the fall, so by Invariant~3, $v$ has at most $4^{i+1}$ neighbors at level $i-1$, and each neighbor has a $\preset_{i-1} = \resetprob{i-1}$ probability of being reset. 
We now argue that $\EX[\xrise] \leq 1/16$. By Matching Property*, $v$ has at most $4^i/(32C\log(n))$ neighbors~$u$ at lower level before the fall,
each of which executes \IncrementPhi{$u,i$}. Our modification to the original algorithm ensures that this increment has a $\prise = C\log(n)/4^i$ chance
of inducing a probabilistic-rise, and by Lemma~\ref{lem:threshold-rise} the probability of a threshold-rise is negligible (Lemma \ref{lem:threshold-rise}), so for simplicity we
upper bound it by $C\log(n)/4^i$. Thus: $\EX[\xrise] = [4^i/(16C\log(n))][2C\log(n)/4^i] = 1/16$.

We now consider the total work to process a fall. Firstly, the fall automatically
triggers $O(4^{i})$ bookkeeping work plus it causes a procedure call at level $i-1$; by definition, the expected
total work to process this additional procedure call can be upper bounded with $\emax_{i-1}$.
We also have to do additional work for each call to \ResetMatching or \Rise. 
Each reset causes two additional procedure calls at level $i-1$, whose running time
we upper-bound by $2\emax_i$ (using the fact that $\emax_{i-1}\le \emax_i$). Each \Rise procedure requires $O(4^i)$ bookkeeping work and causes a new call at level $i$, as well as a call at another level less than $i$ (due to the old mate $w$ becoming free). We upper-bound the time to execute these two
calls at level $i$ or less again by $2\emax_i$. 
Note that this upper bound allows to achieve a crucial probabilistic independence:
although the value of $\xrise$ might be correlated with the time to process these calls to $\Rise$
(both depend on the current hierarchy), the value of $\xrise$ is \emph{completely independent} from~$\emax_i$, since
the latter takes the maximum over all valid hierarchies, and
so does not depend on the current hierarchy. Now, recall that $\EX[\xreset] \leq 1/4$ and $\EX[\xrise] \leq 1/16$. 
Putting it all together, we can write a recursive formula for $\emax_i$.
\begin{equation}
\begin{split}
\emax_i & \leq O(4^i) + \emax_{i-1} + (2\emax_i + O(4^i)) \sum_{k=1}^\infty k\PP[\xreset + \xrise = k] \\
& \leq O(4^i) + \emax_{i-1} + (2\emax_i + O(4^i))(\EX[\xrise + \xreset]) \\
&= O(4^i) + \emax_{i-1} + (2\emax_i + O(4^i))\frac{5}{16} < O(4^i) + \emax_{i-1}  + \frac{5}{8}\emax_i. \qedhere
\end{split}
\end{equation}
Bringing $\frac{5}{8}\emax_i$ to the left side of the inequality and multiplying it by $8/3$ it leads to the statement of the lemma.
\end{proof}

\begin{corollary}
\label{cor:emax}
The expected total work for a call to \FixFreeVertex, \Fall, \GenericRandomSettle, \Rise, or \ResetMatching or \IncrementPhi
at level $i$ is $O(4^i$), where \Rise{$v,i',i$} is said to be a procedure call at level $i$.
\end{corollary}

\begin{proof}
Solving the recurrence relation in Lemma~\ref{lem:emax} yields $\emax_i = O(4^i)$, which gives us
the desired bound for \FixFreeVertex, \Fall, and \GenericRandomSettle. Procedure \ResetMatching 
causes two calls to \FixFreeVertex, so the same $O(4^i)$ bound applies. 
Procedure \Rise 
requires $O(4^i)$ bookkeeping work (Lemma~\ref{lem:rise}), and then causes at most two other calls to \FixFreeVertex,
each of which we know has expected total work $O(4^i)$.
Procedure \IncrementPhi does $O(4^i)$ bookkeeping work and then causes at most one
call to Procedure \Rise at level $i$. Thus the same $O(4^i)$ bound applies.
\end{proof}

\subsection{Bounding the Probability that an Edge Appears in the Matching.}
\label{subsec:matching-probability}

Now that we have analyzed the time to process the individual procedure calls, we turn our attention to the time required to process an adversarial edge insertion/deletion.
Note that the most direct reason the algorithm might have to perform a procedure call at level $i$ is the deletion of matching edge $(v, \mate(v))$ with $v$ and
$\mate(v)$ at level $i$. 
Our modifications to the algorithm allow us to do without the charging argument of Baswana \etal, 
and instead directly bound the probability that a deleted edge $(x,y)$ is a matching edge.
Note that for $(x,y)$ to be a matching edge, 
it must have been chosen by a \GenericRandomSettle{$x,i$} or \GenericRandomSettle{$y,i$} for some level $i$.
There are thus $2(\textsc{l}_0+1) = O(2\log(n))$ possible procedure calls
that could have created this matching edge: 
we bound the probability of each separately.

\begin{lemma}
\label{lem:prob-matching}
Let $(x,y)$ be any edge at any time $\tstar$ during the update sequence, and let $0\leq \ell \leq \floor{\log_4(n)}$ be any level in the hierarchy. 
Then: $\PP[$at time $\tstar$, $(x,y)$ is a matching edge at level $\ell$ and $\responsible(x)$ is True$]$ $= O(\log^3(n)/4^\ell)$, where the probability
is over all random choices made by the algorithm. (Note that this is equivalent to the probability that $(x,y)$ was chosen by \GenericRandomSettle{$x,\ell$}.)
\end{lemma}

\begin{corollary}
\label{cor:prob-matching}
Let $(x,y)$ be any edge at any time $\tstar$ during the update sequence, and let $0\leq \ell \leq \floor{\log_4(n)}$ be any level in the hierarchy. 
Then: $\PP[$at time $\tstar$, $(x,y)$ is a matching edge at level $\ell]$ $= O(\log^3(n)/4^\ell)$, where the probability is over all random choices made by the algorithm.
\end{corollary}

\begin{proof}[Proof Of Corollary \ref{cor:prob-matching}]
For any edge $(x,y)$ at level $\ell$ that is in the matching, we have that either $\responsible(x)$ or $\responsible(y)$ is True. We can thus apply Lemma \ref{lem:prob-matching} to each of those two cases and union bound the two resulting probabilities.
\end{proof}

The proof of this lemma is very involved, and the rest of this subsection is devoted to proving it. Let us first briefly discuss the naive approach and why it fails to work. Let $t$ be the \emph{last} time before~$\tstar$ that \GenericRandomSettle{$x$,$\ell$} is called, and note that assuming $\responsible(x)$ is True at time $\tstar$ the matching edge is precisely the matching edge picked at time $t$. Matching Property* guarantees that at any given call to \GenericRandomSettle{$x$,$\ell$} only has a $O(\log(n)/4^\ell)$ probability of picking the specific edge $(x,y)$. Thus, it is tempting to (falsely) argue that at time $t$ the probability that \GenericRandomSettle{$x$,$\ell$} picked edge $(x,y)$ is at most $O(\log(n)/4^\ell)$.
But this might not be true, because although \GenericRandomSettle{$x$,$\ell$} picks an edge uniformly at random from many options,
the fact that we condition on $t$ being the \emph{last} random settle before $\tstar$ means that we condition on events \emph{after} time $t$, which can greatly skew the distribution at time $t$. 
Consider, for illustration, the update sequence in the star graph at the beginning of Section~\ref{subsec:modification}:  
the sequence repeatedly inserts and deletes all edges other than $(v,v')$, so any edge other than $(v,v')$ is unlikely to be the \emph{last}
matching edge, since it will soon be deleted. To overcome this issue, we now present a more complex analysis that (loosely speaking) bounds the total number of times \GenericRandomSettle{$x$,$\ell$} is called in some critical time period. 

One of the main difficulties of the proof is that we have to be very careful with the assumption of obliviousness. The model assumes that adversarial updates are oblivious to our hierarchy, so the specific mate that $v$ chooses in some \GenericRandomSettle{$v,k$} will not affect future adversarial updates. But the internal changes made by the algorithm might not be oblivious: if $v$ chooses the specific edge $(v,w)$, this will change the level of $w$, which will lead to changes to the neighbors of $w$, which might indirectly increase the probability of some \ResetMatching{$x$}, which will lead to the removal of $(x,y)$ from the matching. Thus, internal updates are adaptive to internal random choices.

To overcome this adaptivity issue, we will show that the higher levels of the hierarchy are in fact oblivious to random choices made by lower levels of the hierarchy.

\paragraph{Comparison to the analysis of Baswana et al.} Our proof of Lemma \ref{lem:prob-matching} consists of two main parts. The first part, which includes Sections \ref{subsec:hierarchy-invariants} and \ref{subsec:higher-and-lower}, establishes that higher levels of the hierarchy are independent from random choices made on lower levels. This part is very similar to an analogous proof of independence in Baswana et al. \cite{BaswanaGS18} (see Lemma 4.13, Theorem 4.2, Lemma 4.17), although we use different notation which is more amenable to the second part of the proof.

In the second part of the proof, which includes Sections \ref{subsec:proof-prob-matching} and \ref{subsec:proof-stronger}, we show that this independence allows us to prove Lemma \ref{lem:prob-matching}. This part is entirely new to our paper, because the claim does \emph{not} hold for the original algorithm of Baswana et al \cite{BaswanaGS18}. We show how our modifications of \cite{BaswanaGS18}, and especially our introduction of the $\ResetMatching$ procedure, allows us to replace the fundamentally amortized guarantees of \cite{BaswanaGS18} with the universal upper bound on the probability in Lemma \ref{lem:prob-matching}.

\subsubsection{Hierarchy Changes and Invariants During the Processing}
\label{subsec:hierarchy-invariants}
So far we have only concerned ourselves with the state of the hierarchy after it completes processing some adversarial update. But the algorithm can make many changes to the hierarchy while processing only a single adversarial update. In this section we will need notation to analyze the hierarchy in the middle of processing.

\begin{definition}
\label{dfn:changes}
We refer to a \emph{change} in the hierarchy as any of the following operations. All line numbers relate to Algorithm \ref{alg:matching 1} or \ref{alg:matching 2}.
\begin{enumerate}
	\item Removing an edge $(x,y)$ from the matching in line \ref{step:reset-delete} or \ref{step:rise-mate-delete} 
	\item Adding an edge $(x,y)$ to the matching in line \ref{step:settle-insert} 
	\item Moving a vertex from some level $i$ to some level $j>i$, and the associated bookkeeping changes in line \ref{step:rise-change} in \Rise($\cdot,i,j)$
	\item Moving a vertex from some level $i$ to level $i-1$, and the associated bookkeeping changes in line \ref{step:fall-change} in \Fall($\cdot,i$)
	\item Performing an adversarial insertion of edge $(x,y)$ in line \ref{step:insert-change}
	\item Performing an adversarial deletion of edge $(x,y)$ in line \ref{step:delete-change}
\end{enumerate}
Given any execution of our dynamic algorithm let $\sigma_1, \sigma_2, \ldots $ be the sequence of all hierarchy-changes made by the algorithm. 
\end{definition}

In Lemma~\ref{lem:invariants} we showed that right after the algorithm has completed the processing of some adversarial insertion/deletion, the hierarchy satisfies Invariants 2--4. Note, however, that if we look at the hierarchy right after some change $\sigma_i$, then Invariant 2 might be violated, since Algorithm \ref{alg:matching 1} might not yet have terminated, so there might still be free vertices that need fixing. 
Also Invariant~3 needs to be (slightly) relaxed to Invariant~3':
\begin{itemize}
\item \textbf{Invariant~3':} For each vertex $v$ and for all
$j > \LL(v)$, $\phi_v(j) \le 4^{j}$ holds true.
\end{itemize}
Thus, we now show that Invariants 3' and 4 are true in \emph{every} instantiation of the hierarchy, i.e., at any point in the algorithm. 

\begin{lemma}[Generalized Invariants] 
\label{lem:invariants-stronger}
Let $\hh$ be the hierarchy right after some change $\sigma_i$. Then, $\hh$ satisfies Invariants 3' and 4 above. (Note that Matching Property* certainly continues to hold because it corresponds to the behavior of the algorithm itself, not to any hierarchy invariant.)
\end{lemma}

\begin{proof}
For Invariant 4, the only time we insert an edge $(v,w)$ into the matching is in line \ref{step:settle-insert} of \GenericRandomSettle, and line \ref{step:settle-rise} ensures that when we do so we have $\LL(v) = \LL(w)$. Whenever a vertex falls in level it is a free vertex, so Invariant 4 is trivially preserved. Finally, a vertex only rises in level inside the \Rise operation, and line \ref{step:rise-mate-delete} ensures that we only perform the actual rise on a free vertex.

The proof for Invariant 3' is much more involved, though it is conceptually straightforward. Recall that Invariant 3' states that $\phi_v(i) \leq 4^i$ for all vertices $v$ and level $i$; see the beginning of Section \ref{subsec:modification} for the definition of $\phi_v(i)$. 

Define a hierarchy change to be \emph{risky} if it increases $\phi_v(i)$ for some pair $v,i$. Since $\phi_v(i)$ can only increase when $|\below{i}(v)|$ increases or when $v$ drops from level $i$ to $i-1$, it is easy to check that our algorithm only performs two types of risky hierarchy changes: (A) an adversarial insertion of edge $(u,v)$ and (B) the falling of a vertex from some level $j$ to level $j-1$. These changes to the hierarchy are made in Line \ref{step:insert-change} of Algorithm \ref{alg:matching 1} and Line \ref{step:fall-change} of Algorithm \ref{alg:matching 2} respectively.

Given a vertex $v$ and a level $i$ we say that pair $v,i$ is \emph{violating} if $\phi_v(i) \ge 4^i$. Note that only a risky operation can cause a pair to become violating. We now prove Invariant 3 via induction on the number of risky hierarchy changes. 

\emph{Induction Hypothesis:} Consider any risky hierarchy change $\sigma$. Then, right before $\sigma$, all pairs $v,i$ are non-violating; that is, we have $\phi_v(i) < 4^i$ for every vertex $v$ and level $i$. 

\emph{Induction Basis:} The proof of the base case is trivial: when the algorithm begins we have $\phi_v(i) = 0$ for every pair $v,i$, and this continues to hold before the first risky change because non-risky changes cannot increase any $\phi_v(i)$. 

\emph{Induction Step:} Assume that the induction hypothesis holds for some risky change $\sigma$. We now show that it also holds for the next risky change $\sigma'$. We consider two cases:

{\bf Case 1: $\sigma$ corresponds to Line \ref{step:fall-change} in Procedure \Fall($v,i$)}. Note that the algorithm performs no risky hierarchy changes between executing lines \ref{step:fall-change} and \ref{step:fall-end} (including all the sub-routines called in between). Thus, it is sufficient to show that by the time the algorithms finishes Procedure \Fall($v,i$) in Line \ref{step:fall-end} there are no violating pairs; this will then continue to hold until the next risky operation $\sigma'$ because the non-risky operations in between cannot create new violating pairs
or cause already fixed pair to become violating again. Note that because all the hierarchy changes made during \Fall($v,i$)
(e.g., the rising of vertices) are non-risky, it is in fact enough to show that every pair $w,j$ is or becomes non-violating at some point between lines \ref{step:fall-change} and \ref{step:fall-end}.

The effect of the hierarchy change $\sigma$ which causes $v$ to fall from level $i$ to level $i-1$ is to increase $\phi_w(i) = |\below{i}(w)|$ by $1$ for every neighbor $w$ that is in $\below{i}(v)$. By the induction hypothesis, we then have $\phi_w(i) \leq 4^i$ for all such pairs, while all other pairs remain non-violating. Now, for every $w \in \below{i}(v)$ the \Fall($v,i$) operation executes \IncrementPhi($w,i$) (Line \ref{step:fall-increment}). When the algorithm executes \IncrementPhi($w,i$), if $\phi_w(i) < 4^i$, then we are done, because as argued above it is enough to show that $w,i$ is non-violating at some point during the execution of \Fall($v,i$). If $\phi_w(i) = 4^i$, then the algorithm raises $w$ to level $i$ (threshold rise in \IncrementPhi($w,i$)), so after the rise we have $\LL(w) = i$ and $\phi_w(i) = 0 < 4^i$, as desired.

{\bf Case 2: $\sigma$ corresponds to Line \ref{step:insert-change} Procedure \Insert($u,v$)}. Note that the algorithm performs no risky hierarchy changes between Line \ref{step:insert-change} and Line \ref{step:insert-reset} (including subroutines called by lines in between). Thus, analogously to the previous case, it is enough to show that every pair $w,j$ is or becomes non-violating at some point between the execution of Line \ref{step:insert-change} and Line \ref{step:insert-reset}. The effect of the hierarchy change $\sigma$ which inserts edge $(u,v)$ is to increase $\phi_u(j) = |\below{j}(u)|$ and  $\phi_v(j) = |\below{j}(v)|$ for every $j > \max \{ \LL(u),\LL(v) \}$. We focus on the pairs $u,j$, since the pairs $v,j$ are analogous. For every such pair $u,j$ the algorithm executes \IncrementPhi($u,j$) (Line \ref{step:insert-rise}). As in Case 1, if $\phi_u(j) < 4^j$ when the algorithm executes \IncrementPhi($u,j$), then pair $u,j$ is non-violating and we are done. If $\phi_u(j) = 4^j$ then again as in Case 1, the algorithm raises $u$ to level $j$, which leads to $\phi_u(j) = 0 < 4^j$, so $j,u$ becomes non-violating.

\end{proof}

\subsubsection{Upper and Lower Hierarchy and Hierarchical Independence}
\label{subsec:higher-and-lower}

We now formally separate changes to the upper and lower parts of the hierarchy and then show the independence of the upper hierarchy from the lower hierarchy.

\begin{definition}
\label{dfn:changes-above}
We say that a hierarchy change $\sigma$ is \emph{above} level $\ell$ if one of the following holds:
\begin{enumerate}
	\item The change removes/adds an edge $(u,v)$ from/to the matching for which $\LL(u) > \ell$ (Recall that $\LL(u) = \LL(v)$ by Invariant 4.)
	\item \label{change-above:rise} The change raises a vertex from level $i$ to level $j>i$ with $j > \ell$ (it does not matter if $i > \ell$.)
	\item The change moves a vertex from level $i$ to level $i-1$ with $i > \ell$.
	\item The change is an adversarial insertion/deletion of edge $(x,y)$ (regardless of $\LL(x)$ and $\LL(y)$).
\end{enumerate}

For any execution of the dynamic algorithm, let $S = \sigma_1, \sigma_2, \ldots $ be the sequence of changes made to the hierarchy. Let $S^\ell = \sigma^\ell_1, \sigma^\ell_2, \ldots$ be the subsequence of $S$ consisting of all changes at level $> \ell$: that is $S^\ell$ contains all $\sigma_i$ above level $\ell$, in the same order as in $S$.
\end{definition}

To prove independence of the hierarchy above level $\ell$, we will think of the algorithm as using two different random bit-streams.

\paragraph{Defining higher and lower random bits:}
Let $\aa$ be the sequence of updates made by the adversary: since the adversary is oblivious, we can fix this sequence in advance. Let $B$ be the entire sequence of random bits used by the algorithm. If the algorithm is given all of $(\aa, B)$ as input, it will always produce the same hierarchy.
	
Now, let us conceptualize the same algorithm a bit differently. Let $\ell$ be the fixed level in the statement of Lemma \ref{lem:prob-matching}. We will have two sequences of random bits: $B_{>\ell}$ and $B_{\leq \ell}$. Whenever the algorithm runs \GenericRandomSettle{$v, k$} for any $v \in V$, if $k>\ell$ then it chooses the new random mate for $v$ using the bits in $B_{>\ell}$; otherwise, it uses the bits in $B_{\leq \ell}$. Similarly, whenever the algorithm executes \Rise($v,\LL(v),k$) with probability $\prise_k$, the random bits for $\prise_k$ are taken from $B_{>\ell}$ if $k > \ell$, and from $B_{\leq \ell}$ otherwise. Finally, whenever the algorithm executes \ResetMatching{$v$} with probability $\preset_{\LL(v)}$, the random bits for $\preset_{\LL(v)}$ are taken from  $B_{>\ell}$ if $\LL(v) > \ell$ and from $B_{\leq \ell}$ otherwise.
	
Note that the execution of the algorithm is completely determined by the triplet $(\aa, B_{>\ell},B_{\leq \ell})$. Previously, we only fixed the update sequence $\aa$, and assumed the bits in $B_{>\ell},B_{\leq \ell}$ were chosen randomly. 
We now modify this as follows.
Given any fixed sequence of bits $B^+$, we say that \emph{the algorithm run with $B_{>\ell}$ set to $B^+$} if the bits from $B_{>\ell}$ are always taken from $B^+$, but the bits from $B_{\leq \ell}$ are chosen randomly. Then, when we speak of probabilities in such an execution, the probability is only over the bits in $B_{\leq \ell}$.

The lemma and its corollary below formally states the independence of the hierarchy above level $\ell$. Intuitively, the lemma says the following. Fix a sequence of adversarial updates $\aa$ and a sequence of bits $B^+$, and say that we run our dynamic matching algorithm with $B_{>\ell}$ set to $B^+$. Then, the sequence $S^\ell$ of changes above $\ell$ is \emph{deterministically} determined by $B^+$; in other words, the sequence will always be the same regardless of the random bits in $B_{\leq \ell}$. There is one caveat: when a vertex $v$ rises from level $i$ to $j > \ell$ (change type \ref{change-above:rise} in Definition \ref{dfn:changes-above}), although $v$ and $j$ are always deterministically determined by $B^+$, $i$ may depend on $B_{\leq \ell}$ if $i \leq \ell$.

\begin{lemma} [Hierarchical Independence]
	\label{lem:above-independence}
Let $B_1, B_2$ be any two bit sequences for $B_{\leq \ell}$. Let $\sigma^\ell_1, \ldots $ be the sequence of changes above $\ell$ performed by the execution that uses $\aa, B^+, B_1$ and let $\tau^\ell_1, \ldots $ be the sequence of changes above $\ell$ performed by the execution that uses $\aa, B^+, B_2$. Then, every $\sigma^\ell_k$ is equivalent to $\tau^\ell_k$ in the following sense:
\begin{itemize}
	\item If $\sigma^\ell_k$ moves vertex $v$ from level $i$ to level $j > i$ (with $j > \ell$), then $\tau^\ell_k$ moves the same vertex $v$ from level $i'$ to the same level $j > i'$; moreover, $i \neq i'$ is only possible if $i \leq \ell$ and $i' \leq \ell$.
	\item If $\sigma^\ell_k$ is any other type of change above $\ell$, then $\sigma^\ell_k$ is identical to $\tau^\ell_k$.
\end{itemize}
\end{lemma}

\begin{proof}
Although it is somewhat involved, conceptually speaking, the proof is quite simple: we go through the possible operations of the algorithm and show by induction that because both algorithms use the same bits $B^+$ for $B_{>\ell}$, the two executions always have the same above-$\ell$ hierarchy.

Intuitively, the induction hypothesis is that all times, both executions have the same above-$\ell$ hierarchy as well as the same queue $Q$ of free vertices above level $\ell$ (see Algorithm \ref{alg:matching 1}). The issue is that when we say ``at all times", this does not refer to execution time, since one algorithm may be ahead of the other. Instead, as in the lemma statement, we formalize the intuition by talking about the sequence of above-$\ell$ changes and of changes to $Q$.

\paragraph{Defining Relevant Operations:}
Let $A_i$ be the execution of the algorithm running on $\aa, B^+, B_i$ for $i=1,2$. 
Consider the following sequence  of operations $\gamma_1, \gamma_2, \ldots$  consisting of
(i) above-$\ell$ changes  and (ii) changes  (addition/removal/move-to-end) to the queue formed by this execution. Whenever the execution makes an above-$\ell$ hierarchy-change, we add this change to the end of the sequence. Similarly, whenever the execution adds or removes a vertex $v$ to queue $Q$ or moves $v$ to end of $Q$ with $\LL(v) > \ell$, we add this change to the end of the sequence. 
We refer to the $\gamma_i$ as \emph{relevant} operations.
Let $A_2$ be the execution running on $\aa, B^+, B_2$ and define relevant operations $\delta_1, \delta_2, \ldots$ analogously for this execution. We say that $\gamma_i = \delta_i$ if one of the following holds:
\begin{itemize}
	\item Both $\gamma_i$ and $\delta_i$ add/remove/move-to-end the same vertex $v$ in $Q$ and $v$ has the same level when $\gamma_i$ is performed in execution $A_1$ as when $\delta_i$ is performed in execution $A_2$.
	\item  $\gamma_i$ and $\delta_i$ correspond to the same non-rise hierarchy change.
	\item Both $\gamma_i$ and $\delta_i$ raise the same vertex $v$ to the same level above $\ell$.
\end{itemize}

\begin{claim}[Main Claim] For every $i$ we have $\gamma_i = \delta_i$.
\end{claim}

It is easy to verify that the main claim proves Lemma \ref{lem:above-independence}, as it implies that the above-$\ell$ hierarchy and the above-$\ell$ free vertices evolve in the same way, which is strictly stronger than the lemma statement, which only concerns the hierarchy changes. We prove the claim by induction. 

\emph{Induction Hypothesis:} For any $i$, for all $j \le i$ it holds that $\gamma_j = \delta_j$ and that both executions $A_1$ and $A_2$ have used the same number of bits from $B^+$.

\emph{Induction Bases:}
The base case is trivially true since the graph starts empty, so $\gamma_1$ and $\delta_1$ both correspond to the same adversarial insertion. 

\emph{Induction Step:}
To show the inductive step we will consider somewhat larger chunks of the algorithm. 
We note that all relevant operations performed by the algorithm -- whether hierarchy changes or changes to $Q$ -- occur in one of three scopes in Algorithm \ref{alg:matching 1}.

\begin{definition}
Define the delete-body of Algorithm \ref{alg:matching 1} to consist of the execution of all lines in Procedure \Delete{$u,v$} \emph{except} \ProcessQueue{} in Line \ref{step:delete-process-queue}, including all subroutines called during the execution of these lines. Define the insert-body to consist of the execution of all of Procedure \Insert{$u,v$} except \ProcessQueue{} in Line \ref{step:insert-process-queue}, including all subroutines called during the execution of these lines. Finally, define the fix-body to contain lines \ref{step:pop}-\ref{step:fix}, including all subroutines called during the execution of these lines. We say that the fix-body is above level $\ell$ if the vertex $v$ popped from $Q$ has $\LL(v) > \ell$.
\end{definition}

\begin{definition}
We say that a relevant operation $\gamma_i$ or $\delta_i$ is \emph{primary} if it is an adversarial insertion, an adversarial deletion, or the removal of a vertex $v$ from $Q$ with $\LL(v) > \ell$ (this last change always occurs in Line \ref{step:pop}). (Note that only a removal from $Q$ counts as a primary change: not an insertion or a move in $Q$.)
\end{definition}

\begin{observation} All relevant operations performed by the algorithm occur within either an insert-body, a delete-body, or a fix-body above level $\ell$. Moreover, each of these three bodies always begins with a primary change $\gamma_i$, and each primary change initiates the corresponding body of the algorithm.
\end{observation}
	
By the observation above, the proof of the Main Claim consists of three possible cases summarized in the following claim.

\begin{claim}
	\label{clm:delete/insert}
Assume that $\gamma_k = \delta_k$ is an adversarial deletion or insertion of some edge $(u,v)$, that the main claim holds for every $i \leq k$,
and that both executions $A_1$ and $A_2$ have used the same number of bits from $B^+$ up to operation $\gamma_k$, resp. $\delta_k$. Let $\gamma_{q}$ be the next primary change performed by execution $A_1$. Then $\gamma_r = \delta_r$ for each $r \in [k,q]$; that is, the main claim holds for every $i \leq q$,
and that and that both executions $A_1$ and $A_2$ have use the same number of bits from $B^+$ between operation up to operation $\gamma_q$, resp.~$\delta_q$.

The same holds if $\gamma_k = \delta_k$ pops some vertex $v$ from $Q$ with $\LL(v) > \ell$.
\end{claim}

It is easy to check that the claim above proves the Main Claim. We first start with a definition and a couple observations.

\begin{definition}
	For any $j$, let $\hh^1_j$ and $Q^1_j$ be the above-$\ell$ hierarchy of execution $A_1$ and the (ordered) queue of above-$\ell$ vertices after $\gamma_1, \ldots, \gamma_j$ have been performed. Define $\hh^2_j$ and $Q^2_j$ analogously.
\end{definition}

\begin{observation}
	\label{obs:identical} Say that the main claim holds for all $i \leq k$. Then, $\hh^1_{k} = \hh^2_{k}$ and $Q^1_{k} = Q^2_{k}$. Moreover, if a vertex $v$ has $\LL(v) > \ell$ in $\hh^1_{k} = \hh^2_{k}$, then the bit $\responsible(v)$ is the same in both hierarchies. 
\end{observation}

\begin{observation}
	\label{obs:fix}
Consider any fix-body that is \emph{not} above level $\ell$. Then, during the execution of these lines, no relevant changes can occur. This observation can be verified by looking at the flow of our algorithm, and observing that fixing a vertex at level $\leq \ell$ can only make hierarchy changes at level $\leq \ell$ and add/move vertices to $Q$ with level $\leq \ell$.
\end{observation}
	
\begin{proof}[Proof of Claim \ref{clm:delete/insert}]
(A) We first show the claim for a deletion as
this is the simplest case. Note that by Observation \ref{obs:identical}, whether or not $(u,v) \in \mathcal{M}$ will be the same in both executions. Thus, either both executions add $u$ and $v$ to the end of $Q$, or they make no relevant changes, so both executions make the same relevant changes in the delete-body, and so continue to have the same above-$\ell$ hierarchy and queue. 

Now, if the above-$\ell$ queue is empty at the end of this delete-body, then by Observation \ref{obs:fix}, the next primary change in both executions will be an adversarial insertion/deletion from $\aa$, which is clearly the same in both executions. Else, since $Q$ is the same in both executions after the delete-body, if $v$ is the first vertex in $Q$ with $\LL(v) > \ell$, then the popping of $v$ will be the next primary operation in both executions. Either way, the next primary operation $\gamma_q = \delta_q$ is the same in both executions  and neither uses
any bits of
$B^+$.

(B) We next show the claim for the case of an insertion.
In Line \ref{step:insert-begin-rise} the algorithm computes $j = \max \{\LL(u),\LL(v)\}$. Note that since the above-$\ell$ hierarchies are the same in both executions when the insert-body begins (Observation \ref{obs:identical}), if $j > \ell$ then $j$ will be the same in both executions. On the other hand, if $j \leq \ell$, then no relevant operations will occur in lines \ref{step:insert-begin-rise}-\ref{step:insert-end-rise} and we only have to analyze the \ResetMatching operations in Line \ref{step:addition/removal/move-to-end} and Line \ref{set:addition/removal/move-to-end2} (see below). Thus, we can assume that $j > \ell$ and that both executions now execute lines \IncrementPhi{$v,j$} and \IncrementPhi{$u,j$}. We now focus on \IncrementPhi{$v,j$}; the argument for $u$ is identical.
	
The first line of \IncrementPhi{$v,j$} checks if $\below{j}(v) \geq 4^j$: since the above-$\ell$ hierarchies of $A_1$ and $A_2$ are the same and since $j > \ell$, it is follows that $\below{j}(v)$ is the same in both executions. Thus, either both executions will perform the same threshold rise or neither will. The next line of \IncrementPhi{$v,j$} performs a rise with probability $\prise_i$. Since $j > \ell$, the bits used to determine this probabilistic rise come from $B^+$, and so are the same in both executions. Note that both executions access the same bit from $B^+$ as by the induction hypothesis both executions have looked at the same bits in $B^+$ so far. Thus, again, either both operations perform the same rise and the corresponding relevant changes or neither do. They also consume the same number of bits from $B^+$.

Afterwards they will use the same bits in $B^+$ to decide whether to execute a \ResetMatching{$u$} and/or \ResetMatching{$v$}.  Within a \ResetMatching{$w$} exactly the same operations will be executed as the same vertices are responsible for an edge (Observation \ref{obs:identical}). 

Thus the above-$\ell$ hierarchy and queue continue to be the same in both executions throughout the insert-body and the number of bits of
$B^+$ that are used is identical. By the same argument as in~(A), the next primary operation will also be the same.

(C) Finally we show the claim if $\delta_k$ pops some vertex $v$ from $Q$ with $\LL(v) > \ell$.
In this case, the algorithm executes \FixFreeVertex{$v$} in both executions. 
Recall that we assume in the lemma statement that $i = \LL(v) > \ell$. Thus, since the above-$\ell$ hierarchies are the same at the beginning of the fix-body (Observation \ref{obs:identical}), we know that when we compute sets $\below{i}(v)$ and $\below{i+1}(v)$ in Line \ref{step:fix-compute-sets} of Algorithm \ref{alg:matching 2}, each set will be the same in both executions. Thus, either they will both execute the If statement of Line \ref{step:fix-if} or they will both execute the Else statement of Line \ref{step:fix-else}.

If they both execute the If statement, then they both execute \GenericRandomSettle{$v,i$}. Note that the random mate $w$ picked by this \GenericRandomSettle will be the same in both executions because $i > \ell$, so $w$ is picked according to bits in $B^+$, which are the same in both executions. Thus, in the case of the If statement, it is easy to check that all relevant operations performed by the fix-body will be the same in both executions and the number of consumed bits of $B^+$ is identical.
	
Now, say that the algorithm instead executes \Fall{$v,i$} from the Else statement. Then the algorithm performs \IncrementPhi{$w,i$} for every $w \in N_{<i}(v)$. Because $N_{<i}(v)$ are the same in both executions, the same operations \IncrementPhi{$w,i$} are performed. By the same argument as in~(B), these \IncrementPhi{$w,i$} then lead to the same relevant operations $\gamma_r = \delta_r$ in both executions and use the same number of bits from $B^+$. The algorithm then performs \ResetMatching{$w$} for each $w \in N_{=i-1}(v)$. Once again, if $i-1 \leq \ell$ then none of these operations are relevant, so the claim trivially holds. Otherwise $i-1 > \ell$ and since the above-$\ell$ hierarchy is the same in both executions up to this point, we have that $N_{=i-1}(v)$ is the same in both executions and they will both use the same bits and the same number of bits from $B^+$ to decide whether to perform a \ResetMatching{$w$}. Within a \ResetMatching{$w$} exactly the same operations will be executed as the same vertices are responsible for an edge. Thus, they will have identical results in both executions.

We have thus shown that throughout the fix-body, both executions perform the same relevant operations, have the same above-$\ell$ hierarchy and queue, and use the same number of bits from $B^+$. By the same argument as in (A), the next primary operation will also be the same. 

\end{proof}

We have thus proved the main claim and completed the proof of Lemma \ref{lem:above-independence}.

\end{proof}

Throughout our analysis, we will rely on the corollary below, which is an extension of Lemma \ref{lem:above-independence}. 
We start with an informal description.

\paragraph{Informal Description of  Corollary \ref{cor:above-independence}:} Say that we run our algorithm on a fixed update sequence~$\aa$, and that at some point during the algorithm we call procedure $\GenericRandomSettle(x,\ell)$. Let $N_{<\ell}(x) = \{y_1, \ldots y_k\}$ when we call the procedure. Let $Y_i$ be the event that $\GenericRandomSettle(x,\ell)$ picks $y_i$ as the mate of $x$; since the mate of $x$ is chosen uniformly at random from $N_{<\ell}(x)$, we know that $\PP[Y_i] = 1/k$. This statement remains true no matter what happened before this execution of $\GenericRandomSettle(x,\ell)$, since $\GenericRandomSettle$ uses fresh randomness. However, say that $\EE$ is some event which depends on the \emph{entire} sequence of hierarchy changes made by the algorithm, including those after the call to $\GenericRandomSettle$. In this case, we might have  $\PP[Y_i \cond \EE] \neq \PP[Y_i]$; for example, since $\EE$ depends on the entire sequence, it might be the case that $\EE$ can be true only if $Y_i$ is true. The crux of our corollary is that if we consider an event $\EE_{>\ell}$ that depends on all hierarchy changes before the call to $\GenericRandomSettle(x,\ell)$ and also on future above-$\ell$ changes made by the algorithm, then we can indeed state that $\PP[Y_i \cond \EE_{>\ell}]  = \PP[Y_i] = 1/k$. The same is true of a call to $\ResetMatching(x)$, as long as $\LL(x) \leq \ell$ when the call is made.

\begin{corollary}
	\label{cor:above-independence}
	Fix some adversarial update sequence $\aa$, and consider the execution of the algorithm on $\aa$.
	The following two properties hold:
	
	\begin{itemize}
		\item Consider some call to $\GenericRandomSettle(x,\ell)$, let $N_{<\ell}(x) = \{y_1, \ldots y_k\}$ when the call is made and let $Y_i$ be the event that $y_i$ is chosen as $\mate(x)$ as a result of the call.
Let $\spast$ be the sequence of \emph{all} hierarchy changes before this call to $\GenericRandomSettle(x,\ell)$ and let $S^\ell$ be the sequence of all above-$\ell$ hierarchy changes after the call.  Let $\EE_{>\ell}$ be any event that depends only on $\spast$ and $S^\ell$; that is, whether $\EE_{>\ell}$ is true or false is uniquely determined by these sequences. Then, for all $1 \leq i \leq k$, $\PP[Y_i \cond \EE_{>\ell}, \spast] = \PP[Y_i \cond \spast] = 1/k$. 
		
		\item Assume that at some point we reach a line that executes $\ResetMatching(u)$ with some probability $p$ (line \ref{step:addition/removal/move-to-end} of Algorithm \ref{alg:matching 1} or Line \ref{step:fall-reset} of Algorithm \ref{alg:matching 2}) and that at this point, $\LL(u) \leq \ell$. As above, let $\spast$ be the sequence of all hierarchy changes before this line, let $S^\ell$ be the sequence of all above-$\ell$ hierarchy changes after the call and let $\EE_{>\ell}$ be any event that depends only on $\spast$ and $S^\ell$. Let $Y$ be the event that the $\ResetMatching(u)$ is in fact executed. Then $\PP[Y \cond \EE_{>\ell},\spast] = \PP[Y \cond \spast] = p$. 
	\end{itemize}
\end{corollary}

\begin{proof} 
	The proof follows easily from Lemma \ref{lem:above-independence}. We will only prove the first property about $\GenericRandomSettle(x,\ell)$; the property about $\ResetMatching(x)$ can be proved in the same fashion.
	
	Since $\GenericRandomSettle(x,\ell)$ picks a random mate using fresh randomness that is independent from all previous choices made by the algorithm, we have $\PP[Y_i \cond \spast] = 1/k$. 
	By Bayes' law, we have $$\PP[Y_i \cond \EE_{>\ell},\spast] = \PP[Y_i \cond \spast] \frac{\PP[\EE_{>\ell} \cond Y_i,\spast]}{\PP[\EE_{>\ell} \cond \spast]}.$$
	
	We now show that $\PP[\EE_{>\ell} \cond Y_i,\spast] = \PP[\EE_{>\ell} \cond \spast]$, which will complete the proof. As in the setup of Lemma \ref{lem:above-independence}, we can think of the algorithm as running on bit streams $B_{>\ell}$ and $B_{\leq \ell}$. Since $\EE_{>\ell}$ depends only on $\spast$, the adversarial updates $\aa$ and above-$\ell$ changes, we know from Lemma \ref{lem:above-independence} that whether or not $\EE_{>\ell}$ is true depends only on $\spast, \aa$ and the bits in $B_{>\ell}$. 
	We will next show that $Y_i$ only depends on bits of $B_{\leq \ell}$ that are independent from $\spast$. All the bits in $B_{\leq \ell}$ are by definition independent from $B_{>\ell}$, and, by the obliviousness of the adversary, from $\aa$. Thus, $Y_i$ is an event that depends only on random bits that are independent from the random bits/events that $\EE_{>\ell}$ depends on, so we have $\PP[\EE_{>\ell} \cond Y_i,\spast] = \PP[\EE_{>\ell} \cond \spast]$, as desired.
	
	It remains to show that $Y_i$ only depends on bits of $B_{\leq \ell}$ that are independent from $\spast$. This follows from the fact that $Y_i$ is determined by the call to $\GenericRandomSettle(x,\ell)$ made after all the changes in $\spast$ have already occurred. Because the call is at level $\ell$, the bits come from $B_{\leq \ell}$; because the call uses fresh randomness, these bits are independent from $\spast$.  
\end{proof}

\subsubsection{Proof of Lemma \ref{lem:prob-matching}}
\label{subsec:proof-prob-matching}
With the hierarchical independence in place, we are now ready to begin the proof of Lemma \ref{lem:prob-matching}. We will actually prove a slightly stronger statement, which shows that only the bits in $B_{\leq \ell}$ need to be random for the lemma to hold; the bits in $B^+$ can be chosen adversarially.

\begin{lemma} [Stronger Version of Lemma \ref{lem:prob-matching}]
\label{lem:prob-matching-stronger}
Let $0\leq \ell \leq \floor{\log_4(n)}$ be any level in the hierarchy. 
Let $\aa$ be the sequence of adversarial updates, and fix \emph{any} bit stream $B^+$. Consider running our dynamic matching algorithm with $B_{>\ell}$ set to $B^+$.
Let $(x,y)$ be any edge at any time $\tstar$ during the update sequence.
Then: $\PP[$at time $\tstar$, $(x,y)$ is a matching edge at level $\ell$ and $\responsible(x)$ is True$]$ $= O(\log^3(n)/4^\ell)$, where the probability is over all random bits in $B_{\leq \ell}$.
\end{lemma}

We will proceed as follows: We will define a special type of hierarchy change, called pivotal change, and will first show  (Lemma~\ref{lem:pivotal-good}) that if a pivotal change exists before time $\tstar$, then the desired statement of Lemma~\ref{lem:prob-matching-stronger}
holds. Next (Lemma~\ref{lem:pivotal-exists}) we show that a pivotal change exists before time $\tstar$ with high probability.
To do we focus our attention on relevant hierarchy changes, which are formalized in the definitions of $(\ell,v)$-critical changes and $(\ell,v)$-reset-opportunities below.

For the rest of this section, we set $\alpha_\ell = 4000 \cdot \log(n) \cdot 4^\ell$ and $\beta_\ell = \alpha/4 = 1000 \cdot \log(n) \cdot 4^\ell$.
	
\begin{definition}
Let $x,y,\tstar,\ell,B^+$ be the variables from the statement of Lemma \ref{lem:prob-matching-stronger}. Consider some execution of the algorithm with $B_{>\ell}$ set to $B^+$. Let $\elemma$ refer to  the property that $(x,y)$ is a matching edge at level $\ell$ and that $\responsible(x)$ is true and let $\elemma_{\tstar}$ be the event that $\elemma$ holds at time~$\tstar$. Note that Lemma \ref{lem:prob-matching-stronger} is equivalent to the statement that $\PP[\elemma_{\tstar}] = O(\log^3(n)/4^\ell)$, where the probability is over all random bits in $B_{\leq \ell}$.
\end{definition}

To characterize which changes in the hierarchy can lead to changes in the matching we introduce the following definition.
\begin{definition}
	\label{dfn:critical}
For any vertex $v$ and any level $\ell$, we say that a hierarchy change $\sigma$ is an $(\ell,v)$-critical change if $\sigma$ is a change above $\ell$ \emph{and} one of the following holds
\begin{enumerate}
	\item $\sigma$ changes $\LL(v)$. \label{critical:level}
	\item For some neighbor $w$ of $v$, $\sigma$ changes $\LL(w)$ from $\ell+1$ to $\ell$. \label{critical:fall}
	\item For some neighbor $w$ of $v$, $\sigma$ raises $w$ from a level $\leq \ell$ to a level $> \ell$. \label{critical:rise}
	\item $\sigma$ is an adversarial insertion/deletion of some edge $(v,w)$ incident to $v$ (regardless of level). \label{critical:adversary}
\end{enumerate}
\end{definition}

We need this definition for the following reason: only a $(\LL(v),v)$-critical change can make the matched edge $(v,u)$ with responsible $v$ unmatched, as shown in the next lemma.
\begin{lemma}\label{lem:reasons for becoming unmatched}
A matched edge $(u,v)$ with responsible $v$ becomes unmatched only after a $(\LL(v),v)$-critical change.
\end{lemma}
\begin{proof}
A matched edge $(u,v)$ with responsible $v$ becomes unmatched only if (a) it is deleted, 
(b) one of its endpoints $u$ or $v$ moves to a higher level, or (c) a $\ResetMatching(v)$
was executed. Note that a $\ResetMatching(u)$ would have no effect as $u$ is not responsible for the matched edge.

Furthermore $\ResetMatching(v)$ is only called if a neighbor of $v$ falls from level 
$\LL(v)+1$ to level $\LL(v)$ or an edge incident to $v$ is inserted.

We can summarize this as follows:
A matched edge $(u,v)$ with responsible $v$ becomes unmatched only if (a) an edge incident to $v$ is inserted or if the matched edge incident to
$v$ is deleted, 
(b) either $u$ or $v$ moves to a higher level, or (c) a neighbor of $v$ falls from level 
$\LL(v)+1$ to level $\LL(v)$.
Note that all of these changes are $(\LL(v),v)$-critical.
Thus a matched edge $(u,v)$ with responsible $v$ becomes unmatched only after a $(\LL(v),v)$-critical change.
\end{proof}

We also need to characterize after what type of hierarchy changes a $\ResetMatching$-operation can be executed. This is the reason for the following definition.
\begin{definition}\label{dfn:reset opportunity}
For any vertex $v$ and level $\ell$, we say that a hierarchy change $\sigma$ is an $(\ell,v)$-reset-opportunity if one of the following holds:
\begin{enumerate}
	\item For some neighbor $w$ of $v$, $\sigma$ changes $\LL(w)$ from $\ell+1$ to $\ell$. \label{reset:neighbor}
	\item $\sigma$ is an adversarial insertion of some edge $(v,w)$ incident to $v$ (regardless of level). \label{reset:adversary}
\end{enumerate}
\end{definition}

\begin{lemma}
A $\ResetMatching(v)$ operation is executed only after a $(\LL(v),v)$-reset opportunity.
\end{lemma}
\begin{proof}
The operation $\ResetMatching(v)$ is called either after an adversarial insertion of an edge
incident to $v$ or if a neighbor of $v$ drops from level $\LL(v)+1$ to $\LL(v)$.
\end{proof}

\begin{observation} \label{obs:reset}
If $v$ is at level $\ell$ with $\responsible(v)$ set to True, and the algorithm performs a change $\sigma$ which is a $(\ell,v)$-reset opportunity, then with probability $\preset_\ell = 1/4^{\ell+3}$ the algorithm does \ResetMatching{$v$}. (See line \ref{step:fall-reset} of Algorithm \ref{alg:matching 2} and line \ref{step:insert-reset} of Algorithm \ref{alg:matching 1}.)
\end{observation}

\subsubsection{Intuition for the Proof of Lemma \ref{lem:prob-matching-stronger}} 
\label{subsec:proof-stronger}
With all our definitions in place, we now give a brief intuition for the proof of Lemma \ref{lem:prob-matching-stronger}. Let $\sigma_1, \ldots, \sigma_{\qstar}$ be all the hierarchy changes performed by the algorithm up to time $\tstar$. Recall that we want to bound the probability that $\elemma$ holds after change $\sigma_{\qstar}$. Let's focus on any change $\sigma_i$, and consider two cases. 

{\bf Case 1:} there are many $(\ell,x)$-critical changes between $\sigma_i$ and $\sigma_{\qstar}$. 
Note that whether or not we fall in Case 1 is independent of any random choices made after $\sigma_i$ (recall that only bits from $B_{\leq \ell}$ are random), because these $(\ell,x)$-critical changes are by definition changes above~$\ell$, so we can apply Corollary \ref{cor:above-independence}. 
We will show that in Case 1, either one of these critical changes causes $x$ to change level (and hence to pick a new mate), or they will cause so many $(\ell,x)$ reset-opportunities that by Observation \ref{obs:reset} there is a high probability that the algorithm will perform a \ResetMatching{$x$} before change $\sigma_{\qstar}$. Either way, $x$ will pick a new random mate at some point between $\sigma_i,\ldots, \sigma_{\qstar}$, so the matching at change $\sigma_i$ bears no relevance to the matching at change $\sigma_{\qstar}$. Any $\sigma_i$ in Case 1 can thus be effectively ignored.

{\bf Case 2:} there are few $(\ell,x)$-critical changes between $\sigma_i$ and $\sigma_{\qstar}$. For simplicity, let us assume that none of these  $(\ell,x)$-critical changes change the level of $x$, and when $\sigma_i$ is performed, $\elemma$ is false. For $\elemma$ to become true, some hierarchy change between $\sigma_i$ and $\sigma_{\qstar}$ must choose $(x,y)$ as the matching. By Matching Property*, we know that any one particular \GenericRandomSettle{$x,\ell$} has only a small chance of picking $(x,y)$; to complete the proof we will show that because we are in Case 2, there are (with high probability) few executions of \GenericRandomSettle{$x,\ell$} between $\sigma_i$ and $\sigma_{\qstar}$. The reason there are few executions is that if \GenericRandomSettle{$x,\ell$} picks some matching edge $(x,z)$, then the only way the algorithm calls a new \GenericRandomSettle{$x,\ell$} is if the edge $(x,z)$ is removed from the matching due to an adversarial deletion or a hierarchy change. We will show that this can only occur due to a $(\ell,x)$-critical change, and that moreover, any $(\ell,x)$-critical change is unlikely to affect $(x,z)$ because $z$ was chosen at random from among many choices, and the $(\ell,x)$-critical changes are independent of this random choice (Corollary \ref{cor:above-independence}). Since there are few $(\ell,x)$-critical changes remaining (because we are in Case 2), they are unlikely to cause many executions of \GenericRandomSettle{$x,\ell$}.

\paragraph{Formalizing the Proof}

We now define the notion of a pivotal change, which corresponds to a hierarchy change that satisfies the assumptions of Case 2 in the above intuition.

\begin{definition} 
	\label{dfn:pivotal}
Define $x,y,\ell,B^+,\tstar$ as in the statement of Lemma \ref{lem:prob-matching-stronger}. Let $\sigma_1, \ldots, \sigma_{\qstar}$ be the sequence of hierarchy changes up to time $\tstar$. Consider some execution of the dynamic matching algorithm with $B_{>\ell}$ set to $B^+$. We say that some  change $\sigma$ performed by the algorithm is \emph{pivotal for time $\tstar$} if it satisfies \emph{all} of the following properties. 
\begin{enumerate}
	\item \label{pivotal:number} The number of $(\ell,x)$-critical changes between $\sigma$ and $\sigma_{\qstar}$ is at most $\alpha_\ell$. (Recall that $\alpha_\ell = 4000 \cdot \log n \cdot 4^{\ell}$.)
	\item \label{pivotal:level} There are no $(\ell,x)$-critical changes between $\sigma$ and $\sigma_{\qstar}$ that alter the level of~$x$. (Note that $x$ may still move levels due to hierarchy changes that are not above~$\ell$.)
	\item \label{pivotal:false} $\elemma$ is false right after change $\sigma$.
\end{enumerate}
Technical note: the lemmas and definitions are made cleaner if at the very beginning of the algorithm (when the graph is still empty) we insert a dummy update $\sigmadummy$ that does nothing; this is solely to allow for the possibility that the so-to-speak 0th update $\sigmadummy$ is itself pivotal.
\end{definition}

The crucial property of a change $\sigma$ that is pivotal for time $\tstar$ is  as follows: whether $\sigma$ is pivotal for $\tstar$ or not only depends on (i) the hierarchy at 
the time of change $\sigma$, (ii) the above-$\ell$ changes between $\sigma$ and time $\tstar$, and (iii) adversarial updates.
\begin{lemma}\label{lem:pivotal-prop}
Consider a change $\sigma= \sigma_i$ with $i < \tstar$. Whether $\sigma$ is pivotal for time $\tstar$ is uniquely determined by
(i) the hierarchy at 
the time of change $\sigma$ including the $\responsible$ bits, (ii) the above-$\ell$ changes between $\sigma$ and time $\tstar$, and (iii) adversarial updates.
\end{lemma}
\begin{proof}
We will show that each of the properties of a pivotal change  depend only on the hierarchy at the time of change $\sigma$, the above-$\ell$ changes between $\sigma$ and time $\tstar$, and the adversarial updates.

For Property (1) note that whether a change is $(\ell,x)$-critical only depends on whether the change is above $\ell$ and whose node's level it changes (if any), and whether it is an adversarial update. All this can be determined for a change up to time $\tstar$ if the information in (i)--(iii) is known. Thus, with information (i)--(iii) it is uniquely determined whether there are at most $\alpha_{\ell}$ $(\ell,x)$-critical changes between $\sigma$ and $\sigma_{q^*}$, i.e., whether Property (1) holds.

Property (2) guarantees that any $(\ell, x)$-critical changes that occur after $\sigma$ 
and up to time $\tstar$
\emph{must} be of type (2)--(4) of the definition of a critical change. As for Property (1) whether a change is $(\ell,x)$-critical and whether it changes the level of $x$ is uniquely determined by information (i)--(iii).

Property (3) of a pivotal change, i.e., whether $\elemma$ is false right after $\sigma$ 
is uniquely determined by  the hierarchy right before change $\sigma$ including the $\responsible$ bits as well as on $\sigma$.
\end{proof}
Note that it follows from the lemma that whether a change $\sigma_i$ with $i < \tstar$ is pivotal depends only on information (i)---(iii) from the lemma. 
We now proceed as follows. Lemma \ref{lem:pivotal-good} will show that Lemma \ref{lem:prob-matching-stronger} holds if we condition on the existence of a pivotal change for $\tstar$; this corresponds to Case 2 in the intuition section above. Lemma \ref{lem:pivotal-exists} will then show that a pivotal change exists with high probability: this corresponds to the intuition above that Case 1 can be effectively ignored, and only Case 2 is relevant.

\begin{lemma}
	\label{lem:pivotal-good}
	Define $x,y,\ell,B^+,\tstar$ as in the statement of Lemma \ref{lem:prob-matching-stronger}.  Consider some execution of the dynamic matching algorithm with $B_{>\ell}$ set to $B^+$, and condition on the fact that during the execution of the algorithm up to time $\tstar$ the algorithm encounters a change $\sigmapivot$ that is pivotal for time $\tstar$. Then, $\PP[$at time $\tstar$, $(x,y)$ is a matching edge that was chosen by some \GenericRandomSettle{$x,\ell$}$]$ $= O(\log^3(n)/4^\ell)$, where the probability is over all random bits in $B_{\leq \ell}$.
\end{lemma}

\begin{proof}
	Let $\sigma'_1, \ldots, \sigma'_\alpha$ be the sequence of $(\ell,x)$-critical changes between change $\sigmapivot$ and time $\tstar$. By definition of a pivotal change for $\tstar$, we have that $\alpha \leq \alpha_\ell$, and that none of the $\sigma'_i$ change the level of $x$. 
	
Recall that by definition of a pivotal change, $\elemma$ is false after change $\sigmapivot$. Thus, in order for $\elemma$ to become true, at some point between $\sigmapivot$ and $\tstar$ the algorithm must call \GenericRandomSettle{$x,\ell$}, and this call must choose the particular edge $(x,y)$ as the new matching edge. Let $\xsettle$ be the random variable which is equal to the number of times we call \GenericRandomSettle{$x,\ell$} between $\sigmapivot$ and $\tstar$ and let $\esettle$ be the event that $\xsettle \leq C' \log^2(n)$ for some large constant $C'$. The crux of the proof is to show that 
	
\begin{equation}
\label{eq:esettle2}
\PP[\esettle] \geq 1 - 1/n^5
\end{equation} 
	
	Before proving this fact, let us show why it allows us to prove the lemma. We have that 
	
	\begin{equation}
	\PP[\elemma_{\tstar}] = \PP[\elemma_{\tstar} \land \esettle] + \PP[\elemma_{\tstar} \land \lnot \esettle] \leq \PP[\elemma_{\tstar} \land \esettle] + \frac{1}{n^5}.
	\end{equation}
	
	We now bound $\PP[\elemma_{\tstar} \land \esettle]$. By Matching Property*, an execution of  \GenericRandomSettle{$x,\ell$} has probability at most $O(\log(n)/4^{\ell})$ of picking edge $(x,y)$. 
	However, we need a bound on 
	the probability of $(x,y)$ being matched \emph{at time $\tstar$}.
	This is where we rely on the independence proved in Corollary \ref{cor:above-independence}: 
	any call to
	\GenericRandomSettle{$x,\ell$} between $\sigmapivot$ and $\tstar$  (a) uses only  bits from $B_{\leq \ell}$ and (b) depends only on which neighbors belong to $\below{\ell}(x)$
	at the time of the call. Note that
	(a) the used bits are fresh and (b) which neighbors belong to $\below{\ell}(x)$ depends only  on $\spast$, i.e., the sequence of hierarchy changes before this call to
	\GenericRandomSettle{$x,\ell$}. 
	{As shown in Lemma~\ref{lem:pivotal-prop} whether or not $\sigmapivot$ is pivotal 
	for time $\tstar$ depends only on the hierarchy at change $\sigmapivot$ (i.e. $\spast$), future above-$\ell$ changes, and $\aa$.} Thus, it is an event that fulfills the
	requirements in Corollary \ref{cor:above-independence}, which 
	shows that $\mate(x)$ is chosen uniformly at random from $\below{\ell}(x)$. 
	
	Now, by definition of $\esettle$ there are $O(\log^2(n))$ executions of \GenericRandomSettle{$x,\ell$} between $\sigmapivot$ and $\tstar$. Each of these picks $(x,y)$ with probability $O(\log(n)/4^{\ell})$. Thus, by a union bound,
	
	\begin{equation}
	\PP[\elemma_{\tstar} \land \esettle] \leq \PP[\elemma_{\tstar} \cond \esettle] = O \left( \log^2(n) \cdot \frac{\log(n)}{4^\ell} \right) = O \left( \frac{\log^3(n)}{4^\ell} \right)
	\end{equation}
	
	Combining the three equations above completes the proof of the lemma. Thus, all that is left to do is to prove Equation \ref{eq:esettle2}.
	
	\paragraph{Proof of Equation \ref{eq:esettle2}}
	
	Recall that $\esettle$ only considers the time period between $\sigmapivot$ and $\tstar$, and that $\sigma'_1, \ldots,  \sigma'_\alpha$ is the sequence of $(\ell,x)$-critical changes in this time period, with $\alpha \leq \alpha_\ell$.
	Consider some call to \GenericRandomSettle{$x,\ell$} during this time period, which results in some edge~$e$ being chosen as the matching edge and $\responsible(x)$ being set to True.
	The only way that \GenericRandomSettle{$x,\ell$} can be called again after this point is that edge $e$ leaves the matching.
	By Lemma~\ref{lem:reasons for becoming unmatched} this can only happen as the consequence of an $(\ell,x)$-critical change.
	Since by the definition of $\sigmapivot$, none of the $(\ell,x)$-critical changes $\sigma'_i$ alter the level of $x$, only $(\ell,x)$-critical changes of types \ref{critical:fall}, \ref{critical:rise}, and \ref{critical:adversary} from Definition~\ref{dfn:critical} are possible in the sequence $\sigma'_1, \ldots,  \sigma'_\alpha$, i.e., the following types of changes:
\begin{itemize}
	\item For some neighbor $w$ of $x$, $\sigma$ changes $\LL(w)$ from $\ell+1$ to $\ell$.
	\item For some neighbor $w$ of $x$, $\sigma$ raises $w$ from a level $\leq \ell$ to a level $> \ell$.
	\item $\sigma$ is an adversarial insertion/deletion of some edge $(x,w)$ incident to $x$ (regardless of level).
\end{itemize}
	We say the corresponding neighbor $ w $ in such a change is the node \emph{affected} by the change.

	In the following we distinguish between calls to \GenericRandomSettle{$x,\ell$} that are the consequence of an $(\ell,x)$-reset opportunity and those that are not and in the latter case we consider affected node to the call to \GenericRandomSettle{$x,\ell$}.

	\emph{Bounding $\xsettle$.} Recall that $\xsettle$ is the random variable that is equal to the number of times the algorithm calls \GenericRandomSettle{$x,\ell$} between $\sigmapivot$ and $\tstar$.
	Consider the sequence of the algorithm's calls to \GenericRandomSettle{$x,\ell$} and \ResetMatching{$x$} between $\sigmapivot$ and $\tstar$.
	Let $\xreset$ be the random variable that is equal to the number of calls to \GenericRandomSettle{$x,\ell$} that are directly preceded by a call to \ResetMatching{$x$} in this sequence and let $ \xcgrs $ be the random variable that is equal to the number of calls to \GenericRandomSettle{$x,\ell$} that are directly preceded by a call to \GenericRandomSettle{$x,\ell$} in this sequence.
	Observe that trivially
	\begin{equation}
		\xsettle \leq \xreset + \xcgrs + 1 \, .\label{eq:bound on xsettle}
	\end{equation}
	(The plus one is necessary because the sequence of calls to \GenericRandomSettle{$x,\ell$} and \ResetMatching{$x$} might start with a call to \GenericRandomSettle{$x,\ell$} for which then no preceding call exists.)

	\emph{Bounding $\xreset$.} We will first bound $\xreset$ by bounding the corresponding number of calls to \ResetMatching{$x$}.
	For each call to \ResetMatching{$x$} preceding a call to \GenericRandomSettle{$x,\ell$} we necessarily have $\LL(x) = \ell$ and therefore there must be some $(\ell,x)$-reset opportunity $\sigma'_i$ causing the call to \ResetMatching{$x$}.
	Each $(\ell,x)$-reset opportunity has a probability of at most $\preset_{\ell} = \resetprob{\ell}$ of causing a \ResetMatching{$x$}.
	By definition of $\sigmapivot$, there are at most $\alpha_\ell = 4000 \cdot \log(n) \cdot 4^\ell$ $(\ell,x)$-critical changes in the sequence $\sigma'_1, \ldots, \sigma'_\alpha$.
	Thus, applying Chernoff bound and a suitable choice of constants, with probability at least $1 - 1/(2n^5)$ we have that $\xreset = O(\log(n))$.

	\emph{Bounding $\xcgrs$.} We will now bound $\xcgrs$.
	In particular, we will show that with high probability $\xcgrs < k := 1024000 C \log^2 (n) $.
	Note that $ \xcgrs $ is equal to the number of pairs of consecutive calls to \GenericRandomSettle{$x,\ell$} between $\sigmapivot$ and $\tstar$ that are not interleaved with calls to \ResetMatching{$x$}
	Consider a pair of consecutive calls to \GenericRandomSettle{$x,\ell$} that are not interleaved with calls to \ResetMatching{$x$}. By Matching Property* we have $ |N_{<\ell} (x)| \geq \frac{4^\ell}{32 C \log(n)}$ during both these calls to \GenericRandomSettle{$x,\ell$}.
	Define $\gamma := \frac{4^\ell}{32 C \log(n)}$ to be this lower bound on the number of choices for the mate of $ x $.
 	Let  $y_1, y_2, \dots $ be the sequence of nodes from $ N_{<\ell} (x) $ appearing as affected nodes in those $(\ell,x)$-critical changes that are not reset opportunities after the first call to \GenericRandomSettle{$x,\ell$} in the order of first appearance.
      We say that the call to \GenericRandomSettle{$x,\ell$} has \emph{large span} if $\mate(x)$ does not occur within the first (up to) $ \tfrac{\gamma}{2} $ elements of this sequence $ y_1, y_2, \dots$. As none of the reset opportunities after the first call is successful (as there is no \ResetMatching{$x$} before the second \GenericRandomSettle{$x,\ell$}), the fact that the first call has a large span implies that there are at least $\gamma/2$  $(\ell,x)$-critical changes that are not reset opportunities between the two calls.

\begin{claim}\label{claim:large span probability lower bound}
Each call to \GenericRandomSettle{$x,\ell$} has large span with probability at least $ \tfrac{1}{2} $ and this bounds holds independently of which earlier calls had large span.
\end{claim}
     \begin{proof}
	Consider a call to \GenericRandomSettle{$x,\ell$} and let $N:= N_{<\ell} (x)$ when this call happens.
	As a result of this call, the mate $\mate(x)$ of $ x $ is chosen uniformly at random from $ N$ with $|N| \ge \gamma$.
	Let $ y_1, \ldots, y_{\gamma'} $ be the set of nodes from $ N_{<\ell} (x) $ appearing as affected nodes in those $(\ell,x)$-critical changes that are not reset opportunities after the call to \GenericRandomSettle{$x,\ell$} in the order of first appearance.

	Note that $ N$ only depends on the sequence of hierarchy changes before the call to \GenericRandomSettle{$x,\ell$} and the sequence $ y_1, \ldots, y_{\gamma'} $ only depends on the sequence of hierarchy changes before the call and the above-$ \ell $ hierarchy changes after the call.
	Therefore, we may apply Corollary~\ref{cor:above-independence} by which the probability that a specific $ y_i $ is equal to $ \mate(x) $ is at most $ \tfrac{1}{\gamma} $.
          Note that this holds for any possible $\spast$, and, thus, independent of $\spast$. Hence it holds in particular no matter which previous calls to \GenericRandomSettle{$x,\ell$} had large span.
	It follows that $ \mate(x) $ occurs within the first (up to) $ \tfrac{\gamma}{2} $ elements of the sequence $ y_1, \ldots, y_{\gamma'} $ with probability at most $ \tfrac{\gamma}{2} \cdot \tfrac{1}{\gamma} = \tfrac{1}{2} $ .
	This means that each call to \GenericRandomSettle{$x,\ell$} has large span with probability at least $ \tfrac{1}{2} $.

	Note that this reasoning also shows that whether a call has large span is independent of whether a \emph{previous} call had large span.
\end{proof}
	
We next bound how often calls with large span can happen for	consecutive calls to \GenericRandomSettle{$x,\ell$}. 
\begin{claim}\label{claim:number of pairs of consecutive calls}
There are at most $ \tfrac{1}{4} k $ pairs of consecutive calls to \GenericRandomSettle{$x,\ell$} between $\sigmapivot$ and $\tstar$ that are not interleaved with any call to \ResetMatching{$x$} such that the first of these calls has large span.
\end{claim}
\begin{proof}
 Consider any pair of consecutive calls to \GenericRandomSettle{$x,\ell$} between $\sigmapivot$ and~$\tstar$ that are not interleaved with any call to \ResetMatching{$x$}.
	If the first of these two calls has large span, then there are at least $ \tfrac{\gamma}{2} = \tfrac{4^\ell}{64 C \log(n)} $ $(\ell,x)$-critical changes that are not reset opportunities between these two calls, as argued above.
	Since the total number of $(\ell,x)$-critical changes that are not reset opportunities between $\sigmapivot$ and~$\tstar$ is at most $\alpha_\ell = 4000 \cdot \log(n) \cdot 4^\ell$, it must be the case that the situation above occurs at most $ \tfrac{\alpha_\ell}{\gamma/2} = 256000 C \log^2(n) = \tfrac{1}{4} k $ times.
\end{proof}

Recall that $\xcgrs$ is the total number of pairs of consecutive calls to \GenericRandomSettle{$x,\ell$} between $\sigmapivot$ and $\tstar$ that are not interleaved with any call to \ResetMatching{$x$}.
	The claim together with the fact that each call to \GenericRandomSettle{$x,\ell$} has large span with probability at least $ \tfrac{1}{2} $ gives us a bound on $\xcgrs$  as follows:
	For each $ i \geq 1 $, define $ Z_i $ as the binary random variable that (1) if $ i \leq \xcgrs $ and in the $i$-th pair of consecutive calls to \GenericRandomSettle{$x,\ell$} the first call to \GenericRandomSettle{$x,\ell$} has large span is $ 1 $, (2) if $ i \leq \xcgrs $ and in the $i$-th pair of consecutive calls to \GenericRandomSettle{$x,\ell$} the first call to \GenericRandomSettle{$x,\ell$} does not have large span is $ 0 $, and (3) if $ i > \xcgrs $ is $ 1 $ with probability $ \tfrac{1}{2} $ and $ 0 $ otherwise.
	Furthermore, define, for each $ i \geq 1 $, $ Z_i^* $ as the binary random variable that is $ 1 $ with probability $ \tfrac{1}{2} $ and $ 0 $ otherwise.
	
	We will now show that $ \Pr [\xcgrs \geq k] \leq \Pr [\sum_{i=1}^{k} Z_i \leq \tfrac{1}{4} k] $ by arguing that	the event $ \xcgrs \geq k $ implies the event $ \sum_{i=1}^k Z_i \leq \tfrac{1}{4} k $.
	Observe that this implication is equivalent to the statement $ \Pr [\sum_{i=1}^k Z_i \leq \tfrac{1}{4} k \cond \xcgrs \geq k] = 1 $.
	Given that $ \xcgrs \geq k $, it follows from the definition above that for each $ i \leq k $, $ Z_i $ is $ 1 $ if in the $i$-th pair of consecutive calls to \GenericRandomSettle{$x,\ell$} the first call to \GenericRandomSettle{$x,\ell$} had large span.
	In other words, conditioned on the event $ \xcgrs \geq k $, $ \sum_{i=1}^k Z_i $ is precisely the number of pairs of consecutive calls to \GenericRandomSettle{$x,\ell$} between $\sigmapivot$ and~$\tstar$ that are not interleaved with any call to \ResetMatching{$x$} and in which the first call to \GenericRandomSettle{$x,\ell$} has large span.
	By Claim~\ref{claim:number of pairs of consecutive calls} we have that this number is at most $ \tfrac{1}{4} k $, i.e., conditioned on $ \xcgrs \geq k $ we have $ \sum_{i=1}^k Z_i \leq \tfrac{1}{4} k $ as desired.
	
	Now by Claim~\ref{claim:large span probability lower bound} each $ Z_i $ with $ i \leq \xcgrs $ is $ 1 $ with probability at least $ \tfrac{1}{2} $ \emph{independent} of the outcomes of the random variables with smaller index, and for $ i > \xcgrs $ this property obviously holds as well.
	More formally, the following holds for all $z_1, \dots z_{i-1} \in \{0,1\} $:
	\begin{equation*}
		\PP [Z_i = 1 \cond Z_1 = z_1, \dots, Z_{i-1} = z_{i-1}] \geq \frac{1}{2} = \PP [Z_i^* = 1] \, .
	\end{equation*}
	Under this precondition, the sum of the $ Z_i $'s stochastically dominates the sum of the $ Z_i^* $'s (see Lemma 1.8.7 in~\cite{doerr}), i.e., $ \PP [\sum_{i=1}^k Z_i \leq \lambda] \leq \PP [\sum_{i=1}^k Z_i^* \leq \lambda] $ for all $ \lambda $.
	Using the shorthand $ \mu := \EX [\sum_{i=1}^k Z_i^*] = \tfrac{1}{2} k $ we now apply a standard Chernoff bound to get the following estimation:
	\begin{align*}
		\PP [\xcgrs \geq k] &\leq \PP \left[ \sum_{i=1}^{k} Z_i \leq \frac{1}{4} k \right] = \PP \left[ \sum_{i=1}^{k} Z_i \leq \frac{1}{2} \mu \right] \\ &\leq \PP \left[ \sum_{i=1}^k Z_i^* \leq \frac{1}{2} \mu \right] \leq \exp \left( - \frac{1}{8} \mu \right) = \exp \left( - \frac{1}{16} k \right) \leq \frac{1}{2 n^5} \, .
	\end{align*}
	Overall, we have thus argued that with probability at least $ 1 - \tfrac{1}{2 n^5} $ both $\xreset$ and $\xcgrs$, and thus $\xsettle$ by~\eqref{eq:bound on xsettle}, are at most $ O (\log^2 n) $.
\end{proof}

\begin{lemma}
	\label{lem:pivotal-exists}
Define $x,y,\ell,B^+,\tstar$ as in the statement of Lemma \ref{lem:prob-matching-stronger}.  Consider some execution of the dynamic matching algorithm with $B^+ = B_{>\ell}$. Then, $\PP[$at some point during the execution the algorithm encounters a pivotal change $\sigmapivot$ for time $\tstar] \geq 1 - 1/n^{10}$, where the probability is over all random bits in $B_{\leq \ell}$.
\end{lemma}

\begin{proof}
Let $\sigma^\ell_1, \ldots \sigma^\ell_{q}$ be all the $(\ell,x)$-critical changes made by the algorithm up to time $\tstar$.
We next define two $(\ell,x)$-critical changes and argue about their relative order.
(a)  If $q \leq \alpha_\ell$, we set $\sigmagap = \sigmadummy$, where $\sigmadummy$ is the empty update defined in Definition \ref{dfn:pivotal}. 
If $q> \alpha_\ell$, 
let $\sigmagap$ be the $(\ell,x)$-critical change $\sigma^\ell_{q- \alpha_\ell}$; note that this is chosen to satisfy Property \ref{pivotal:number} of the definition of a pivotal change (Definition \ref{dfn:pivotal}).
(b) Let $\sigmalevel$ be the last $(\ell,x)$-critical change that changes $\LL(x)$ and $\sigmalevel = \sigmadummy$ if no such change exists. It is clear that $\sigmalevel$ satisfies property \ref{pivotal:level} of the definition of a pivotal change for $\tstar$. Observe that $\sigmalevel$ also satisfies property \ref{pivotal:false} because our algorithm only alters the level of a \emph{free} vertex, so $x$ is free right after $\sigmalevel$, which implies that $\elemma$ is false at that time. (This last argument might feel like cheating, since soon after $\sigmalevel$ the algorithm will assign a new matching edge to $x$; but these future matching edges were already handled in Lemma \ref{lem:pivotal-good}, where we bounded the probability that $\elemma$ becomes true at some point before $\tstar$.) 

We now consider two cases. The simple case is that $\sigmalevel$ comes after, or is the same as, $\sigmagap$. In this case, $\sigmalevel$ satisfies all the properties of a pivotal change for $\tstar$, thus proving the lemma.

The second case is that $\sigmagap$ comes after $\sigmalevel$. In this case $\sigmagap$ satisfies Properties \ref{pivotal:number} and \ref{pivotal:level}, but may fail to satisfy property \ref{pivotal:false}. If $\sigmagap  = \sigmadummy$, then since $\elemma$ is clearly false after $\sigmadummy$ (the graph is still empty), $\sigmadummy$ is itself pivotal for $\tstar$. Also note that $\elemma$ is false before and after $\sigmagap$ if $\LL(x) \neq \ell$ at time $\sigmagap$.

{\bf Assumption:} We can thus assume for the rest of the proof that $\sigmagap \neq \sigmadummy$, that $\LL(x) = \ell$ right before and right after $\sigmagap$ and $x$ remains on this level until $\tstar$.

Note that if there exists \emph{any} change, let's call it $\sigmapivot$, between $\sigmagap$ and $\tstar$ such that $\elemma$ is false after change $\sigmapivot$, then $\sigmapivot$ satisfies all the properties of a pivotal change for $\tstar$, as all changes after change $\sigmagap$
fulfill properties \ref{pivotal:number} and \ref{pivotal:level}
of a pivotal change for $\tstar$. 
Let $\epivot$ be the event that such a $\sigmapivot$ exists; we now show that $\epivot$ is true with high probability. 
The crux of our argument is to show that there are many 
$(\ell,x)$-reset-opportunities between $\sigmagap$ and $\tstar$: each of them performs a \ResetMatching{$x$} only with a small probability, but if there are enough of them at least one of them will indeed perform a \ResetMatching($x$), which will imply that $x$ becomes free, i.e. Property \ref{pivotal:false} holds after this change, i.e., it is a
pivotal change for $\tstar$.
Thus we first need to show the following claim. Recall that $\beta_{\ell} = \alpha_{\ell}/4$. Recall that we are in the case $\sigmagap > \sigmalevel$, which implies that there are exactly $\alpha_\ell$ $(\ell,x)$-critical changes between $\sigmagap$ and $\tstar$. 
\begin{claim}
If $\sigmagap > \sigmalevel$  at least $\beta_\ell$ of the $\alpha_{\ell}$ $(\ell,x)$-critical changes between $\sigmagap$ and $\tstar$ are also $(\ell,x)$-reset-opportunities. 
\end{claim}

\begin{proof}[Proof of Claim] Because $\sigmagap > \sigmalevel$, we know that none of the $(\ell,x)$-critical changes between $\sigmagap$ and $\tstar$ change the level of $x$. Thus, each of these critical changes either: {\bf 1)} adds a vertex $w$ to $N_{= \ell}(x)$ by moving $w$ from level $\ell+1$ to level $\ell$ (item \ref{critical:fall} of Definition \ref{dfn:critical}), or {\bf 2)} removes a vertex $w$ from $N_{\leq \ell}(x)$ (item \ref{critical:rise}), or {\bf 3)} it makes an adversarial update incident to $x$ (item \ref{critical:adversary}). 
Let $t_i$ be the number of changes of each type. Note that $t_1 + t_2 + t_3 = \alpha_\ell$.
Types 1 and 3 are $(\ell,x)$ reset opportunities by definition. Now, note that at the start of $\sigmagap$, since $\LL(x) = \ell$ (see assumption above), Invariant 3 guarantees that $|N_{\leq \ell}(x)|  = \below{\ell+1}(x)\leq  4^{\ell + 1}$.  Moreover, $N_{\leq \ell}$ can only grow as a result of changes of type 1 and 3, and only by 1 after each such change. Thus,  $ t_2 \leq$ $4^{\ell + 1} + t_1 + t_3$. 
Thus, it follows that [\# reset opportunities] =  $ t_1 + t_3 \geq (\alpha_\ell - 4^{\ell + 1})/2 \geq \alpha_\ell / 4 = \beta_\ell$. 
\end{proof}

\paragraph{Back to Proof of Lemma \ref{lem:pivotal-exists}}
Let $\sigma^*_1, \ldots, \sigma^*_{\beta_\ell}$ be $(\ell,x)$-critical changes between $\sigmagap$ and $\tstar$ that are also $(\ell,x)$-reset opportunities; these changes exist by the claim above (there may be more than $\beta_\ell$ such changes, in which case pick any $\beta_\ell$). 
Our goal is to show that $\elemma$ is false after one of these changes whp, which implies that $\epivot$ holds whp.

Let $\ebad_i$ be the event that all of the following hold right before change $\sigma^*_i$: $\responsible(x)$ is True and $\LL(x) = \ell$ and the change $\sigma^*_i$ does not lead to \ResetMatching{$x$}. 
Recall that by our assumption we know that $\LL(x) = \ell$  for each of the $\sigma^*_i$ that we are considering.
We need to show that with high probability  after at least one of the changes $\sigma^*_i$ 
$\elemma$ is false, i.e.  $\responsible(x) =$ False or 
\ResetMatching($x$) is executed in the change.
Event $\epivot$ is false if no such change exists. Thus $\epivot$ can only be false if \emph{all} the $\ebad_i$ are true. We thus have

$$\PP[\lnot \epivot] \leq \PP[\ebad_1 \land \ldots \land \ebad_{\beta_\ell}] \leq \ebad_1 \cdot \prod_{i=2}^{\beta_\ell} \PP[\ebad_i \cond \ebad_1 \land \ldots \land \ebad_{i-1}]$$

We now observe that $\PP[\ebad_1] \leq (1 - \preset_\ell) = (1-1/4^{\ell + 3})$ and $\PP[\ebad_i \cond \ebad_1 \land \ldots \land \ebad_{i-1}] \leq (1 - \preset_\ell) = (1-1/4^{\ell + 3})$. The reason is simply that each $\sigma^*_i$ is a $(\ell,x)$ reset-opportunity, so either $\responsible(x)$ is False at change $\sigma^*_i$, in which case $\ebad_i$ is false by definition, or otherwise by Observation \ref{obs:reset} the algorithm performs \ResetMatching$(x)$ with probability $\preset_\ell$. Moreover, the probability of \ResetMatching{$x$} occurring is using a fresh random bit and, thus, is clearly independent of everything that came before. 

Putting everything together, we have that 
$$\PP[\lnot \epivot] \leq (1 - \frac{1}{4^{\ell + 3}})^{\beta_\ell} = (1 - \frac{1}{64 \cdot 4^{\ell}})^{1000\cdot\log(n)\cdot4^\ell} \leq \frac{1}{n^{10}}.$$
Note that this crucially relies on Corollary \ref{cor:above-independence} to ensure that the probability of \ResetMatching{$x$} is not correlated with any of the future $(\ell,x)$-critical changes.	
\end{proof}

\begin{proof}[Proof of Lemma \ref{lem:prob-matching-stronger}]
The proof follows immediately from the combination of Lemmas \ref{lem:pivotal-exists} and \ref{lem:pivotal-good}
\end{proof}

\subsubsection{Putting Everything Together}

We are now ready to prove that the algorithm processes \emph{every} adversarial update in expected time $O(\log^4(n))$.

\begin{proof}[Proof of Theorem~\ref{thm:matching-expected}]
Let us first consider the insertion of a new edge $(u,v)$. The algorithm has to update some subset of $\OO_u, \OO_v, \EE_u^{\LL(v)}, \EE_v^{\LL(u)}$ in $O(1)$ time,
and then it has to perform $O(\log(n))$ calls to \IncrementPhi{$u,i$} and \IncrementPhi{$v,i$} (line \ref{step:insert-rise} or Algorithm \ref{alg:matching 1}). Each \IncrementPhi{$v,i$} has a 
$\prise = \Theta(\log(n)/4^i)$ chance of leading to a probability-rise, and a negligible probability of leading to a threshold-rise 
(Lemma~\ref{lem:threshold-rise}), so plugging in the cost of a call to \Rise from Lemma~\ref{lem:rise}, 
the expected total time to process the rises is 
$O(\sum_{i=0}^{\floor{\log_4(n)}} (\log(n)/4^i) \cdot 4^i = O(\log^2(n))$.
\Insert($x,y$) also causes \ResetMatching{$u$} and \ResetMatching{$v$} with Probabilities $\preset_{\LL(u)} = \resetprob{\LL(u)}$ and $\preset_{\LL(v)} = \resetprob{\LL(v)}$. By lemma \ref{lem:emax} the expected time to process the resets is $O(4^{\LL(u)})$ and $O(4^{\LL(v)})$; multiplying by the reset probabilities yields an expected time of $O(1)$.

Let us now consider the deletion of an edge $(u,v)$. If $(u,v)$ is a non-matching edge, then as in the case of insertion, the algorithm
only needs to perform $O(1)$ bookkeeping work and $O(\log(n))$ calls to \DecrementPhi{$u,i$}; calls to \DecrementPhi do not
lead to changes in the hierarchy, so the algorithm stops there. Thus the only case left to consider is the deletion of a matching edge $(u,v)$.
By Invariant~4 $\LL(u) = \LL(v)$, so let us say they are both equal to $\ell$. The deletion of $(u,v)$ requires the algorithm to execute 
\FixFreeVertex{$u$} and \FixFreeVertex{$v$}, which by Corollary~\ref{cor:emax} requires time $O(4^\ell)$. By Corollary~\ref{cor:prob-matching},
the total expected update time is thus 
$O(\sum_{\ell=0}^{\floor{\log_4(n)}} 4^\ell \cdot (\log^3(n)/4^\ell)) = O(\log^4(n))$.
\end{proof}

\paragraph{Explicitly Maintaining a List of the Edges in the Matching.}

Both the amortized expected algorithm of Baswana et al. \cite{BaswanaGS18}, as well as our worst-case expected modification in 
Theorem~\ref{thm:matching-expected}, store the matching in the simplest possible data structure $D$: they are both able to
maintain a single list containing all the edges of a maximal matching. By Remark~\ref{remark:main}, the high-probability worst-case result in~\ref{thm:matching-converted} stores the matching in a slightly different data structure: it
stores $O(\log(n))$ lists $D_i$, along with a pointer to some $D_j$ such that $D_j$ is guaranteed to contain the edges of a maximal matching. Dynamic algorithms are typically judged by update and query time, and from this perspective 
our data structure is equivalently powerful, since we can use the correct $D_j$ to answer queries about the matching.

However, in some applications, it is desirable to maintain the matching as a single list. The reason is that this way one
ensures ``continuity'' between the updates: for example, the $O(\log(n))$ update time of Baswana \etal\ guarantees
that every update only changes the underlying maximal matching by $O(\log(n))$ edges (amortized).
This is no longer true of our high-probability worst-case algorithm in Theorem~\ref{thm:matching-converted}, because
a single update might cause the algorithm to switch the pointer from some $D_i$ to some $D_j$: this still results in a fast
update time, but the underlying maximal matching can change by $\Theta(n)$ edges. 

As discussed in Remark~\ref{remark:main}, if we insist on maintaining a single list of edges in the matching, we can do so with almost the same high-probability worst-case update time as stated in Theorem~\ref{thm:matching-converted}, but the resulting matching is only $(2+\eps)$-approximate, and no longer maximal.
This $(2+\eps)$-approximation is achieved as follows.
The algorithm of Theorem~\ref{thm:matching-converted} stores $O(\log(n))$ lists $D_i$, one of which is guaranteed to be a maximal matching.
In particular, this algorithm maintains a fully dynamic data structure with query access to a $2$-approximate matching that can output $ \ell $ arbitrary edges of the matching in time $ O (\ell) $.
The very recent black-box reduction in \cite{SolomonS18} takes such a ``discontinuous'' algorithm for dynamic maximum matching and turns it into a ``continuous'' one at the cost of an extra $(1+\eps)$ factor in the approximation.
By applying this reduction with $ \epsilon' = \epsilon/2 $ we obtain a fully dynamic algorithm for maintaining a matching with an approximation factor of $ 2 (1+\eps/2) = (2+\eps)$ and a high-probability worst-case update time of $O(\log^6(n) + 1/\epsilon)$.
The reduction of~\cite{SolomonS18} also applies to the dynamic $ (2 + \epsilon) $-approximate matching algorithms of Arar et al.~\cite{ArarCCSW18} and Charikar and Solomon~\cite{CharikarS18}, whose update times we can beat for certain regimes of $ \epsilon $.

\section{Conclusion}

In this article, we have provided a meta-algorithm that converts dynamic algorithms with a bound on the worst-case expected update time
 into ones with a high-probability bound on the worst-case time at the expense of logarithmic factors in the update time.
We have then applied this reduction to two graph problems: dynamic spanner and dynamic maximal matching.
Our main observation for these two problems was that certain deterministic amortization techniques in the known algorithms can be replaced by randomized ones to obtain a worst-case expected bound instead of an amortized one.
We conjecture that this approach also works for other graph problems.
Additionally it would be interesting to have more meta-theorems for dynamic graph algorithms in addition to sparsification~\cite{EppsteinGIN97} and the expected to high-probability conversion presented here.


\printbibliography[heading=bibintoc] 

\end{document}